\newif\ifDRAFT 
\DRAFTfalse

\documentclass[11pt]{article}
\usepackage[fleqn]{amsmath}
\usepackage[margin=1in]{geometry}
\usepackage{graphicx} %
\usepackage{amsfonts, amsmath, amssymb, amsthm}

\usepackage{newtxtext}
\usepackage{enumitem}
\setlist[itemize]{itemsep=0pt, topsep=2pt}
\setlist[enumerate]{itemsep=0pt, topsep=2pt}
\setlist[description]{itemsep=0pt, topsep=2pt}

\usepackage{mathrsfs}  
\usepackage{algorithm}
\usepackage[noend]{algpseudocode}
\usepackage{hyperref}
\usepackage{color}
\usepackage{array}
\usepackage{cleveref}
\usepackage{multirow}
\usepackage[table,xcdraw]{xcolor}
\definecolor{DarkRed}{rgb}{0.5,0.1,0.1}
\definecolor{DarkBlue}{rgb}{0.1,0.1,0.5}
\definecolor{ForestGreen}{rgb}{0.1333,0.5451,0.1333}
\hypersetup{colorlinks=true, linkcolor=DarkRed,citecolor=ForestGreen}
\usepackage[inkscapeformat=pdf]{svg}
\usepackage{multirow}
\usepackage{mathtools}
\usepackage{tabulary} %
\usepackage{booktabs}
\usepackage{arydshln}
\usepackage{pdfpages}
\usepackage{booktabs}
\usepackage[table,xcdraw]{xcolor}
\usepackage{arydshln}
\usepackage[x11names, table]{xcolor}
\usepackage{enumitem}
\usepackage{comment}

\theoremstyle{plain}
\newtheorem{theorem}{Theorem}[section]
\newtheorem{lemma}[theorem]{Lemma}
\newtheorem{corollary}[theorem]{Corollary}
\newtheorem{observation}[theorem]{Observation}
\newtheorem{claim}[theorem]{Claim}

\theoremstyle{definition}
\newtheorem{definition}[theorem]{Definition}
\newtheorem*{remark}{Remark}%

\usepackage{thm-restate}

\DeclareSymbolFont{bbold}{U}{bbold}{m}{n}
\DeclareSymbolFontAlphabet{\mathbbold}{bbold}

\newcommand{\cD}{\mathcal{D}}

\newcommand{\cT}{\mathcal{T}}

\newcommand{\cH}{\mathcal{H}}

\newcommand{\cS}{\mathcal{S}}

\newcommand{\Near}{\mathsf{Near}}
\newcommand{\Far}{\mathsf{Far}}

\newcommand{\LCA}{\mathrm{LCA}}

\newcommand{\RepInf}{\textsc{rep}}

\newcommand{\cDdown}{\cD_{\downarrow}}
\newcommand{\cDup}{\cD^{\uparrow}}
\newcommand{\cTdown}{\cT_{\downarrow}}
\newcommand{\cTup}{\cT^{\uparrow}}
\newcommand{\cSdown}{\cS_{\downarrow}}
\newcommand{\cSup}{\cS^{\uparrow}}

\ifDRAFT
    \usepackage{lineno}
    \linenumbers
\fi

\title{Space-Optimal Sensitivity Oracles for Single-Source Mincuts}

\ifDRAFT
    \author{Anonymous}
\else
    \author{Koustav Bhanja \thanks{Weizmann Institute of Science, Israel. Email: \texttt{koustav.bhanja@weizmann.ac.il}. Supported by the Koshland Prize fellowship, and Merav Parter’s European Research Council (ERC) grant under the European Union’s Horizon 2020 research and innovation programme, grant agreement
No. 949083.}
    \and Merav Parter \thanks{Weizmann Institute of Science, Israel. Email: \texttt{merav.parter@weizmann.ac.il}. Supported partially by the European
Research Council (ERC) under the European Union’s Horizon 2020 research and innovation programme, grant agreement
No. 949083.}
    \and Asaf Petruschka \thanks{Weizmann Institute of Science, Israel. Email: \texttt{asaf.petruschka@weizmann.ac.il}. Supported by an Azrieli Foundation fellowship,
and by Merav Parter’s European Research Council (ERC) grant under the European Union’s Horizon 2020 research
and innovation programme, grant agreement No. 949083.}
    }
\fi

\date{}

\begin{document}

\maketitle

\pagenumbering{roman}
\begin{abstract}
How much information is needed to describe the effect of any single edge failure on all
mincuts?  
Classical mincut representations, such as Gomory-Hu trees and
cactus-like structures, show that mincuts often admit unexpectedly compact static
certificates. However, much less is known about whether such compactness survives failures.

We study this question in the single-source setting. Given an undirected multi-graph $G$
and a source $s$, we preprocess $G$ into an oracle that, for any failed edge $e$, reports
all vertices $t$ such that the $s,t$-mincut value decreases in $G\setminus\{e\}$.
Equivalently, the oracle reports all destinations $t$ for which $e$ belongs to some
minimum $s,t$-cut in $G$. 
This problem is the natural single-source analogue of both the fixed-pair mincut sensitivity result of Picard and Queyranne [MPS’80] and the all-pairs sensitivity result of Baswana and Pandey [SODA’22]. It also serves as the failure counterpart to the arguably simpler single-source insertion oracle of Baswana, Gupta, and Knollmann [Algorithmica’22]. To this date, no subquadratic space solution is known for this single-source setting. We show:

\begin{itemize}
    \item Our main result is a space-optimal single-source mincut sensitivity oracle for edge
failures: it uses $O(n)$ space and reports all affected vertices in $O(n)$ time,
matching both the best known bounds for the insertion case and the unavoidable worst-case linear output-size
barrier. 
\item We also give an $O(n^{3/2})$ space oracle for the corresponding membership query:
given a target $t$ and a failed edge $e$, it decides in $O(\log n)$ time whether the
$s,t$-mincut value drops in $G\setminus\{e\}$, equivalently whether $e$ belongs to some
minimum $s,t$-cut in $G$. No subquadratic-space oracle with such query guarantees was
previously known, even for insertions. At the heart of this oracle is a reachability data
structure for farthest mincut DAGs (directed acyclic graphs) of Baswana, Gupta, and Knollmann. We show that these DAGs have a special structure that
bypasses the conjectured linear-query barrier for subquadratic-space reachability oracles
on general DAGs.
\end{itemize}

Our main technical contribution is in establishing novel and intricate connections between two seemingly distant objects,
representing two different families of mincuts.
The first is the DAG representation of farthest mincuts to the source, which was the central tool for the insertion oracle. %
The second is the \emph{Connectivity Carcass} for Steiner mincuts of Dinitz and Vainshtein \mbox{[STOC'94]}, which generalizes well-known cactus representations of global mincuts.
Our work demonstrates the relatively unexplored potential of the carcass beyond its ``obvious'' Steiner mincuts scope.

\end{abstract}
\newpage
{\small \setlength{\parskip}{0em}
\tableofcontents{}
}
\newpage
\pagenumbering{arabic}

\newpage
\section{Introduction}

Let $G=(V,E)$ be an undirected multi-graph on $n=|V|$ vertices and $m=|V|$ edges. Minimum cuts, or \emph{mincuts} for short, are among the most ubiquitous concepts in graph theory and algorithms, and are defined as follows.
A mincut between a pair of vertices $u,v$, denoted by \textit{$u,v$-mincut}, is the set of edges of the least cardinality whose removal disconnects $u$ from $v$. Designing compact data structures for various mincuts has been a major field of research \cite{DBLP:journals/mp/PicardQ80, DBLP:journals/jacm/Karger00, DBLP:conf/stoc/DinitzV94, DinitzN95, DBLP:journals/jacm/KawarabayashiT19, DBLP:conf/soda/AbboudKT22, DBLP:conf/stoc/YKKrauthgamer26}. Classical examples include the cactus graphs for all global mincuts \cite{DinicKL76}, the DAG structure for all $s,t$-mincuts with fixed $s,t\in V$ \cite{DBLP:journals/mp/PicardQ80}, and the celebrated Gomory-Hu tree \cite{GomoryH61}: a compact tree structure of only $O(n)$ space, that represents at least one $u,v$-mincut for every pair of vertices $u,v\in V$. However, in many real-world settings, a basic static network undergoes small and commonly transient changes, such as temporary link failures. This naturally leads to the following central question, which is the main focus of our work.\\
\textit{\text{~~~~~~~~~~} How much information is needed to describe the effect of any edge failure on all
mincuts?}\\ This question motivates the vast practical and theoretical study of \emph{sensitivity oracles}: compact data structures that can effectively report the new solution for a given graph problem once such a change occurs.  
Sensitivity oracles have been designed extensively for many fundamental graph problems, namely, shortest paths  \cite{DBLP:journals/talg/BaswanaCHR20, bilo2023approximate, DBLP:conf/focs/Gupta25}, reachability  \cite{DBLP:conf/soda/ItalianoKP21, DBLP:conf/icalp/Choudhary16}, distances \cite{DBLP:conf/icalp/GuptaK17, DBLP:conf/stoc/Dey024}. For mincuts, however, the area is comparatively much less developed despite its fundamental importance. %

Sensitivity oracles for mincuts were first introduced by Picard and Queyranne in 1980 \cite{DBLP:journals/mp/PicardQ80}. After more than four decades, the area has recently emerged as an active line of research %
\cite{DBLP:conf/soda/BaswanaP22, DBLP:journals/talg/BaswanaBP23, DBLP:conf/icalp/BaswanaB24,  DBLP:conf/esa/BaswanaBR25, DBLP:conf/icalp/Bhanja25, DBLP:conf/innovations/AhiCPPS26}; here we only focus on results pertaining to undirected multi-graphs, which are most relevant to our context. The original setting considered by Picard and Queyranne was of a \emph{fixed pair} of vertices $s,t$, and showed an $O(n)$ space structure that can determine whether the $s,t$-mincut value changes upon the failure (i.e., the deletion) or the insertion of any given edge.
Recently, Bhanja~\cite{DBLP:conf/icalp/Bhanja25} has shown that the oracle of~\cite{DBLP:journals/mp/PicardQ80} is tight in a very strong sense: an oracle that can support queries of two (or more) failing edges already requires (information-theoretically) $\Omega(n^2)$ space.
However, there remains a different intriguing avenue for improvement: going beyond a fixed ``source-destination'' pair of vertices.

The first result in this direction was by Baswana, Gupta, and Knollmann~\cite{BaswanaGK22}.
They showed an $O(n)$ space oracle for edge insertions in the \emph{single-source} setting:
Given any inserted edge, it can determine \emph{all of the vertices} whose mincut value to the distinguished source $s$ is affected by the insertion.
Their central structural insight is an $O(n)$ space representation for the family of \emph{farthest} mincuts to $s$, called the \emph{Farthest Mincut DAG}; %
we briefly overview it in~\Cref{sec:technical-overview}.
However, their work did not address the arguably more important edge failure queries, %
which are considerably harder. 
Essentially, insertion queries require representing only at most $2(n-1)$  mincuts to the source: the nearest and the farthest mincut from each vertex. In contrast, failure queries must handle all mincuts to the source: for any vertex $u\in V$, the $u,s$-mincut value reduces upon failure of an edge $e$ if and only if $e$ belongs to \textit{some} $u,s$-mincut. Such mincuts can be exponential in number; furthermore, unlike nearest/farthest mincuts, no analogous subquadratic space structure is known that supports efficient handling of these mincuts. 
The ``brute-force'' approach to handle failure queries in the single-source setting is to store $n-1$ oracles of~\cite{DBLP:journals/mp/PicardQ80}, one for every pair
in $\{s\} \times V\setminus\{s\}$. A more sophisticated oracle of $O(n^2)$ space, which in fact works for the most general \emph{all-pairs} setting, was provided by~\cite{DBLP:conf/soda/BaswanaP22}, but this does not improve over the brute-force for the single-source case\footnote{Designing an oracle for single-source mincuts with sublinear ($o(m)$) space and query time for handling failure queries was also posed as an interesting problem by Keerti Choudhary at the Frontiers of Graph Algorithms Workshop 2025 \cite{ChoudharyFoGA2025}.}.

\medskip
\noindent\textbf{Space Optimal Oracle.}
Our main result closes this gap between the $O(n)$ space insertion oracle and the previously best $O(n^2)$ space failure oracles by presenting the following compact single-source mincut sensitivity oracles for edge failure queries of just $O(n)$ space: %
\begin{theorem}\label{thm:main}
    For any undirected multi-graph $G = (V,E)$ on $n$ vertices, with a designated source vertex $s \in V$, there is an $O(n)$ space data structure that, given a query edge $e \in E$, can report all those vertices whose mincut value to $s$ decreases upon failure of $e$ within $O(n)$ time.
\end{theorem}
Both the space and query time in~\Cref{thm:main} precisely match those of~\cite{BaswanaGK22} for edge insertion queries.
The query time is ``worst-case optimal'', as there exist graphs (say a cycle) where failure of an edge can affect the mincuts to $s$ of $\Omega(n)$ vertices.
However, one can hope for a better \emph{output-sensitive} query time, as we will discuss shortly. 
The space is optimal even for the fixed-pair setting, and even in simple graphs: this easy information-theoretic space lower bound appears in~\Cref{sec:lowerbound}.

Our oracle is obtained via intricate usage of the \emph{Connectivity Carcass}~\cite{DBLP:conf/stoc/DinitzV94,DBLP:conf/soda/DinitzV95,DBLP:journals/siamcomp/DinitzV00,BaswanaP25}.
The carcass, first announced by Dinitz and Vainshtein in 1994, is a data structure that compactly represents \emph{all Steiner mincuts} of a given terminal set $S \subseteq V$ ---
a Steiner mincut is one with minimum value among those cuts that separate at least one pair of terminals.
However, this description barely conveys its complex interface, which unfortunately cannot be compressed into one short ``blackbox'' result.
We do our best to provide a clean and succinct interface in~\Cref{sec:carcass-prelim}, but it is nevertheless quite dense.
Moreover, the inner workings of the carcass are also highly complicated.
Only very recently, more than two decades after its original announcement, Baswana and Pandey~\cite{BaswanaP25} provided the first simple and comprehensive treatment of this structure\footnote{The odd case was fully treated in \cite{DBLP:journals/siamcomp/DinitzV00}, whereas the even case was in M.S. thesis \cite{Belous2014}, but has remained unpublished.}.
Perhaps for these reasons, the full potential of the carcass for algorithmic applications seems relatively unexplored.
The few existing applications~\cite{DBLP:conf/soda/DinitzV95,DBLP:conf/soda/BaswanaP22,DBLP:conf/icalp/Bhanja25} employ it in contexts where Steiner mincuts are either the main focus themselves, or of rather immediate relevance.

In our view, one of the central \emph{conceptual} contributions of this paper is in demonstrating how the connectivity carcass can be effectively and neatly utilized for a problem which, on the face of it, seems to have little to do with terminal sets and Steiner mincuts.
Surprisingly, the key to our solution is treating the vertices inside each node in the farthest mincut DAG as a terminal set.
In each such node, we embed carefully chosen components of the corresponding carcass.
These are utilized through a machinery that ``translates'' between the global view of mincuts to the source $s$, through the DAG, and into the local views of the carcasses inside the nodes.
This high-level intuition is explained in detail in the Technical Overview (\Cref{sec:technical-overview}).

\medskip\noindent
\textbf{Near-Optimal Query Time.}
We complement our  optimal space-oracle of~\Cref{thm:main} with a larger but still subquadratic $O(n^{1.5})$ space oracle, whose query time guarantees are optimal up to logarithmic factors:

\begin{theorem}\label{thm:faster-queries}
    For any undirected multi-graph $G = (V,E)$ of $n$ vertices,  with a designated source vertex $s \in V$, there is an $O(n^{1.5})$ space data structure supporting the following queries:
    \begin{enumerate}
     \item         \label{item:faster-queries-one-vertex}
        \textbf{(One-Destination)}
        Given a query of edge $e \in E$ (resp., $e \in V\times V$) and destination vertex $u \in V \setminus \{s\}$, it can report whether the $u,s$-mincut value changes upon failure (resp., insertion) of $e$, within $O(\log{n})$  time (resp., $O(1)$ time).
     
        \item         \label{item:faster-queries-reporting-all}
        \textbf{(Output-Sensitive)}
        Given a query of edge $e \in E$ (resp., $e \in V \times V$), it can report all those vertices whose mincut value to $s$ changes upon failure (resp., insertion) of $e$, in amortized $O(\log^2 n)$ time (resp., $O(1)$ time) per reported vertex.

    \end{enumerate}
\end{theorem}

Near-optimal query complexities as above were not previously achieved, even for just edge insertion queries, by any oracle with $o(n^2)$ space.
It turns out that the key for these lies in supporting fast \emph{reachability queries} in the farthest mincut DAG of~\cite{BaswanaGK22}.
In general, answering reachability queries in a DAG is believed to be inherently hard: the \emph{Reachability Conjecture}~\cite{DBLP:journals/siamcomp/Patrascu11, DBLP:conf/wads/GoldsteinKLP17} states that reachability oracles for DAGs on $n$ vertices and $m$ edges require $\tilde{\Omega}(n^2)$ space to achieve $o(m)$ time queries, in the worst case.

Interestingly, the nice structural properties of farthest mincuts are ``inherited'' by their representing DAG, endowing it with a special structure.
We show that this can be exploited to bypass the conjectured worst-case hardness, allowing us to design an 
$O(n^{1.5})$ space oracle for constant time reachability queries.
Combining this main ingredient with some more technical embellishments to the oracle of~\Cref{thm:main} results in the above~\Cref{thm:faster-queries}.
Moreover, the reachability oracle is the only ``heavy'' component: any improvement to its space would immediately yield the same improvement%
\footnote{
More accurately, the improvement would be to~\Cref{item:faster-queries-reporting-all} of~\Cref{thm:faster-queries}; for~\Cref{item:faster-queries-one-vertex} we also need fast \emph{LCA queries} on the DAG. %
}
to~\Cref{thm:faster-queries}.
However, our approach seems hard to improve beyond $O(n^{1.5})$ space; it would be very interesting if one could find a more compact oracle for reachability in the farthest mincut DAG, or provide evidence for hardness.

As an almost-immediate corollary of~\Cref{thm:faster-queries}(\ref{item:faster-queries-reporting-all}), we get the following result in \emph{all-pairs} setting\footnote{A trivial extension of \Cref{thm:faster-queries}(2) incurs an additional $O(n)$ time to report all affected pairs from $V\times V$. }:

\begin{theorem}\label{thm:all-pairs-reporting-all}
    For any undirected multi-graph $G = (V,E)$ of $n$ vertices and $m$ edges, there is an $O(n^{2.5})$ space data structure that, given a query edge $e \in E$, can report all those pairs of vertices between which the mincut value decreases upon failure of $e$, in amortized $O(\log^2 n)$ time per reported pair.
\end{theorem}

We note that~\cite[Theorem 4.1(2)]{DBLP:conf/soda/BaswanaP22}
stated an $O(n^2)$ space oracle that reports all affected pairs by an edge failure in amortized $O(1)$ time per pair.
However, this specific statement seems inaccurate 
as their query time incurs an additive $O(n)$ term (private communication with S. Baswana, 2026).
Thus, our \Cref{thm:faster-queries}(\ref{item:faster-queries-reporting-all}) is the first to achieve near-optimal output-sensitive query time \emph{for edge failure queries} in the all-pairs setting; for edge insertion queries,~\cite{BaswanaGK22} gave $O(n^2)$ space with optimal output-sensitive time.
\medskip\noindent
\textbf{Related Works.} 
We give a broader view of sensitivity oracles through the lens of closely related and central topics of dynamic algorithms and sparsifiers for mincuts. The dynamic algorithms aim to maintain information of mincuts where the graph changes over time due to edge updates (insertion/deletion).  There exist many efficient partial/fully dynamic algorithms for fixed-pair mincuts \cite{DBLP:conf/icalp/GoranciHRS25, DBLP:conf/icalp/GoranciH23, HenzingerK99} and global mincuts \cite{Thorup07, DBLP:conf/soda/GoranciHNSTW23, DBLP:conf/soda/JinST24}. However,
the first and only nontrivial algorithm for dynamically maintaining a Gomory-Hu tree (or all-pairs mincuts) was established very recently by Kenneth-Mordoch and Krauthgamer \cite{DBLP:conf/stoc/YKKrauthgamer26}. They design an algorithm that uses $O(n^{1.5})$ space and achieves $O(n^{1.5})$ update time. %
Their techniques are mostly based on the \emph{friendly cut sparsifier} of \cite{DBLP:conf/soda/AbboudKT22}, which is inherently restricted to unweighted \emph{simple} graphs (without parallel edges). Also, there exist graph instances for which any such friendly cut sparsifier must occupy $\Omega(n^{1.5})$ space. It is now interesting to explore whether this bottleneck can be broken to achieve better bounds, where our approach might turn out to be useful.

Sparsification, in general, is a powerful tool that helps us preserve properties of a graph in a relatively much \textit{smaller} graph. Several classical sparsifiers for mincuts (or cuts in general) have been developed that sparsify the graph while the relevant cuts can be preserved approximately \cite{BenczurK96, DBLP:journals/jacm/KawarabayashiT19, FungHHP11} or exactly \cite{NagamochiI92, DBLP:conf/soda/AbboudKT22, DinicKL76, GomoryH61}. The space occupied by the exact mincut sparsifiers preserving {\em all} mincuts between {\em every} vertex pair \cite{NagamochiI92, DBLP:conf/stoc/YKKrauthgamer26} 
is parameterized on the value $\omega$ 
which is the maximum value of any $s,t$-mincut over all $s,t \in V$,
and they can achieve $O(n)$ space only when $\omega$ is constant.

\medskip\noindent
\textbf{Future Directions.}
A natural next step is to extend our results to the \emph{all-pairs setting}. An intriguing open problem is whether one can obtain an $O(n)$ space sensitivity oracle for all-pairs mincuts.
Our results for the single-source setting seem to provide evidence in this direction, especially in light of recent efficient algorithmic reductions from all-pairs to single-source mincuts, e.g., in~\cite{DBLP:conf/focs/AbboudK0PST22,DBLP:conf/focs/Abboud0PS23,DBLP:conf/focs/AbboudKLPGSYY25}. 

Another direction is \emph{weighted edge failures}. 
In a sense, our oracles do work for weighted graphs:
\Cref{thm:main} gives an $O(n)$ space data structure that, given an edge $e \in E$, can determine all the vertices $u$ such that $e$ contributes to some $u,s$-mincut (and \Cref{thm:faster-queries} works similarly).
So, in the weighted setting, our ``edge failure'' query translates to determining whether the $u,s$-mincut value changes when the capacity (i.e., weight) of the edge is reduced by some $\Delta \in (0,1]$ (assuming, by scaling, that weights are integral).
Extending our oracles to handle a weighted edge failure which reduces its capacity to zero (or more generally, by an arbitrary amount) remains an interesting challenge. 
This is particularly intriguing as Baswana and Bhanja \cite{DBLP:conf/icalp/BaswanaB24} recently extended the fixed-pair oracle of~\cite{DBLP:journals/mp/PicardQ80} to this weighted setting.

\section{Preliminaries} \label{sec: peliminaries}

Throughout, we consider an unweighted, undirected, connected multi-graph $G = (V,E)$ with $n = |V|$ vertices and $m = |E|$ edges, which has a designated source vertex $s \in V$. 

For $A \subseteq V$, the \emph{capacity} $c(A)$ 
is the number of edges crossing from $A$ to $\overline{A} := V\setminus A$.
Those edges are called the \emph{edge-set of the cut} $A$, or the \emph{contributing edges} of $A$.
If $\emptyset \neq X \subseteq A$ and $\emptyset \neq Y \subseteq \overline{A}$, we say that $A$ is an \emph{$(X,Y)$-cut.}
An $(X,Y)$-cut of smallest capacity is called \emph{$(X,Y)$-mincut}.
The term \emph{$X,Y$-(min)cut} refers to either an $(X,Y)$-(min)cut or a $(Y,X)$-(min)cut.
When $X,Y$ are singletons $X=\{x\}$ and $Y=\{y\}$, we may write $x,y$-(min)cut, and such shorthands for singletons apply to other notations as well.
For $S \subseteq V$, the (cut defined by) $A$ is said to be an \emph{$S$-cut} or a \emph{Steiner cut for $S$} if $S \cap A \neq \emptyset$ and $S \cap \overline{A} \neq \emptyset$. 
An $S$-cut of smallest capacity is called \emph{$S$-mincut} or a \emph{Steiner mincut for $S$}, and its capacity is denoted by $\lambda_S$.

It is well-known that $c(\cdot)$ satisfies \emph{submodularity}:
$c(A) + c(B) \geq c(A \cap B) + c(A \cup B)$ for all $A,B \subseteq V$.
As $G$ is undirected, this  also implies \emph{posimodularity}: $c(A) + c(B) \geq c(A \setminus B) + c(B \setminus A)$ for all $A,B \subseteq V$.
Rather than using these inequalities directly, we will use their consequence on mincuts:
\begin{lemma}\label{lem:sub-posi-general}
    Let $S,T \subseteq V$, and suppose $A$ is an $S$-mincut and $B$ is a $T$-mincut.
    \begin{itemize}
        \item (``Submodularity'') If one of $A \cap B$ and $A \cup B$ is an $S$-cut and the other is a $T$-cut, then the former is an $S$-mincut, the latter is a $T$-mincut, and 
        no edge of $G$ goes between $A \setminus B$ and $B \setminus A$.
        
        \item (``Posimodularity'') If one of $A \setminus B$ and $B \setminus A$ is an $S$-cut and the other is a $T$-cut, then the former is an $S$-mincut, the latter is a $T$-mincut, and 
        no edge of $G$ goes between $A \cap B$ and $\overline{A \cup B}$.
    \end{itemize}
    Furthermore, the 
    above also holds if we replace $S$-(min)cut by $S_1,S_2$-(min)cut, and/or replace $T$-(min)cut by $T_1,T_2$-(min)cut, for  $S_1, S_2, T_1, T_2 \subseteq V$. 
\end{lemma}
\def\SUBANDPOSI{
\begin{proof}[Proof of~\Cref{lem:sub-posi-general}]
    We prove the statement with $S$ and $T$, and the ``furthermore'' part is proved in exactly the same way.
    
    (Submodularity) 
    Let $C$ denote the $S$-cut among $A \cap B$ or $A \cup B$, and $D$ denote the other one of them, which is a $T$-cut. 
    As $A$ is an $S$-mincut and $B$ is a $T$-mincut, we get the inequalities $c(A) \leq c(C)$ and $c(B) \leq c(D)$.
    But by submodularity $c(A) + c(B) \geq c(C) + c(D)$, so these inequalities must in fact hold with equality.
    Now, suppose towards contradiction that $e$ is an edge between $A \setminus B$ and $B \setminus A$.
    Consider the graph $G'$ obtained by removing $e$ from $G$, and let $c'(\cdot)$ denote capacities in $G'$.
    As $e$ contributed to $A$, it remains an $S$-mincut in $G'$.
    Similarly, $B$ remains a $T$-mincut in $G'$.
    Thus, by the first part of the proof, $C$ is an $S$-mincut also in $G'$.
    However, $e$ does not contribute to $C$, so we get
    $
    0 = c(C) - c'(C) = c(A) - c'(A) = 1,
    $
    a contradiction.

    (Posimodularity)
    Without loss of generality, say $A \setminus B$ is the $S$-cut and $B \setminus A$ is the $T$-cut.    
    Note that $A \setminus B = A \cap \overline{B}$ and $B \setminus A =\overline{A \cup \overline{B}}$.
    So, by submodularity on $A$ and $\overline{B}$, we get that $A \setminus B$ and $\overline{B \setminus A}$ are $S$-mincut and $T$-mincut respectively, and there are no edges between $A \setminus \overline{B} = A \cap B$ and $\overline{B} \setminus A = \overline{A \cup B}$.
\end{proof}
}

Results of very similar flavor to~\Cref{lem:sub-posi-general} have been widely observed, e.g.\ in~\cite{DinicKL76,DBLP:journals/siamcomp/DinitzV00,DBLP:journals/mp/PicardQ80,BaswanaGK22,DBLP:journals/talg/BaswanaBP23}.
We give the short proof in~\Cref{sec:missing-proofs}.
By submodularity, one sees that if $A$ and $B$ are both $(X,Y)$-mincuts, then so are $A \cap B$ and $A \cup B$.
This allows us to define \emph{farthest} and \emph{nearest} mincuts:
\begin{definition}
    The \emph{farthest (resp., nearest) $(X,Y)$-mincut}, denoted $\Far(X,Y)$ (resp., $\Near(X,Y)$), is the union (resp., intersection) of all $(X,Y)$-mincuts.
    In other words, $\Far(X,Y)$ (resp., $\Near(X,Y)$) is the unique $(X,Y)$-mincut which contains (resp., is contained in) every other $(X,Y)$-mincut.
\end{definition}

Another lemma along these lines is concerned with \emph{three} mincuts to the source $s$.
Again, this lemma has been essentially observed before, e.g.\ in~\cite{DBLP:journals/siamcomp/DinitzV00,BaswanaGK22,DBLP:conf/soda/HeHS24,BaswanaP25}.
See~\Cref{sec:missing-proofs} for the proof.

\begin{lemma}\label{lem:three-mincuts-to-source}
    Let $a,b,c \in V \setminus \{s\}$.
    Let $A,B,C$ be $(a,s)$-mincut, $(b,s)$-mincut, $(c,s)$-mincut respectively.
    If $a \notin B \cup C$, $b \notin A \cup C$ and $c \notin A \cup B$,
    then the intersection $A \cap B \cap C$ is empty.
\end{lemma}
\def\THREEMINCUTSTOSOURCE{
\begin{proof}[Proof of~\Cref{lem:three-mincuts-to-source}.]
    The assumption implies that $A \setminus (B \cup C)$, $B \setminus (A \cup C)$ and $C \setminus (A \cup B)$ are $(a,s)$-cut, $(b,s)$-cut and $(c,s)$-cut, respectively.
    Thus,
    \[
    c(A \setminus (B \cup C)) + c(B \setminus (A \cup C)) +c(C \setminus (A \cup B)) \geq c(A) + c(B) + c(C).
    \]
    However, in undirected graphs,  we have the following ``submodularity-like'' inequality which can be shown by edge counting arguments, cf.~\cite[Problem 6.48(iii)]{lovasz1979combinatorial}:
    \[
    c(A) + c(B) + c(B) \geq  c(A \setminus (B \cup C)) + c(B \setminus (A \cup C)) +c(C \setminus (A \cup B)) + c(A \cap B \cap C).
    \]
    Combining these two inequalities 
    and rearranging, we get $c(A \cap B \cap C) \leq 0$, i.e., $A \cap B \cap C$ has $0$ capacity.
    As $G$ is connected, the only two vertex subsets with $0$ capacity are $\emptyset, V$, but the latter option is impossible since $s \notin A,B,C$.
\end{proof}
}

\subsection{Essentials of the Connectivity Carcass}\label{sec:carcass-prelim}

Here, we succinctly present the ``must-knows'' of the seminal \emph{Connectivity Carcass}~\cite{DBLP:conf/stoc/DinitzV94, DBLP:conf/soda/DinitzV95, DBLP:journals/siamcomp/DinitzV00}.
The carcass is a rich and complex structure; we strongly recommend referring to~\cite{BaswanaP25} for a simplified and comprehensive presentation.
Our terminology differs slightly from theirs, mainly since we bypass the \emph{flesh} component of the carcass.
See~\Cref{sec:carcass-example} for an example demonstrating the carcass.

Let $S \subseteq V$ be a distinguished set of vertices called \emph{terminals}.
Two vertices $u,v \in V$ are \emph{$S$-equivalent} if no $S$-mincut keeps $u$ and $v$ on different sides.
We give unique identifiers to the $S$-equivalence classes, and denote by $\phi_S (u)$ the identifier of $u$'s class, so $\phi_S (u) = \phi_S (v)$ iff $u,v$ are $S$-equivalent.

A partition of $S$ into $S_1,S_2$ is called a \emph{valid partition} if there is some $S$-mincut $C$ s.t.\ $C \cap S = S_1$ and $\overline{C} \cap S = S_2$, or equivalently if any $S_1,S_2$-mincut is an $S$-mincut.
For a valid partition $S_1,S_2$, we say that: 
\begin{itemize}
    \item $S_1,S_2$ \emph{distinguishes edge $e \in E$} if there exists an $S_1,S_2$-mincut to which $e$ is contributing. 
    \item $S_1,S_2$ \emph{distinguishes vertex $v \in V$} if there exists two $S_1,S_2$-mincuts, one keeping $v$ on the side of $S_1$, and the other keeping $v$ on the side of $S_2$.
\end{itemize}

Vertex $v \in V$ is called \emph{stretched} (w.r.t.\ $S$) if $v$ is distinguished by some valid partition of $S$, and \emph{non-stretched} otherwise. 
All terminals are non-stretched, but a non-stretched vertex need not be a terminal. 

The carcass is based on a \emph{cactus} graph and related notions defined as follows:
\begin{definition}%
    A \emph{cactus} is an undirected connected graph where no two cycles share a node.\footnote{Usually, a cactus only means that no \emph{edge} is shared by two cycles, but we only deal with cacti having node-disjoint cycles.}    
    Edges/nodes in cycles are called \emph{cycle-edges/nodes}; the rest are called~\emph{tree-edges/nodes}.
    
    A \emph{canonical cut} of the cactus is either one tree-edge or a pair of non-adjacent edges from the same cycle.
    Deleting the edges of such a cut leaves exactly two connected components, called its \emph{sides}.
    
    A \emph{proper path} in the cactus is a simple path which shares at most one edge with any cycle.
    Given any two nodes in a cactus, either there is a \emph{unique} proper path between them, or there is no such path.
    Hence, a proper path can be specified by its two endpoints.
\end{definition}

The \emph{skeleton} is a cactus that compactly stores all valid partitions (see~\cite[Section 4.2]{BaswanaP25}).
\begin{theorem}[\protect{Skeleton}]
\label{thm:skeleton}
    There is an $O(|S|)$ size cactus $\cH_S$ and a mapping\footnote{The mapping $\pi_S (\cdot)$ need not be surjective, i.e., there could be nodes of $\cH_S$ without mapped terminals.} $\pi_S (\cdot)$ from $S$ to nodes of $\cH_S$, 
    called the \emph{skeleton of $S$},
    such that valid partitions of $S$ bijectively correspond to canonical cuts of $\cH_S$: 
    Canonical cut $C$ corresponds to the valid partition $S_1,S_2$ such that $S_1$ are the terminals mapped to one side of $C$, and $S_2$ are those mapped to its other side.

    Further, $\cH_S$ has the following structural properties: Each empty node (i.e., without mapped terminals) has degree at least $3$.
    Each cycle consists only of empty nodes of degree exactly $3$, and has length at least $4$.
    In particular, this implies that each edge of $\cH_S$ is contained in at least one canonical cut.
\end{theorem}

A canonical cut in $\cH_S$ is said to distinguish a vertex/edge if its corresponding valid partition distinguishes it.
The projection mapping extends $\pi_S (\cdot)$ to \emph{non-terminals} and to \emph{edges}, and provides a characterization of those canonical cuts distinguishing them (see~\cite[Sections 5.1, 5.2, 6 and 8.2]{BaswanaP25}). 

\begin{theorem}[Vertex projections]
\label{thm:proj-vertex}
    Each vertex $v \in V $ has a \emph{projection} $\pi_S (v)$ in $\cH_S$ such that:
    \begin{itemize}
        \item If $v$ is stretched, $\pi_S (v)$ is the unique proper path in $\cH_S$ such that the canonical cuts distinguishing $v$ are precisely those canonical cuts that contain an edge from $\pi_S (v)$ (i.e., those that keep $\pi_S (v)$'s endpoints on different sides).

        \item If $v$ is non-stretched, $\pi_S (v)$ is a tree-node of $\cH_S$. 
        If $v \in S$, this node coincides with the previous definition of $\pi_S (v)$ from~\Cref{thm:skeleton}.

        \item The following holds:
        Let $C$ be a canonical cut corresponding to valid partition $S_1,S_2$.
        If $C$ keeps $\pi_S (v)$ entirely on the side of $S_1$,
        then every $S_1,S_2$-mincut keeps $v$ with $S_1$.
    \end{itemize}
\end{theorem}

\begin{theorem}[Edge projections]
\label{thm:proj-edge}
    Let $e = \{x,y\}$ be an edge such that $\phi_S (x) \neq \phi_S (y)$. 
    The \emph{projection} $\pi_S (e)$ of $e$ is the unique proper path in $\cH_S$ with the following properties:
    \begin{itemize}
        \item The canonical cuts of $\cH_S$ that distinguish $e$ are precisely those that contain an edge from $\pi_S (e)$
        (i.e., those canonical cuts that keep $\pi_S (e)$'s endpoints on different sides).

        \item The endpoint nodes of $\pi_S (e)$ are labeled as the \emph{$x$-endpoint} and the \emph{$y$-endpoint}; $\pi_S (x)$ and $\pi_S (y)$ respectively form a prefix and a suffix of $\pi_S (e)$ in the direction from the $x$-endpoint to the $y$-endpoint.%
        \footnote{
        If $x$ is non-stretched, this means that the node $\pi_S (x)$ is the $x$-endpoint of $\pi_S (e)$, and similarly with $y$.
        }

        \item The following holds: Let $C$ be a canonical cut corresponding to valid partition $S_1,S_2$.
        If $C$ keeps the $x$-endpoint of $\pi_S (e)$ 
        with $S_1$ and its $y$-endpoint with $S_2$,
        then every $S_1,S_2$-mincut to which $e$ is contributing keeps $x$ with $S_1$ and $y$ with $S_2$.
        This is called the \emph{unidirectionality} property.
    \end{itemize}
\end{theorem}

\section{Technical Overview}\label{sec:technical-overview}

We give a technical overview focusing on our main result,~\Cref{thm:main}, structured as follows.
We first explain the edge insertion oracle of~\cite{BaswanaGK22}, who introduced the farthest mincut DAG, and provide some basic insights towards supporting failure queries.
Next, we introduce the key connecting link between the farthest mincut DAG and the connectivity carcass.
Based on this link, we first discuss how the simpler skeleton component alone can yield a non-trivial $O(m)$ space oracle with $O(n)$ query time.
Finally, we let carcass projections into the picture, and explain how these are used to obtain the $O(n)$ space oracle.

\subsection{From Insertion to Failure}\label{sec:DAG-overview}

Our starting point is the single-source mincuts oracle for \emph{edge insertion} by~\cite{BaswanaGK22}.
Observe that insertion of edge $e$ changes the $u,s$-mincut value if and only if $e$ contributes to $\Near(u,s)$ and to $\Far(u,s)$.
Thus, the problem reduces to efficiently representing $\mathcal{N} = \{\Near(u,s)\}_{u \in V \setminus \{s\}}$ and $\mathcal{F} = \{\Far(u,s)\}_{u \in V \setminus \{s\}}$.

Using sub/posimodularity, it is not hard to show that $\mathcal{N}$ is a laminar family, and thus admits an $O(n)$ size representation called the \emph{nearest mincut tree $\cT$}: its nodes form a partition of $V$ (where the root contains only $s$), such that, for each $u \in V \setminus \{s\}$, $\Near(u,s)$ has precisely those vertices found in descendants of the node containing $u$.
While $\mathcal{F}$ may not be laminar, the central contribution of~\cite{BaswanaGK22} was showing that it admits an analogous representation: the \emph{farthest mincut DAG} $\cD$, stated in~\Cref{thm: baswana et al DAG} below.
Throughout the paper, we use `tree terminology' for DAGs:
the in-neighbors and out-neighbors of a node are its \emph{parents} and \emph{children}, respectively; 
its \emph{descendants} are those nodes reachable from it, and its \emph{ancestors} are those that can reach it; a \emph{strict} ancestor/descendant is one that is not the node itself. 
\begin{restatable}[Theorem 4 in \cite{BaswanaGK22}]{theorem}{farthest}
\label{thm: baswana et al DAG}    
There is a DAG $\cD$ whose nodes form a partition of $V$, such that for every $u \in V \setminus \{s\}$, $\Far(u,s)$ consists precisely of those vertices found in descendants of the node containing $u$ in $\cD$.
Each node in $\cD$ has at most $2$ parents, and thus $\cD$ has only $O(n)$ edges.
The source $s$ is found alone in a source node $\{s\}$ with no parents, that is an ancestor of all other nodes.
\end{restatable}

For each structure $\mathcal{S} \in \{\cT, \cD\}$, let $\mathcal{S}(u)$ be the node of $\mathcal{S}$ that contains $u$, $\cSup(u)$ be the set of vertices in ancestors of $\mathcal{S}(u)$, and $\cSdown(u)$ be the set of vertices in descendants of $\mathcal{S}(u)$.
So, for every $u,w \in V \setminus \{s\}$,
\[
 w \in \Near(u,s) \Leftrightarrow w \in \cTdown(u) \Leftrightarrow u \in \cTup(w) \quad \text{and} \quad
    w \in \Far(u,s)  \Leftrightarrow w \in \cDdown(u) \Leftrightarrow  u \in \cDup(w).
\]
See~\Cref{sec:tree-DAG-example} for an example demonstrating $\cT$ and $\cD$.
Thus, when edge $e = \{x,y\}$ is inserted, the set of vertices $u$ such that $e$ contributes to $\Near(u,s)$ and to $\Far(u,s)$ is precisely $\cTup(x) \setminus \cDup(y) \cup \cTup(y) \setminus \cDup (x)$,
i.e., each $u$ having $x$ as a descendant in $\cT$ (so, $x\in \Near(u,s)$) and $y$ not a descendant in $\cD$ (so, $y\notin \Far(u,s)$), or vice versa.
Hence, the data structure of~\cite{BaswanaGK22} simply stores $\cT$ and $\cD$, and answers an insertion query in $O(n)$ time by $O(1)$ BFS searches in (the reverses of) $\cT$ and $\cD$.

Given this background, we move on to consider the problem of~\Cref{thm:main}: supporting a \emph{failure query} of edge $e = \{x,y\} \in E$.
Observe that the failure of $e$  changes the $u,s$-mincut value if and only if $e$ contributes to \emph{some} $u,s$-mincut.
This makes the failure variant a lot more challenging than the insertion one, which just requires examining $\Near(u,s)$ and $\Far(u,s)$, only two very special $u,s$-mincuts.
In light of this, it might seem that the nearest mincut tree $\cT$ and farthest mincut DAG $\cD$ have little to do with our task.
While $\cT$ is indeed less relevant, it turns out that $\cD$ is useful also for failure queries, as explained next.

\paragraph{Utilizing Ancestry Structure.}
A vertex $u$ is said to \emph{separate} two non-empty vertex subsets $P,Q \subseteq V$ if there exists a $u,s$-mincut that keeps $P$ and $Q$ on different sides.
Hence, the goal in a failure query of an edge $e = \{x,y\} \in E$ is to report all vertices that separate $x,y$.\footnote{Here and througout, ``to separate $x,y$'' is shorthand for ``to separate $\{x\},\{y\}$''.}
We now give a few basic observations on where these vertices lie in $\cT$ and in $\cD$, which help in answering our query, and also highlight what makes the latter more useful. 
For each $\mathcal{S} \in \{\cT, \cD\}$, let us classify the nodes into three types based on their ancestry relations to the nodes of $x$ and $y$:
\begin{description}
    \item[Common ancestor:] an ancestor of both $\mathcal{S}(x)$ and $\mathcal{S}(y)$.
    \item[Exclusive ancestor:] an ancestor of exactly one of $\mathcal{S}(x)$ or $\mathcal{S}(y)$.
    \item[Non-ancestor:] neither an ancestor $\mathcal{S}(x)$, nor of $\mathcal{S}(y)$.
\end{description}

Exclusive ancestors are trivial to handle both in $\cT$ and in $\cD$: all their vertices separate $x,y$.
Indeed, if $\cT(u)$ (resp., $\cD(u)$) is an exclusive ancestor, then the $u,s$-mincut $\Near(u,s)$ (resp., $\Far(u,s)$) keeps $x$ and $y$ on different sides.
The other two types play dual roles in $\cT$ and in $\cD$, as follows:
\begin{itemize}
    \item If $\cT(u)$ is a common ancestor, then $\Near(u,s)$ contains both $x$ and $y$, and hence $e$ cannot contribute to any $u,s$-mincut.
    Namely, common ancestors in $\cT$ cannot contain vertices that separate $x,y$.

    \item If $\cD(u)$ is a non-ancestor, then $\Far(u,s)$ keeps both $x$ and $y$ outside, and hence $e$ cannot contribute to any $u,s$-mincut.
    Namely, non-ancestors in $\cD$  cannot contain vertices that separate $x,y$.
\end{itemize}
To conclude, each of $\cT$ and $\cD$ has one non-trivial type: in $\cT$ these are the non-ancestors, while in $\cD$ these are the common ancestors.
Intuitively, this is what makes $\cD$ preferable: its non-trivial nodes are well-structured.
Indeed, if $\cD$ were a tree, the common ancestors would form a path from the root to the LCA (lowest/least common ancestor), so any two of them would be related by ancestry.
It turns out that even though $\cD$ is a DAG, it still admits a very similar property: among any three common ancestors, at least two are related by ancestry.
This is by the following lemma:
\begin{lemma}\label{lem:at-most-two-incomparable-ancestors}
    If $\mu_1, \mu_2, \mu_3$ are three nodes of $\cD$ such that\footnote{For a node $\mu$, we naturally denote $\cDup(\mu)$ and $\cDdown(\mu)$ for the sets of vertices in ancestors and in descendants of $\mu$, respectively.} $\cDdown(\mu_1) \cap \cDdown(\mu_2) \cap \cDdown(\mu_3) \neq \emptyset$, then at least two of them must be related by ancestry.
\end{lemma}
\begin{proof}
    If some $\mu_i$ is the source node $\{s\}$ this is immediate, as it is an ancestor of all other nodes.
    Otherwise, for each $i = 1,2,3$, take a vertex $u_i \in \mu_i$, so $\cDdown(\mu_i) = \Far(u_i, s)$. 
    Because 
    the intersection of these three mincuts to $s$ is non-empty,
    by \Cref{lem:three-mincuts-to-source} there must be some $i \neq j$ with $u_i \in \Far(u_j, s) =  \cDdown(\mu_j)$, and thus $\mu_i = \cD(u_i)$ is a descendant of $\mu_j$.
\end{proof}

In particular, this lemma implies that $\cD(x)$ and $\cD(y)$ can have at most two LCAs in $\cD$.\footnote{Similarly to trees, an LCA of two nodes in a DAG is a common ancestor of them, which does not have any child that is also a common ancestor. Unlike trees, there could be several LCAs, but they must be unrelated by ancestry.}
These two LCAs can be easily found in $O(n)$ time, and as discussed, all the non-trivial nodes are ancestors of (at least) one of them.
Therefore, our goal becomes to give a data structure for the following query:
\begin{enumerate}[label=\textbf{(Q)}]
    \item Given $e = \{x,y\} \in E$ and LCA $\mu$ of $\cD(x),\cD(y)$, report all vertices $u \in \cDup(\mu)$ that separate $x,y$. 
    \label{query-Q}
\end{enumerate}

\subsection{Key Insight: The Bridge to the Carcass}

We now present the key insight that underlies our approach.
The idea is to analyze how arbitrary mincuts to the source vertex $s$ ``interact'' with the farthest mincut DAG $\cD$.
Consider an arbitrary $(u,s)$-mincut $U$.
Then of course, $U \subseteq \Far(u,s) =  \cDdown(u)$, i.e., $U$ only contains vertices from descendants of $\cD(u)$.
Other than that, $U$ need not ``respect'' the structure of $\cD$: for a descendant node $\mu$ of $\cD(u)$, it could be that $U$ keeps $\mu$ entirely inside, entirely outside, or---``worst of all''---it could even splits $\mu$ between its sides, that is both $\mu\cap U$ and $\mu\cap \overline{U}$ are nonempty, which indeed turns out most challenging for us.
Our key insight to mitigate this challenge is that if $\mu$ is split by $U$, then the partition $\mu \cap U, \mu \cap \overline{U}$ cannot be arbitrary: it must be a \emph{valid partition} of $\mu$ (see~\Cref{sec:carcass-prelim}).
This is because $U$ ``induces'' a $\mu$-mincut, as formalized in the following simple-yet-crucial lemma:

\begin{lemma}\label{lem:mu-mincuts}
    Let $U$ be a $(u,s)$-mincut for some $u \in V \setminus \{s\}$, and let $\mu$ be a node of $\cD$.
    Suppose that $U$ splits $\mu$ (i.e., $\mu \cap U \neq \emptyset$ and $\mu \cap \overline{U} \neq \emptyset$).
    Then $U \cap \cDdown(\mu)$ is a $\mu$-mincut, and its capacity $\lambda_\mu$ is equal to the capacity of $\cDdown(\mu)$. 
    In particular, $\mu \cap U, \mu \cap \overline{U}$ is a valid partition of $\mu$.
\end{lemma}
\begin{proof}
    We need to prove that $\lambda_\mu = c(U \cap \cDdown(\mu)) = c(\cDdown(\mu))$.
    
    First, we note that $\lambda_\mu \leq c(U \cap \cDdown(\mu))$,
    as $U \cap \cDdown(\mu)$ is a $\mu$-cut and $\lambda_\mu$ is the capacity of $\mu$-mincut.

    Next, we argue that $c(U \cap \cDdown(\mu)) = c(\cDdown(\mu))$.
    Take some $w \in \mu \cap U$.
    As $w \in \mu$, 
    $\cDdown(\mu)$ is (the farthest) $(w,s)$-mincut.
    Now, $U \cap \cDdown(\mu)$ is a $(w,s)$-cut and $U \cup \cDdown(\mu)$ is a $(u,s)$-cut.
    So, by submodularity (\Cref{lem:sub-posi-general}), 
    $c(U \cap \cDdown(\mu)) = c(\cDdown(\mu))$.

    Finally, we show that $c(\cDdown(\mu)) \leq \lambda_\mu$.
    Take any $\mu$-mincut $C$ such that $s \notin C$.
    Choose any $v \in \mu \cap C$, hence $C$ is a $(v,s)$-cut.
    Because $v \in \mu$,
    $\cDdown(\mu)$ is (the farthest) $v,s$-mincut, so $c(\cDdown(\mu)) \leq c(C) = \lambda_\mu$.
\end{proof}

Thus, \Cref{lem:mu-mincuts} provides a ``bridge'' between single-source mincuts and Steiner mincuts:
It serves as a key that unlocks the powers of the connectivity carcass, and allows us to harness them for our purposes.

\subsection{Utilizing Skeletons for $O(m)$ Space and $O(n)$ Query Time}\label{sec:rep-framework-overview}

As detailed in~\Cref{sec: peliminaries}, the connectivity carcass has several components;
the simplest and most well-known is the \emph{skeleton}, a cactus representation for all valid partitions of the terminals (see~\Cref{thm:skeleton}).
It turns out that by only utilizing skeletons, together with our previous insights on $\cD$, we can already provide a non-trivial data structure for Query~\ref{query-Q} with $O(m)$ space and $O(n)$ query time, 
through a tool we term \emph{the representatives framework}.

On a high-level, the representative framework provides a succinct characterization for the set of vertices $u$ that separate two vertices $x,y$, which is precisely what we need for Query~\ref{query-Q}.
The framework even applies more generally, if we replace the two vertices $x,y$ with two non-empty subsets of vertices $P,Q$.
This generalization would be critical for us to obtain an optimal $O(n)$ space bound, as will be explained later on in this overview.
The framework associates $O(1)$ \emph{representative information} $\RepInf(P,Q)$ with any such pair $P,Q$.
This information, along with the $O(n)$ global information consisting of the DAG $\cD$ and the skeletons $\cH_\nu$ associated with each node $\nu$ of $\cD$, essentially encodes all the vertices $u$ that separate $P,Q$ (namely, such that some $u,s$-mincut keeps $P$ and $Q$ on different sides).
The following lemma provides the interface of this framework.
\begin{lemma}\label{lem:rep-special-case}
    There is an $O(n)$ space data structure such that, given $\RepInf(P,Q)$ for a pair of non-empty $P,Q \subseteq V$, and given a node $\mu$ of $\cD$ such that $P \cup Q \subseteq \cDdown(\mu)$, reports in $O(n)$ time the set of vertices from ancestors of $\mu$ that separate $P,Q$, namely $\{u \in \cDup(\mu) \mid \text{exists $u,s$-mincut with $P$ and $Q$ on different sides}\}$.
\end{lemma}

Given~\Cref{lem:rep-special-case}, we immediately get an $O(m)$ space data structure that answers Query~\ref{query-Q} within $O(n)$ time: we just store $\RepInf(x,y)$ for every edge $e = \{x,y\} \in E$ and the $O(n)$ space global structure from~\Cref{lem:rep-special-case}.
We now outline what is going on behind the hood of~\Cref{lem:rep-special-case}, and explain the essential role of skeletons in obtaining it.
\begin{itemize}
    \item First, by exploiting~\Cref{lem:at-most-two-incomparable-ancestors}, we show that there exists $O(1)$ \emph{representative nodes} $\nu_1, \nu_2, \dots$ for $P,Q$, such that every $u$ which separates $P,Q$ must lie in a descendant node of some $\nu_i$.
    Conversely, by using sub/posimodularity and the structure of $\cD$, we prove that if vertex $u$ lies in a \emph{strict} descendant of some $\nu_i$ and satisfies $P \cup Q \subseteq \cDdown(u)$, then $u$ must separate $P,Q$.
    We let $\RepInf(P,Q)$ store the names of $\nu_1, \nu_2, \dots$.
    So, using this information and the ``global'' DAG $\cD$,
    we can classify every vertex $u \in \cDup(\mu)$ as separating $P,Q$ or not, except for those vertices \emph{inside} the representative nodes.
    
    \item 
    Let us now zoom-in on one representative node $\nu_i$.
    By a structural analysis based on sub/posimodularity, we show that there exists $O(1)$ mincuts-to-$s$ $U_1, U_2, \dots$ such that those vertices of $\nu_i$ that separate $P,Q$ are precisely $\nu_i \cap (U_1 \cup U_2 \cup \cdots)$.
    By the ``bridge''~\Cref{lem:mu-mincuts}, each $\nu_i \cap U_j$ defines a valid partition of $\nu_i$, and thus can be determined as the terminals in one side of a corresponding canonical cut $C_{ij}$ in the skeleton $\cH_{\nu_i}$.
    Crucially, such a canonical cut $C_{ij}$ with a ``marked'' side can be compactly represented by storing the one or two skeleton edges of $C_{ij}$, with marking on their endpoints from the marked side.
    So, we let $\RepInf(P,Q)$ store $C_{i1},C_{i2},\dots$ (with the markings), and together with the ``global'' skeleton $\cH_{\nu_i}$, this lets us determine which vertices in $\nu_i$ that separate $P,Q$.
\end{itemize}

To keep the overview focused, from now on we will only care about the space of our data structure for Query~\ref{query-Q}, and neglect query time; the latter can be made $O(n)$ by using some technical extensions of the representative framework.

\subsection{Utilizing Projections Towards $O(n)$ Space}

Up until now, we have only used the skeleton component of the carcass, which enabled us to get down to $O(m)$ space.
To get optimal $O(n)$ space, we need to use the more complicated and crucial component of the carcass: the projection mapping (see~\Cref{thm:proj-vertex} and~\Cref{thm:proj-edge}).

\paragraph{Warm-Up: Internal Edges.}\label{sec:internal-overview}
As a warm-up to demonstrate our use of projections, we first discuss a special case of Query~\ref{query-Q}: the query edge $e = \{x,y\}$ is \emph{internal}, i.e., $x$ and $y$ are from the same node $\mu$ of $\cD$, which is trivially the unique LCA of $\cD(x),\cD(y)$.
The main insight for this special case is the following.
Consider any vertex $u \in \cDup(\mu)$ which we should classify as separating $x,y$ or not.
Let $U$ be some arbitrary $(u,s)$-mincut.
Using our ``bridge to the carcass'', we can characterize precisely when $U$ separates $x,y$ in terms of the projection $\pi_\mu (e)$, which is a proper path in the skeleton $\cH_\mu$.
\begin{claim}\label{claim:internal-helper}
    $U$ separates $x,y$ if and only if $\mu \cap U, \mu \cap \overline{U}$ is a valid partition of $\mu$ whose corresponding canonical cut $C$ contains an edge $g$ from $\pi_\mu (e)$.
\end{claim}
\begin{proof}
    ($\Rightarrow$)
    By~\Cref{lem:mu-mincuts}, $U \cap \cDdown(\mu)$ is a $\mu$-mincut with $e$ contributing, so its valid partition $\mu \cap U, \mu \cap \overline{U}$ distinguishes $e$, hence by~\Cref{thm:proj-edge} the latter corresponds to a canonical cut with an edge from $\pi_\mu (e)$.

    ($\Leftarrow$)
    As $x,y \in \mu$, $\pi_\mu (x)$ and $\pi_\mu (y)$ are nodes, and they form $\pi_\mu (e)$'s endpoints by~\Cref{thm:proj-edge}, so $C$ keeps them on different sides.
    Hence, by~\Cref{thm:proj-vertex}, any $\mu$-mincut whose valid partition is $\mu \cap U, \mu \cap \overline{U}$ separates $x,y$.
    But $U \cap \cDdown(\mu)$ is such a $\mu$-mincut by~\Cref{lem:mu-mincuts}, and $x,y \in \cDdown(\mu)$, so $U$ separates $x,y$.
\end{proof}

To utilize this insight through the representative framework, we show a rather simple property of the skeleton:
Each skeleton edge $g$ in $\cH_\mu$ has two sets of \emph{associated terminals} $\mu_0 (g), \mu_1 (g) \subseteq \mu$, such that the canonical cuts containing $g$ correspond exactly to the valid partitions keeping $\mu_0 (g)$ and $\mu_1 (g)$ on different parts.
We say that $u$ separates $g$ if it separates $\mu_0 (g), \mu_1 (g)$ (i.e., some $u,s$-mincut keeps these sets on different sides), and define $\RepInf(g) = \RepInf(\mu_0 (g), \mu_1 (g))$.
Thus, by the previous claim, we obtain that
\begin{claim}\label{claim:internal}
    $u$ separates $x,y$ if and only if $u$ separates some skeleton edge $g$ from $\pi_\mu (e)$.
\end{claim}

In light of this discussion, our $O(n)$ space data structure for ``internal'' Query~\ref{query-Q} consists of:
\begin{itemize}
    \item The $O(n)$ space global structure of the representatives framework from~\Cref{lem:rep-special-case}.
    \item For each node $\mu$ in $\cD$, the skeleton $\cH_\mu$ is stored along with $\RepInf(g)$ for each skeleton edge $g$ of $\cH_\mu$.
\end{itemize}
As $\cH_\mu$ has $O(|\mu|)$ edges, and $\RepInf(g)$ takes up $O(1)$ space, this amounts only to $O(n)$ storage overall.
Now, to answer a query of $e = \{x,y\} \in E$ with $x,y \in \mu$, we find the unique proper path between the nodes $\pi_\mu (x)$ and $\pi_\mu (y)$ in the skeleton $\cH_\mu$, which is $\pi_\mu (e)$ by~\Cref{thm:proj-edge}.
Then, for each edge $g$ of $\pi_\mu (e)$, we use $\RepInf(g)$ in~\Cref{lem:rep-special-case} to find all vertices $\cDup(\mu)$ that separate $g$, and return the union of all found vertices, which is all those $u \in \cDup(\mu)$ that separate $x,y$ by Claim~\ref{claim:internal}.

\paragraph{The General Case.}\label{sec:general-case-overview}
From now on, we discuss the strategy for generalizing the internal case and handling Query~\ref{query-Q} with any edge $e = \{x,y\}$.
Unlike the internal case, now we need to consider two types of mincuts to $s$ which separate $x,y$:
    1) those that split $\mu$,
    and 2) those placing all of $\mu$ entirely on one side
    (the latter were impossible when $x,y \in \mu$).
We thus fork query~\ref{query-Q} into two queries:
\begin{quote}
    Given $e = \{x,y\} \in E$ and an LCA $\mu$ of $\cD(x),\cD(y)$,
    \begin{enumerate}[label=\textbf{(Q\arabic*)}]
        \item report all $u \in \cDup(\mu)$ such that some $u,s$-mincut separates $x,y$ and splits $\mu$.
        \label{query-Q1}

        \item report all $u \in \cDup(\mu)$ such that some $u,s$-mincut separates $x,y$ and 
        keeps $\mu$ on one side.
        \label{query-Q2}
    \end{enumerate}
\end{quote}
Query~\ref{query-Q1} turns out to be the more challenging one, so we focus only on it in this overview.

\paragraph{Anchors and Query Translation.}\label{sec:translation-overview}
As in the internal case, the projection path $\pi_\mu(e)$ in the skeleton $\cH_\mu$ will still play a crucial role for us, though its use will be more involved. But, before we can use $\pi_\mu (e)$, we are first faced with a more basic challenge: how can we even \emph{find} $\pi_\mu (e)$?
Recall that $\pi_\mu (e)$ is the unique proper path that has $\pi_\mu (x)$ and $\pi_\mu (y)$ as a prefix and suffix.
When $e$ was internal, $x,y$ were \emph{terminals}, so their projections are given directly from the skeleton $\cH_\mu$, which we store.
But in general, $x$ and/or $y$ are non-terminals, and so we do not have the budget to explicitly store their projections in $\cH_\mu$.
Our strategy to tackle this challenge goes through the notion of \emph{anchors}, introduced in the following definition:

\begin{restatable}[Pivots and Anchors]{definition}{anchors}\label{def: anchors}
    For each node $\nu$ of $\cD$, choose an arbitrary vertex $z_\nu \in \nu$ called the \emph{pivot} of $\nu$.
    For a vertex $x \in \cDdown(\mu)$, where $\mu$ is a node in $\cD$, define the \emph{anchor} of $x$ w.r.t.\ $\mu$ as follows:
    \begin{itemize}
        \item If $x \in \mu$, the anchor is just $x$ itself.
        \item If $x \notin \mu$, the anchor is a pivot $z_\nu$ where $\nu$ is an arbitrary child of $\mu$ such that $x \in \cDdown(\nu)$.
    \end{itemize}
\end{restatable}

Relying on sub/posimodularity and the properties of the farthest mincut DAG $\cD$, we are able to prove that the anchors $z_x$ and $z_y$ are ``equivalent'' to the endpoints of the query edge $e = \{x,y\}$, in every meaningful sense for us. 
So, our strategy is to ``translate'' the original query endpoints $x,y$ to their anchors $z_x, z_y$.
In particular, it turns out that $\pi_\mu (z_x)$ and $\pi_\mu (z_y)$ form a prefix and suffix of the proper path $\pi_\mu(e)$, in a similar fashion to $\pi_\mu (x)$ and $\pi_\mu (y)$ (although $\pi_\mu (z_x)$ might not be equal to $\pi_\mu (x)$, and similarly with $z_y$ and $y$).
Thus, we can store the projections $\pi_\mu (z_\nu)$ for each edge $(\mu, \nu)$ of $\cD$, and still be able to recover the projection path $\pi_\mu (e)$ at query time as desired; this only requires $O(n)$ space since $\cD$ has $O(n)$ edges.

As mentioned before, finding $\pi_\mu (e)$ is merely one basic challenge that we need to overcome.
However, the equivalence of the anchors $z_x,z_y$ to the endpoints $x,y$ is key to reducing the space needed to tackle other challenges as well: in essence, whenever we would like to store some $O(1)$ information for a vertex as a ``potential endpoint'' of the yet-unknown query edge whose LCA is $\mu$, we can just store it for the anchor instead, and this will result in $O(n)$ storage overall.
So, exploiting this equivalence, in the remaining discussion we simply assume that $x = z_x$ and $y = z_y$.

\paragraph{Generalizing the Structural Analysis.}
Finally, we give an overview of the more complicated structural insights to treat the general case of Query~\ref{query-Q1}, which play an analogous role to Claim~\ref{claim:internal} for the internal case.
One direction of Claim~\ref{claim:internal} remains true: if $u$ separates $x,y$ (by a cut that splits $\mu$), then $u$ must separate some edge $g$ of $\pi_\mu (e)$.
However, the converse direction is no longer true. 
This is because now $x$ or $y$ \emph{might be stretched vertices}.
We illustrate the issue through a hypothetical example:
\begin{itemize}
    \item Say $\pi_\mu (x)$ is a prefix of $\pi_\mu (e)$ with at least one edge $g$, which is a tree-edge;
    then $\mu_0 (g), \mu_1 (g)$ form the valid partition corresponding to $g$, which distinguishes $x$ by~\Cref{thm:proj-vertex}.
    So, there are $(\mu_0 (g),\mu_1 (g))$-mincuts $A$ and $B$ s.t.\ $x \in A$ and $x \notin B$; say $A,B$ are the only $(\mu_0 (g) ,\mu_1 (g))$-mincuts.

    \item Say there is a unique $(u,s)$-mincut $U$ splitting $\mu$, where  $\mu_0 (g) \subseteq U$, $\mu_1 (g) \subseteq \overline{U}$.
    Thus, $u$ separates $g$.

    \item Finally, say $y \in \overline{U}$.
    Thus, $u$ should be included in the output for Query~\ref{query-Q1} iff $x \in U$.
    By~\Cref{lem:mu-mincuts}, $W := U \cap \cDdown(\mu)$ is a $(\mu_0 (g),\mu_1 (g))$-mincut.
    Note that $x \in U$ iff $x \in W$.
    However, we have no way of telling if $W$ is equal to $A$ or to $B$, and accordingly, both $x \in U$ and $x \notin U$ are possible.
\end{itemize}
So, we need new tools to help us characterize precisely those ``false positive'' cases where $u$ separates $\pi_\mu (e)$ but does not separate $x,y$.
These come from storing some additional information with the vertex projections $\pi_\mu (x)$ and $\pi_\mu (y)$, as follows.

First, we show that one can associate two terminal sets $\mu_0 (x)$ and $\mu_1 (x)$ with the proper path $\pi_\mu (x)$, in a similar fashion to what we did with single skeleton edges:
the canonical cuts having an edge from $\pi_\mu (x)$ correspond precisely to the valid partitions keeping $\mu_0 (x)$ and $\mu_1 (x)$ on different parts.
Further, each of the sets $\mu_0 (x), \mu_1 (x)$ has a ``matching endpoint'' of $\pi_\mu (x)$: the endpoint which matches $\mu_i (x)$ lies in the same side as $\mu_i (x)$ of any canonical cut that splits $\mu_0 (x)$ and $\mu_1 (x)$.
We store the representative information $\RepInf (\mu_0 (x) \cup \{x\}, \mu_1 (x))$ and $\RepInf(\mu_0(x), \mu_1 (x) \cup \{x\})$.
This is done during preprocessing, hence we also store symmetric information also for $\pi_\mu (y)$. 

At query time, using some additional technical tools, we determine which of $\mu_0 (x)$ or $\mu_1 (x)$ matches the $x$-endpoint of $\pi_\mu (e)$; say $\mu_0 (x)$.
Then, define $A_\mu (x)$ as the set of all vertices from $\cDup(\mu)$ that have some mincut to $s$ which keeps $x$ with $\mu_0 (x)$ on one side, and $\mu_1 (x)$ on the other.
This set can be found by using $\RepInf(\mu_0 (x) \cup \{x\}, \mu_1 (x))$ in the representative framework (\Cref{lem:rep-special-case}).
We define $A_\mu (y)$ and find it in a symmetric fashion.
Next, similarly as in the internal case, we find $\pi_\mu (e)$ as the unique proper path in $\cH_\mu$ that starts with $\pi_\mu (x)$ and ends with $\pi_\mu(y)$ (by~\Cref{thm:proj-edge}), and for each edge $g$ in $\pi_\mu (e)$, we use $\RepInf(g)$ in the representatives framework to find the set of all vertices from $\cDup(\mu)$ that separate $g$.

The analog of Claim~\ref{claim:internal} from the internal case is the following characterization, showing that the information we have obtained can be used to exactly find those vertices $u \in \cDup(\mu)$ that have some $u,s$-mincut which separates $x,y$ and splits $\mu$, which should be reported for Query~\ref{query-Q1}.

\begin{itemize}
    \item If $u$ separates a ``middle'' edge from $\pi_\mu (e) \setminus (\pi_\mu (x) \cup \pi_\mu (y))$: Report $u$.
    
    \item Else, if $u$ separates an ``overlap'' edge from $\pi_\mu (x) \cap \pi_\mu (y)$: Report $u$ iff $u \in A_\mu (x) \cap A_\mu (y)$.

    \item Else, if $u$ separates a ``wing'' edge from $\pi_\mu (e) \setminus (\pi_\mu (x) \cap \pi_\mu (y))$: Report $u$ iff $u \in A_\mu (x) \cup A_\mu (y)$.

    \item Else, when $u$ does not separate any edge from $\pi_\mu (e)$: Do not report $u$.

\end{itemize}
The correctness proof of the above characterization is based on delicate structural arguments, relying on sub/posimodularity and carcass properties (using the bridge \Cref{lem:mu-mincuts}).

\subsection{Organization}
The rest of the paper is organized as follows.

\begin{itemize}
\item
We start by introducing the tools required for our oracles. \Cref{sec:carcass-tools} gives some new carcass tools needed for our results, which essentially add some capabilities to its basic interface from~\Cref{sec:carcass-prelim} (the proofs for this section appear in~\Cref{sec:missing-proofs-carcass}).
The next~\Cref{sec:rep-framework} gives the details of the representatives framework discussed in~\Cref{sec:rep-framework-overview}.
The following~\Cref{sec:translation} gives the formal ``query translation'' using the anchors discussed in~\Cref{sec:translation-overview}.

\item
After all the needed tools are given, we provide the data structures for Query~\ref{query-Q1} and Query~\ref{query-Q2} in~\Cref{sec:splitting-mu} and~\Cref{sec:not-splitting-mu} respectively, which together give the proof of our main~\Cref{thm:main}.

\item \Cref{sec:faster-queries} then introduces our data structure for queries on the farthest mincut DAG, and combines it with the previous tools to the fast-query oracles of~\Cref{thm:faster-queries} and~\Cref{thm:all-pairs-reporting-all}.
The insertion variant of~\Cref{thm:faster-queries} appears in~\Cref{sec:edge-insertion}.

\item Some straightforward lower bounds are given in~\Cref{sec:lowerbound}: (1) $\Omega(n)$ space is needed even for the fixed-pair setting, and (2) the \emph{directed} single-source setting needs $\Omega(n^2)$ space.

\end{itemize}

\section{Carcass Tools}\label{sec:carcass-tools}

This section introduces several new tools regarding the connectivity carcass, which was presented in the Preliminaries (\Cref{sec:carcass-prelim}).
These are used in a black-box manner by our data structures; we therefore provide only their interfaces here, with the proofs deferred to~\Cref{sec:missing-proofs-carcass}.
As in~\Cref{sec:carcass-prelim}, the terminal set is denoted by $S$.

\paragraph{Associated Terminals.}

As mentioned in~\Cref{sec:internal-overview,sec:general-case-overview}, we use a mechanism of \emph{associated terminal sets} for skeleton edges and proper paths.
These are defined as follows:

\begin{definition}[Associated terminals]\label{def:associated-terminals}
    We first define the associated terminal sets for skeleton edges.
    Let $g = \{a,b\}$ be an edge in $\cH_S$.
    Then, the associated terminal sets are denoted $S(a,g)$ and $S(b,g)$, where $S(a,g)$ is defined as follows, and $S(b,g)$ is defined symmetrically:
    \begin{itemize}
        \item If $g$ is a tree-edge: 
        Let $C$ be the canonical cut defined by $g$.
        Then, $S(a,g)$ are the terminals mapped to the side of $C$ where $a$ lies.

        \item If $g$ is a cycle-edge:
        Let $a'$ be the other neighboring node to $a$ (which is not $b$) on the cycle.
        Let $C$ be the canonical cut consisting of $g$ and the cycle edge-adjacent to $a'$ from the opposite side to $a$. (Note that $C$ is indeed a canonical cut because the cycle is of length at least 4.)
        That is, $C$ is the unique canonical cut that keeps $a$ and $a'$ on one side, and all other nodes from the cycle on the opposite side.
        Then, $S(a,g)$ are the terminals mapped to the side of $C$ where $a$ and $a'$ lie.
    \end{itemize}

    We extend this definition to proper paths in $\cH_S$ as follows.
    Let $P$ be a proper path between nodes $a$ and $b$ in $\cH_S$, and let $g_a$ and $g_b$ be its edges to $a$ and to $b$ respectively (it might be that $g_a = g_b$).
    The associated terminal sets of $P$ are defined by $S(a,P) := S(a,g_a)$ and $S(b,P) := S(b,g_b)$.
\end{definition}

We have the following property of associated terminal sets on proper paths:
\begin{lemma}\label{lem:containment-of-associated-terminals}
    Let $P = (a_0, a_1, \dots, a_k)$ be a proper path in $\cH_S$, and denote its edges by $g_i = \{a_{i-1},a_i\}$.
    Then, for every $1 \leq i \leq k$, it holds that $S(a_0, g_1) \subseteq S(a_{i-1}, g_i)$ and $S(a_i, g_i) \supseteq S(a_k, g_k)$.
\end{lemma}
\def\CONTAINMENTPROPERPATH{
\begin{proof}[Proof of~\Cref{lem:containment-of-associated-terminals}]
    We observe that for every $1 \leq j < k$, $S(a_{j-1}, g_j) \subseteq S(a_j, g_{j+1})$ and $S(a_j, g_j) \supseteq S(a_{j+1},g_{j+1})$.
    This is easy to see from~\Cref{def:associated-terminals}, by considering the three possible cases: (i) both $g_j,g_{j+1}$ are tree-edges, (ii) $g_j$ is a tree-edge and $g_{j+1}$ is a cycle-edge, or (iii) $g_j$ is a cycle-edge and $g_j$ is a tree-edge.
    These are the only possible cases because $P$ is proper.
    We thus get that $S(a_0, g_1) \subseteq S(a_1,g_2) \subseteq \cdots \subseteq S(a_{i-1},g_i)$, and $S(a_i,g_i) \supseteq S(a_{i+1}, g_{i+1}) \supseteq \cdots S(a_k, g_k)$.
\end{proof}
}%

The main point of~\Cref{def:associated-terminals} is to ensure the correctness of the following~\Cref{lem:associated-terminals}.
The somewhat tedious proof is based on the skeleton properties of~\Cref{thm:skeleton}, the main one being the bijection between canonical cuts and valid partitions.

\begin{lemma}\label{lem:associated-terminals}
    Let $P$ be a proper path in $\cH_S$.
    The valid partitions of $S$ that keep the associated terminal sets of $P$ in different parts correspond precisely to the canonical cuts of $\cH_S$ that contain an edge of $P$.
\end{lemma}
\def\ASSOCIATEDTERMINALS{
\begin{proof}[Proof of~\Cref{lem:associated-terminals}]
    Consider the first direction of the lemma: suppose $C$ is a canonical cut containing some $g_i$.
    Then by~\Cref{def:associated-terminals}, we see that $C$ must keep the associated terminal sets of $g_i$ on different sides.
    Hence, 
    by~\Cref{lem:containment-of-associated-terminals},
    $C$ keeps $S(a_0, g_1)$ and $S(a_k, g_k)$ on different sides, and so its corresponding valid partition keeps them on different parts, as required.

    Let us now consider the other direction: suppose $C$ is a canonical cut corresponding to valid partition $S_1,S_2$ such that $S(a_0, g_1) \subseteq S_1$ and $S(a_k, g_k) \subseteq S_2$, and we should show that $C$ contains some edge $g_i$.
    The proof is by complete induction on the length $k$ of the proper path $P$.
    
    The inductive step is when $P$ has some \emph{intermediate tree-edge} $g_i$ with $1 < i < k$.
    In this case $S(a_{i-1}, g_i), S(a_i, g_i)$ is the valid partition corresponding to the canonical cut $\{g_i\}$ (by~\Cref{def:associated-terminals}).
    If $C = \{g_i\}$ we are done.
    Otherwise, because $g_i$ is a tree-edge $C$ has all its edges from one side of the cut defined by $g_i$; let $A_{i-1}$ be the side of $a_{i-1}$, and $A_i$ be the side of $a_i$.
    \begin{itemize}
        \item If $C$ has both edges in $A_{i-1}$: 
        Then $C$ has one side contained in $A_{i-1}$ and the other containing $A_i$.
        Further, the side containing $A_i$ must be the side $S_2$; Indeed, otherwise $S_1 \supseteq S(a_i, g_i) \supseteq S(a_k, g_k)$ 
        (the last containment is 
        by~\Cref{lem:containment-of-associated-terminals})
        but this means that $\emptyset \neq S(a_k, g_k) \subseteq S_1 \cap S_2$, a contradiction.
        Therefore, we have that $S_2 \supseteq S(a_i, g_i)$.
        As $S_1 \supseteq S(a_0, g_0)$
        and $i < k$, we can apply the induction hypothesis on the subpath path $(a_0, a_1, \dots, a_i)$ and obtain that $C$ has an edge from it.

        \item If $C$ has both edges in $A_i$:
        This is symmetrical to the previous case, now applying the induction hypothesis on $(a_{i-1}, a_i, \dots, a_k)$ which is strictly shorter since $i > 1$.
    \end{itemize}

    It remains to deal with the base cases of the induction, when $P$ has no intermediate tree-edge.    
    These are somewhat tedious.

    First, suppose $P$ is one tree-edge $g_1 = \{a_0,a_1\}$, then $S(a_0,g_1), S(a_1,g_1)$ is the valid partition corresponding to the canonical cut $\{g_1\}$ (by~\Cref{def:associated-terminals}), hence $S_1,S_2$ must be equal to this partition, so by the bijectivety in~\Cref{thm:skeleton} we must have $C = \{g_1\}$.

    Next, suppose $P$ is two adjacent tree-edges.
    Denote by $x$ the intermediate node of $P$, which must be a tree-node by~\Cref{thm:skeleton}.
    Let $f_j = \{x, y_j\}$ for $1\leq i\leq r$ be the incident edges to $x$.
    Note that $S(y_j, f_j), S(x,f_j)$ is the valid partition corresponding to the canonical cut $\{f_j\}$ (by~\Cref{def:associated-terminals}), and that $\{S(y_j,f_j)\}_{1 \leq j \leq r}$ are mutually disjoint.
    Suppose that $P = (y_1, x, y_2)$, so $S(y_1,f_1) \subseteq S_1$, $S(y_2, f_2) \subseteq S_2$, and we should show that $C$ is $\{f_1\}$ or $\{f_2\}$.
    Note that $f_1$ and $f_2$ are the only edges $f_j$ that define cuts placing $S(y_1,f_1)$ and $S(y_2,f_2)$ on different sides, hence it suffices to show that $C$ must be an edge incident to $x$.
    Towards contradiction, suppose not: then $C$ has all its edges from the side of $y_j$ in the cut defined by $f_j$, for some $1 \leq j \leq r$.
    Therefore, one of $S_1,S_2$ is contained in $S(y_j,f_j)$.
    Further, this containment is strict: otherwise the partitions $S_1,S_2$ and $S(y_j,f_j),S(x,f_j)$ are identical, so the bijectivity in~\Cref{thm:skeleton} gives $C = \{f_j\}$ in contradiction.
    Hence, the other part among $S_1,S_2$ must intersect $S(y_j, f_j)$, and it also contains $S(y_\ell, f_\ell)$ for each $\ell \neq j$,
    but this is also contradiction as $S(y_1,f_1) \subseteq S_1$ and $S(y_2,f_2) \subseteq S_2$.
        
    Finally, it remains to treat the base cases with $P$ having a cycle-edge;
    say this edge is $g= \{x_1,x_2\}$ on a cycle $\sigma = (x_0, x_1, \dots, x_r = x_0)$.
    By~\Cref{thm:skeleton}, $r \geq 4$, and each node $x_j$ on $\sigma$ has a unique incident tree-edge $f_j = \{x_j, y_j\}$.
    Note that $S(y_j,f_j),S(x_j,f_j)$ is the valid partition corresponding to the canonical cut $\{f_j\}$ (by~\Cref{def:associated-terminals}), and that $\{S(y_j,f_j)\}_{1\leq j <r}$ are mutually disjoint.
    As $P$ is a proper path without an intermediate tree-edge, exactly one of the following cases holds:
    \begin{itemize}
        \item[1)] $P = (x_1,x_2)$, and then $S(y_1, f_1) \cup S(y_0, f_0) \subseteq S_1$ and $S(y_2, f_2) \cup S(y_3, f_3) \subseteq S_2$.
        \item[2a)] $P = (y_1,x_1,x_2)$, and then $S(y_1,f_1) \subseteq S_1$ and $S(y_2,f_2) \cup S(y_3,f_3) \subseteq S_2$.
        \item[2b)] $P = (x_1,x_2,y_2)$, and then $S(y_1,f_1) \cup S(y_0,f_0) \subseteq S_1$ and $S(y_2,f_2) \subseteq S_2$.
        \item[3)] $P = (y_1,x_1,y_2,x_2)$, and then $S(y_1,f_1) \subseteq S_1$ and $S(y_2,f_2) \subseteq S_2$.
    \end{itemize}
    The implications in each case are due to~\Cref{def:associated-terminals}.
    Note that in any case, $S(y_1,f_1) \subseteq S_1$ and $S(y_2,f_2) \subseteq S_2$.
    Thus, by repeating the argument from the case where $P$ was two tree-edges, we can eliminate the possibility that $C$ has all its edges in the side of $y_j$ in the cut defined by $\{f_j\}$ for some $1 \leq j < r$.
    Hence, $C$ can either consist of two edges from $\sigma$, or $C = \{f_j\}$ for some $j$.
    But as $C$ keeps $S(y_1,f_1)$ and $S(y_2,f_2)$ on different sides, the possible options are (i) $C = \{f_1\}$, (ii) $C = \{f_2\}$, or (iii) $C$ contains $g$.
    We can eliminate some options according to the different cases:
    \begin{itemize}
        \item[1)] Options (i) and (ii) are eliminated as they keep $S(y_0,f_0)$ and $S(y_3,f_3)$ on the same side.
        \item[2a)] Option (ii) is eliminated as it keeps $S(y_1,f_1)$ and $S(y_3,f_3)$ on the same side.
        \item[2b)] Option (i) is eliminated as it keeps $S(y_0,f_0)$ and $S(y_2,f_2)$ on the same side.
        \item[3)] No option is eliminated.
    \end{itemize}
    So, in any case, the remaining options for $C$ correspond to canonical cuts having an edge from $P$.
\end{proof}
} %

\paragraph{Translation Lemma.}

Our next tool concerns the process of query translation to the anchors, overviewed in~\Cref{sec:general-case-overview}.
Specifically, it is used to prove that the projections of the anchors $z_x$ and $z_y$ form a prefix and suffix of the projection of $e = \{x,y\}$ in \Cref{cor:projection-of-edge-skeleton-translation}.
(Here, $x',y'$ play the roles of $z_x,z_y$.)

\begin{lemma}[Translation]\label{lem:skeleton-translation}
    Let $e = \{x,y\} \in E$ and $x',y' \in V$ such that:
    \begin{enumerate}[label=(P\arabic*)]
        \item For every valid partition $S_1,S_2$ of $S$, there exists some $S_1,S_2$-mincut keeping $x$ with $S_1$ and $y$ with $S_2$ if and only if there exist some $S_1,S_2$-mincut keeping $x'$ with $S_1$ and $y'$ with $S_2$. \label{prop:P1}

        \item $\phi_S(x') \neq \phi_S (y')$. Note that given (P1), this is equivalent to $\phi_S (x) \neq \phi_S (y)$.
        \label{prop:P2}
    \end{enumerate}
    Then $\pi_S (e)$ is defined, and has $\pi_S (x')$ as a prefix and $\pi_S (y')$ as a suffix, where $\pi_S (e)$ in the direction from the $x$-endpoint to the $y$-endpoint.
\end{lemma}

\def\SKELETONTRANSLATION{
\begin{proof}[Proof of~\Cref{lem:skeleton-translation}]

    For the proof, we need the following result; it is usually stated in terms of the flesh (see~\cite[Theorem 5]{BaswanaP25} and~\cite[Lemma 5.2]{DBLP:journals/siamcomp/DinitzV00}), but the formulation here is equivalent.
    \begin{lemma}\label{thm:distinctness-property-of-flesh}
        Let $u,v \in V$ such that $\phi_S (u) \neq \phi_S (v)$.
        If $S_1,S_2$ is a valid partition of $S$ that distinguishes $u$ or $v$ (or both), then there exists an $S_1,S_2$-minuct where $u$ and $v$ lie on different sides.
    \end{lemma}

    We now give the proof.
    Because (P1) and (P2) imply that $\phi_S (x) \neq \phi_S (y)$, $\pi_S (e)$ is indeed defined.
    We prove that $\pi_S (x')$ is a prefix of $\pi_S (e)$ from the $x$-endpoint); the proof for $\pi_S (y')$ being a suffix is symmetric.
    
    We start by showing that $\pi_S (x')$ is \emph{contained} in $\pi_S (e)$, in two cases:

    \begin{itemize}
        \item[(i)] If $x'$ is stretched:
        Then $\pi_S (x')$  is a proper path with at least one edge.
        By the properties of the skeleton in~\Cref{thm:skeleton} and the projection mapping in~\Cref{thm:proj-vertex,thm:proj-edge}, to show that $\pi_S (x') \subseteq \pi_S (e)$, it is enough to prove the following statement: If a valid partition $S_1,S_2$ of $S$ distinguishes $x'$, it also distinguishes $e$.
        To see this, recall that $\phi_S (x') \neq \phi_S (x')$ by P2, so~\Cref{thm:distinctness-property-of-flesh} guarantees an $S_1,S_2$-mincut separating $x'$ from $y'$.
        By (P1), this implies an $S_1,S_2$-mincut where $e$ is contributing. 

        \item[(ii)]
        If $x'$ is non-stretched: Then $\pi_S (x')$  is a tree-node, hence incident only to tree-edges;
        seeking contradiction, suppose none of these edges is from $\pi_S (e)$.
        Note that if $y'$ is stretched then $\pi_S (y') \subseteq \pi_S (e)$ by case (i), and if $y'$ is non-stretched then the node $\pi_S (y')$ is different than $\pi_S (x')$ by (P2) (as otherwise, $x',y'$ would be $S$-equivalent).
        In any case, this means we can find a canonical cut $C$ consisting of a tree-edge $g$ incident to $\pi_S (x')$, that keeps $\pi_S (x')$ and $\pi_S (y')$ on different sides.
        Let $S_1,S_2$ be the valid partition corresponding to $C$, where $S_1$ are the terminals mapped to the side of $\pi_S (x')$.
        Then by (the last item in)~\Cref{thm:proj-vertex}, any $S_1,S_2$-mincut separates $x'$ from $y'$.
        But by (P1), this means that $S_1,S_2$ distinguishes $e$, so the tree-edge $g$ must in fact be in $\pi_S (e)$ by~\Cref{thm:proj-edge} --- contradiction.
    \end{itemize}

    Finally, we prove that $\pi_S (x')$ is a prefix of $\pi_S (e)$ from the side of the $x$-endpoint.
    Seeking contradiction, suppose otherwise.
    Then as $\pi_S (e)$ contains $\pi_S (x')$, we can find a canonical cut $C$ of $\cH_S$ that keeps the $x$-endpoint of $\pi_S (e)$ on one side, but keeps the $y$-endpoint along with $\pi_S (x')$ on the other side.
    Let $S_1$ and $S_2$ be the terminal sets mapped to the sides of the $x$-endpoint and of the $y$-endpoint, respectively.
    Then by~\Cref{thm:proj-edge} there exists some $S_1,S_2$-mincut keeping $x$ with $S_1$ and $y$ with $S_2$.
    By (P1), this means that there is some $S_1,S_2$-mincut keeping $x'$ with $S_1$ (and $y'$ with $S_2$).
    But $\pi_S (x')$ lies on the side of $C$ containing $S_2$, so this is a contradiction to the last item in~\Cref{thm:proj-vertex}.
\end{proof}
} %

\paragraph{Skeleton Tools.}
We now present some tools pertaining to the skeleton $\cH_S$.
The first is an algorithm for computing the union of (marked sides) of canonical cuts, which will be useful for a generalization of the representatives framework (\Cref{thm:blackbox}) that allows us to achieve the (worst-case optimal) $O(n)$ query time in~\Cref{thm:main}.
\begin{lemma}\label{lem:union-of-cuts-in-skeleton}
    Let $C_1, \dots, C_\ell$ be canonical cuts in $\cH_S$,
    where each $C_j$ has a marked side specified markings on the endpoints of its edges that belong to the marked side.
    Then one can compute the union of all marked sides within $O(|S| + \ell)$ time.
\end{lemma}
\def\UNIONOFCUTSINSKELETON{
\begin{proof}[Proof of~\Cref{lem:union-of-cuts-in-skeleton}]
    We first consider the simpler case where $\cH_S$ is a tree.
    Then, each $C_i$ consists of one tree-edge with a marked endpoint.
    If two cuts $C_i, C_j$ are the same edge with different marked endpoints, then the union is all of $\cH_\nu$ and we are done, so suppose the cut edges are distinct.
    Let $C_1$ be the edge $\{a_1,b_1\}$ with marked endpoint $a_1$.
    We identify the marked side of $C_1$ by a BFS from $a_1$ that ignores the edge to $b_1$.
    Suppose that during this BFS we scan an edge $\{a_i, b_i\}$ from some cut $C_i$, in the direction from the marked endpoint $a_i$ and into the unmarked endpoint $b_i$.
    This means that the marked side of $C_1$ contains all of the unmarked side of $C_i$, and vice versa.
    Thus, the union is all of $\cH_S$ and we are done.
    If this does not happen, we add the marked side of $C_1$ to the union, delete it from the tree $\cH_\nu$ and continue the process on the remaining subtree with the surviving cuts among $C_2, \dots, C_\ell$.
    The deletions ensure that the total time we spend to compute the union is only $O(|S|)$.
    (Note that, when $\cH_S$ is a tree there can be only $O(|S|)$ canonical cuts, so $\ell = O(|S|)$.)

    Now we consider the general case where $\cH_S$ is a cactus that can have cycles.
    Let us first focus on some cycle $\sigma = (a_0, a_1, \dots, a_{k-1}, a_k = a_0)$ of length $k$.
    Suppose that among the given cuts $C_1, \dots, C_\ell$, there are $r$ cuts having two cycle-edges from $\sigma$.
    We identify each such cut with a marked interval on the cycle $\sigma$: this is the interval between the marked endpoints of the two cycle-edges that define the cut, which leaves their unmarked endpoints out.
    Within $O(k+r)$ time, we can create a list $M_\sigma$ of all nodes of $\sigma$ that lie in some marked interval; we defer the argument to the end of the proof.
    Next, we add to $\cH_S$ a dummy node $a_\sigma$, and replace the cycle $\sigma$ by a star of edges between $a_\sigma$ and each node $a_i \in \sigma$.
    Finally, we replace the $r$ cuts in $\sigma$ by $|M_\sigma|$ one-edge cuts, consisting of the edges $\{a_\sigma, a_i\}$ with $a_i \in M_\sigma$ with marked endpoint $a_i$.
    The union of marked sides of these new cuts is the same as for the original $r$ cuts in the cycle $\sigma$.
    Hence, executing this transformation maintains the same overall union of marked sides (ignoring the new dummy node $a_\sigma$).
    Because the cycles are edge-disjoint, and edge of the given $\ell$ cuts corresponds comes from at most one cycle, the total time to eliminate them in this manner is $O(|S| + \ell)$.
    At this point, we are left with a tree with $O(|S|)$ nodes, and at most $O(|S|)$ single tree-edge cuts marked sides, so we can compute the union in $O(|S|)$ time as explained before.
    Thus, the total running time is $O(|S| + \ell)$.
    
    Finally, let us address the procedure for computing $M_\sigma$, namely, uniting $r$ intervals in a length-$k$ cycle on $\{0, 1 \dots, k-1\}$.
    By splitting each interval that wraps over the $0$ into two, we get $\leq 2r$ intervals of the form $[a,b]$, $0 \leq a \leq b \leq k-1$, that have the same union.
    Sort these intervals by non-decreasing order of their starting points;
    this can be done in $O(k+r)$ with Counting Sort.
    Now, in one pass over the sorted list, we merge overlapping intervals to obtain a new list of disjoint intervals with the same union.
    This takes $O(k)$ time.
    We can then explicitly list all of the integers in this union within $O(r)$ time.
\end{proof}
}%

Next, we present some data structures for the skeleton.
These are not required for our main~\Cref{thm:main}, but needed for the data structures of~\Cref{thm:faster-queries} that support faster queries.

\begin{lemma}\label{lem:skeleton-data-structure}
    There is an $O(|S|)$ space data structure for the skeleton $\cH_S$, supporting the following queries:
    \begin{itemize}
        \item \textbf{(Cut side)}
        Given a canonical cut $C$ and a node $a$, return which endpoints of the edges of $C$ are on the side of $C$ that contains $a$.
        The query takes $O(1)$ time.

        \item \textbf{(Order on proper path)} 
        Given $k$ nodes $a_1, \dots, a_k$ which are promised to all lie on a proper path, return the order of appearance of $a_1, \dots, a_k$ on the unique minimal proper path containing them (in an arbitrary direction).
        The query takes $O(k \log k)$ time.

        \item
        \textbf{(Cut-sides formula)}
        Given $k$ canonical cuts $C_1, \dots, C_k$ with marked sides $A_1, \dots, A_k$ (specified by marking the endpoints of edges in $C_j$ that belong to $A_j$), and a Boolean formula $\varphi(x_1,\dots,x_k)$ where each variable $x_i$ occurs exactly once, return
        \[
        S' = \Big\{t \in S \mid \varphi\big(\mathbf{1}_{[\pi_S (t) \in A_1]}, \dots, \mathbf{1}_{[\pi_S (t) \in A_k]} \big) = 1 \Big\}.
        \]
        That is, $S'$ is the set of terminals mapped to the nodes of $\cH_S$ obtained by evaluating $\varphi$ as a set expression, replacing $x_j$ by $A_j$, conjunction by intersection, disjunction by union, and negation by complement.
        The query takes $O(|S'| + k^2)$ time.
    \end{itemize}
\end{lemma}
\def\SKELETONDATASTRUCTURE{
\begin{proof}[Proof of~\Cref{lem:skeleton-data-structure}]
    Let $\mathcal{T}_S$ be the tree obtained from $\cH_S$ by replacing cycles with stars as follows:
    for each cycle $\sigma = (a_0, a_1, \dots, a_k = a_0)$ in $\cH_S$, add a new auxiliary node $a_\sigma$ with edges into every node $a_i$ in $\sigma$, and delete the original edges of $\sigma$.
    The auxiliary node $a_\sigma$ also stores the cyclic order of its neighbors according to $\sigma$, so the skeleton $\cH_S$ can be completely recovered from $\mathcal{T}_S$.
    Note that $\mathcal{T}_S$ also has $O(|S|)$ nodes and edges.
    We root $\mathcal{T}_S$ arbitrarily, and let each node in $\mathcal{T}_S$ store its depth (i.e., its distance from the root).
    We store $\cH_S$ and $\mathcal{T}_S$, along with an data structure for $\mathcal{T}_S$ that has linear size in the number of nodes in the tree, and answers LCA or level-ancestor queries\footnote{A level-ancestor query gets node $a$ and non-negative integer $d$, and returns the ancestor of $a$ that has depth $d$ (or a null value if $a$ has depth smaller than $d$.)} in constant time~\cite{HT84,BenderF00,BenderF04}.
    So, the total size of the data structure is $O(|S|)$.
    We now describe how the different queries are answered.

    (Cut side) There are two cases:
    \begin{itemize}
        \item $C$ consists of one tree-edge of $\cH_S$:
        Denote this edge by $g = \{z_0,z_1\}$, where $z_1$ is deeper than $z_0$ in $\mathcal{T}_S$.
        We make an LCA query between $a$ and $z_1$ in $\mathcal{T}_S$: if the LCA is $z_1$ than we return $z_1$ as the endpoint on the side of $a$, and otherwise we return $z_0$.

        \item $C$ consists of two cycle-edges from the same cycle of $\cH_S$:
        Let $\sigma = (z_0, z_1, \dots, z_k = z_0)$ be the cycle, denote the edges of $C$ by $g_1 = \{z_i, z_{i+1}\}$ and $g_2 = \{z_j, z_{j+1}\}$, and suppose without loss of generality that $0 \leq i < j < k$.
        Now, consider the path in $\mathcal{T}_S$ from $a$ to $a_\sigma$.
        The penultimate node of this path is some $z_\ell$ on $\sigma$.
        If $i+1 \leq \ell \leq j$, then we should report $z_{i+1}$ and $z_j$ as the endpoints on the side of $a$.
        Otherwise, we should report $z_i$ and $z_{j+1}$.
        To find $z_\ell$, we first make an LCA query between $a_\sigma$ and $a$.
        If the LCA is not $a_\sigma$, then $z_\ell$ must be the parent of $a_\sigma$ in $\mathcal{T}_S$.
        Otherwise, $z_\ell$ must be the ancestor of $a$ whose depth is smaller by 1 than the depth of $a_\sigma$, so we find it with a level-ancestor query.
    \end{itemize}

    (Order on proper path) 
    Let $z_1 = a_1$, and $z_i = \LCA_{\mathcal{T}_S} (z_{i-1}, a_i)$ for $i = 2, \dots, k$, and denote $z = z_k$.
    Thus, $z$ is found within $k$ ancestry queries.
    Note that the unique minimal proper path which contains $a_1,\dots,a_k$ corresponds to a simple path in $\mathcal{T}_S$ which preserves the same ordering, goes up the tree until it hits $z$, and then either stops or continues down the tree.
    Thus, we just need to classify which $a_i \neq z$ belongs to the ``upward'' part, and which to the ``downward'' part.
    At this point, we sort the $a_i$'s in those two parts according to increasing or decreasing order of depth, respectively.
    Finally, we concatenate the two sorted lists (possibly adding $z$ in the middle, in case it is one of the $a_i$'s) and return this list as the ordering.
    For the classification process, we use $k$ level-ancestor queries.
    For each $a_i \neq z$, we find the ancestor of $a_i$ whose depth is smaller by 1 than the depth of $z$;
    there can be at most two children of $z$ returned by these queries, and we can classify $a_i$ according to the child found in its query.

    (Cut-sides formula)
    The query algorithm is simple.
    Let $F$ be the set of all edges contained in at least one of the cuts $C_1, \dots, C_k$, so $|F| \leq 2k$.
    For each endpoint $a$ of an edge from $F$, use cut side queries to determine if $a$ lies on the marked side $A_j$ of $C_j$, for $j = 1, \dots, k$.
    Then, check if $a$ satisfies the formula, i.e., if $\varphi(\mathbf{1}_{a \in A_1}, \dots, \mathbf{1}_{a \in A_k}) = 1$.
    Let $I$ be the set of all such $\varphi$-satisfying endpoints of $F$.
    Run multi-origin BFS from $I$ in $\cH_S-F$, and let $M$ be the set of visited nodes.
    Report $S'$ as the terminals mapped to $M$.

    To see the correctness, consider the connected components $H_1, \dots, H_\ell$ of $\cH_S - F$.
    All nodes inside one $H_i$ coincide on the sides of all cuts $C_1, \dots, C_k$, as they are connected by paths that do not go through any edge of any $C_j$.
    Also, each $H_i$ contains some endpoint of $F$.
    Thus, those $H_i$'s that satisfy $\varphi$ are exactly those that intersect $I$, i.e., those visited by the BFS.
    By re-indexing, suppose these are the first $r$ components, so $M = V(H_1) \sqcup \dots \sqcup V(H_r)$.
    Thus, $M$ consists of all $\varphi$-satisfying nodes, so $S'$ are the terminals mapped to $M$.

    We now analyze the running time.
    Computing $I$ takes involves $|F| \cdot k = O(k^2)$ cut side queries, each taking constant time.
    Once these are computed, we need to evaluate $\varphi$ on each node in $F$, which takes $|F| \cdot O(k) = O(k^2)$ time (as each of the $k$ variables in $\varphi$ appears exactly once).
    Analyzing the running time of the BFS is a bit more involved.
    Note that, as $\cH_S$ is a cactus graph, the number of edges that touch $M$ is $O(|M| + |F|) = O(|M|+k)$, so the BFS takes $O(|M|+k)$ time.
    In the following we show that $|M| = O(|S'| + k)$, which yields the desired $O(|S'| + k^2)$ bound on the total time to answer the query.

    Let $S_i$ denote the set of terminals mapped to $H_i$, and let $f_i = \sum_{a \in V(H_i)} \deg_F (a)$.
    The crux of the argument is showing that $|V(H_i)| \leq 3(|S_i| + f_i)$, which we will shortly do.
    Given this, we sum over $i = 1, \dots, r$ and get
    \[
    |M| = 3\sum_{i=1}^r |S_i| + 3\sum_{i=1}^r f_i = 3|S'| + 3\sum_{i=1}^r \sum_{a \in V(H_i)} \deg_F (a) \leq 3|S'| + 3\cdot 2|F| \leq 3|S'| + 12k    \]

    We now prove that $|V(H_i)| \leq 3(|S_i| + f_i)$.
    Color the nodes of $H_i$ black or white, where the black nodes are those that have mapped terminals.
    Let $H'_i$ be obtained from $H_i$ by adding $f_i$ new edges, going into $f_i$ new black nodes, such that each $a \in V(H_i)$ gets $\deg_F (a)$ new edges.
    Thus, for every node $a \in V(H_i)$, we have $\deg_{H'_i} (a) = \deg_{\cH_S - F} (a) + \deg_F (a) = \deg_{\cH_S} (a)$.
    Therefore, by the properties of the skeleton $\cH_S$ (see~\Cref{thm:skeleton}):
    \begin{itemize}
        \item the white nodes in $H'_i$ all have degree at least $3$,
        \item the cycles in $H'_i$ contain only white nodes of degree exactly 3, and
        \item each such cycle has length at least 4, and does not share nodes with any other cycle.
    \end{itemize}
    Let $X_i$ be the nodes found in cycles of $H'_i$, let $Y_i$ be its remaining white nodes outside the cycles, and let $Z_i$ be its black nodes.
    Consider the tree $T_i$ obtained by contracting each cycle $\sigma$ of $H'_i$ into a (white) super-node $a_\sigma$.
    Observe that $\deg_{T_i} (a_\sigma)$ is the number of edges that leave the cycle $\sigma$ in $H'_i$.
    This number is equal to the length of the cycle $\sigma$, because each node in $\sigma$ cycle has two edges in $\sigma$ and one additional edge going out of $\sigma$.
    Therefore, the sum of degrees of super-nodes in $T_i$ is precisely $|X_i|$.
    Nodes in $Y_i$ and $Z_i$ have the same degree in $T_i$ as in $H'_i$, so the sum of degrees in $Y_i$ is at least $3|Y_i|$, and in $Z_i$ it is at least $|Z_i|$.
    We get that
    \[
    |X_i| + 3|Y_i| + |Z_i| \leq 2|E(T_i)| \leq 2|V(T_i)| \leq 2(|X_i|/4 + |Y_i| + |Z_i|) = |X_i|/2 + 2|Y_i| + 2|Z_i|.
    \]
    (where the last inequality is because there can be at most $|X_i|/4$ super-nodes in $T_i$, since each cycle in $H'_i$ has at least 4 nodes).
    Rearranging the above inequality yields that $|X_i| + 2|Y_i| \leq 2|Z_i|$.
    Using this, we finally get that
    \[
    |V(H_i)| \leq |V(H'_i)| = |X_i| + |Y_i| + |Z_i| \leq 3|Z_i| \leq 3(|S_i| + f_i),
    \]
    where the last inequality is because every node in $Z_i$ except the $f_i$ new nodes has a mapped terminal. 
\end{proof}
}

\begin{lemma}\label{lem:proper-path-information-data-structure}
    Suppose each proper path $P$ in $\cH_S$ has some $O(1)$ information $I(P)$ associated with it.
    Then, there an $O(|S| \log |S|)$ space data structure that supports the following queries:
    Given the two endpoints of a proper path $P$, return a partition of $P$ into $\ell = O(\log |S|)$ subpaths $P_1, \dots, P_\ell$ (where each subpath is specified by its endpoints), along with the information $I(P_1), \dots ,I(P_\ell)$. The query takes $O(\log |S|)$ time.
\end{lemma}
\def\PROPERPATHINF{
\begin{proof}[Proof of~\Cref{lem:proper-path-information-data-structure}]
    First, we store the same $O(|S|)$ data structure as in~\Cref{lem:skeleton-data-structure}.
    In particular, we store the tree $\mathcal{T}_S$ (where each cycle $\cH_S$ is transformed into stars with an auxiliary center node), and the LCA and level-ancestor data structures for $\mathcal{T}_S$ which support these queries in $O(1)$ time.
    
    Next, consider any proper path $Q$ between the endpoints $z_0$ and $z_1$ in $\cH_S$.
    We explicitly store $I(Q)$ only if one of the following conditions is met:
    \begin{enumerate}[label=(\roman*)]
        \item $Q$ consists of a single cycle-edge of $\cH_S$.
        \item $z_0$ is a descendant of $z_1$ in $\mathcal{T}_S$, and their distance in $\mathcal{T}_S$ is either $2^j$ or $2^j+1$ for some integer $j \geq 0$.
    \end{enumerate}
    Thus, each $z_0$ can be an endpoint of at most $O(\log |S|)$ distinct proper paths $Q$ for which we store $I(Q)$: indeed, case (i) can only occur twice because $z_0$ can only be a part of one cycle, and case (ii) can only occur $O(\log |S|)$ times because $\cT_S$ is of $O(|S|)$ height.    
    Therefore, in total, we only store $O(|S| \log |S|)$ information of $I(\cdot)$ for different proper paths.

    We now show how a query is answered.
    Let $a$ and $b$ be the endpoints of the proper path $P$ given as input.
    Let us first handle the case when $a$ and $b$ are related by ancestry in $\cT_S$, say $b$ is an ancestor of $a$.
    Afterwards, we will reduce the general case to this case.
    The algorithm is as follows.
    Initialize $i = 1$ and $a_1 = a$.
    While $a_i \neq b$, do the following:
        \begin{itemize}
            \item Let $j$ be the largest integer such that the distance between $a_i$ and $b$ is at least $2^j$.
            \item By a level-ancestor query, find the ancestor $z_i$ of $a_i$ within distance exactly $2^{j_i}$ from $a_i$.
            \item If $z_i$ is an original node from $\cH_S$, then define $a_{i+1} := z_i$.
            
            \item Otherwise, $z_i$ is an auxiliary node of $\cT_S$ corresponding to a cycle in $\cH_S$, and its parent $z'_i$ in $\cT_S$ is an original node from that cycle.
            In this case, defined $a_{i+1} = z'_i$.

            \item Report the subpath $P_i$ of $P$ as the proper path between $a_i$ and $a_{i+1}$, along with its corresponding information $I(P_i)$.
            Note that the latter is explicitly stored by the data structure, as the distance between $a_i$ and $a_{i+1}$ in $\cT_S$ is either $2^j$ or $2^j+1$.

            \item Increment $i \gets i+1$.
        \end{itemize}
    Thus, the $i$-th iteration of the while loop ensures that the distance between $a_{i+1}$ and $b$ in $\cT_S$ is at most half of the distance between $a_i$ and $b$.
    As the starting distance between $a = a_1$ and $b$ is at most $O(|S|)$, there can be only $O(\log |S|)$ iterations.
    As each iteration takes $O(1)$ time, we also get the desired query time.

    Finally, consider the case where $a$ and $b$ are not related by ancestry in $\cT_S$.
    Let $c$ be the LCA of $a,b$, which we find by one LCA query.
    If $c$ is an original node from $\cH_S$, then $c$ appears on the proper path $P$ between $a$ and $b$.
    Thus, we just apply the query algorithm from before $a,c$ and for $b,c$, and concatenate the results.
    The remaining case is when $c$ is an auxiliary node corresponding to a cycle $\sigma$ in $\cH_S$.
    This means that the proper path $P$ goes through exactly one edge of $\sigma$.
    Furthermore, this edge must connect $c_a$ and $c_b$, which are, respectively, the children of $c$ that are ancestors of $a$ and of $b$ in $\mathcal{T}_S$.
    So, we find $c_a$ and $c_b$ by two level-ancestor queries,
    apply the algorithm for before on $a,c_a$ and on $b,c_b$, and concatenate the results along with $I(\{c_a,c_b\})$ which we've explicitly stored.
\end{proof}
}

\paragraph{Numbering for Projection Directionality.}
The following~\Cref{lem:sigma-numbers} is a technical lemma that assigns numbers to vertices, helping us deal with certain ``directionality properties'' of projections in certain corner cases.
We use it in~\Cref{sec:splitting-mu}.
Specifically, our use-case is (roughly) as follows: if edge $e = \{x,y\}$ has its projection $\pi_S (e)$ equal to both $\pi_S (x)$ and $\pi_S (y)$, we can use the numbers assigned to $x$ and $y$ in order to determine the $x$-endpoint and the $y$-endpoint of $\pi_S (e)$.
This lemma is not new: it has appeared before in~\cite{DBLP:conf/soda/BaswanaP22,DBLP:conf/icalp/Bhanja25}, but was stated using carcass terminologies which we bypass in this paper.
Thus, we give an explicit statement and provide a proof that assumes only the carcass preliminaries of~\Cref{sec:carcass-prelim}.

\begin{lemma}\label{lem:sigma-numbers}
    There exists a numbering $\sigma_S : V \to [n]$ such that the following hold.
    Let $P$ be a proper path in $\cH_S$ whose endpoints are $a$ and $b$, where node $a$ has smaller identifier than node $b$ in $\cH_S$.
    Let $z_1, z_2 \in V$ such that $\phi_S (z_1) \neq \phi_S (z_2)$, but they both have $P$ as their projection, i.e., $\pi_S (z_1) = \pi_S (z_2) = P$.
    Then:
    \begin{itemize}
        \item It holds that $\sigma_S (z_1) \neq \sigma_S (z_2)$.
        \item If $\sigma_S (z_1) < \sigma_S (z_2)$, then there exists some $S$-mincut $W$ with $z_1 \in W$, $z_2 \in \overline{W}$ and a canonical cut of $\cH_S$ that keeps node $a$ with $S \cap W$ on one side, and node $b$ with $S \setminus W$ on the other side.
    \end{itemize}
\end{lemma}
\def\SIGMANUMBERS{
\begin{proof}[Proof of~\Cref{lem:sigma-numbers}]
    Define a directed graph $H$ on the vertices $V$ whose edges are defined as follows: 
    edge $(z,z')$ is \textbf{not} in $E(H)$ iff (at least) one of the following conditions holds:
    \begin{itemize}
        \item[(i)] $\pi_S (z) \neq \pi_S (z')$.
        \item[(ii)] There exists an $S$-mincut $W$ and a canonical cut $C$ of $\cH_\mu$ such that $z' \in W, z \in \overline{W}$ and $C$ keeps the smaller-identifier endpoint of $\pi_S (z)$ with $S \cap W$ on one side, and the larger-identifier endpoint of $\pi_S (z)$ with $S \setminus W$ on the other.
    \end{itemize}

Contract each SCC (strongly connected component) of $H$, find a topological ordering of the resulting DAG, and define $\sigma_S (z)$ as the place of the SCC containing $z$ in the ordering.

Given this definition of the numbering $\sigma_S$, let us prove its properties as stated in the lemma.
We first prove the second item, which is more straightforward.
Suppose $\sigma_S (z_1) < \sigma_S (z_2)$.
Then, the edge $(z_2,z_1)$ cannot exist in $H$, as it would violate the topological ordering.
So, by the definition of $H$, one of the conditions (i) and (ii) must hold, but it cannot be (i) because $\pi_S (z_1) = \pi_S (z_2)$.
Thus, condition (ii) holds (where $z = z_2$ and $z' = z_1$), but this condition is precisely what we need to prove.

We now prove the more involved first item.
Because $\phi_S (z_1) \neq \phi_S (z_2)$, there exists an $S$-mincut $W$ which separates $z_1$ from $z_2$.
Let $C$ be the canonical cut of $\cH_S$ corresponding to the valid partition $S \cap W, S \setminus W$.
As $P = \pi_S (z_1) = \pi_S (z_2)$, it is impossible that $P$ lies completely on one side of $C$; Indeed, if that were to happen, then $z_1,z_2$ would be in the same side of every $S \cap W, S \setminus W$-mincut by~\Cref{thm:proj-vertex}, but $W$ is such a cut keeping them on different sides.
Hence, $C$ keeps the endpoints $a$ and $b$ of $P$ on different sides.
We may assume that $a$ is with $S \cap W$ and $b$ is with $S \setminus W$ (otherwise, replace $W$ by $\overline{W}$).
Without loss of generality, say $z_2 \in W$ and $z_1 \in \overline{W}$ (the case $z_1 \in W$ and $z_2 \in \overline{W}$ is symmetric).

Now, suppose towards contradiction that $\sigma_S (z_1) = \sigma_S (z_2)$, meaning $z_1, z_2$ are in the same SCC.
Take some path $(z_1 = v_0, v_1, \dots, v_k = z_2)$ from $z_1$ to $z_2$ in $H$.
Note that for each edge $(v_{i-1},v_i)$ on this path, condition (i) does not hold, meaning $\pi_S (v_{i-1}) = \pi_S (v_i)$.
Thus, we obtain that that $\pi_S (v_i) = P$ for every $v_i$ on the path.
Now, consider the largest $0 \leq r \leq k-1$ with $v_r \in \overline{W}$ (which is well-defined since $v_0 = z_1 \in \overline{W}$ and $v_k = z_2 \in W$).
Then $v_{r+1} \in W$.
But now, $W$ and $C$ imply that condition (ii) holds for $(v_i, v_{i+1})$, contradicting the fact that this edge belongs to $H$.
\end{proof}
}%

\section{The Representatives Framework}\label{sec:rep-framework}

In this section, we give the details of the representative framework, which was overviewed in~\Cref{sec:rep-framework-overview}.
The framework associates $O(1)$ \emph{representative information} with any unordered pair $P,Q$ of non-empty vertex subsets, in such a manner that guarantees the following (which is generalized version of~\Cref{lem:rep-special-case}):

\begin{lemma}\label{thm:blackbox}
    There is an $O(n)$ space data structure with the following property:
    Given $\{\RepInf(P_i,Q_i)\}_{i=1}^\ell$ for $\ell$ pairs of non-empty vertex subsets $P_i, Q_i \subseteq V$, and given a node $\mu$ of $\cD$ such that $P_i \cup Q_i \subseteq \cDdown(\mu)$ for each $i=1,\dots,\ell$, it reports within $O(n + \ell)$ time those vertices in $\cDup(\mu)$ that separate some pair $P_i, Q_i$, namely $\{u \in \cDup(\mu) \mid \text{exists $1 \leq i \leq \ell$ and $u,s$-mincut $U$` keeping $P_i$ and $Q_i$ on different sides}\}$.
\end{lemma}

Recall that vertex $u$ is said to separate $P,Q$ if there exists a $u,s$-mincut which keeps $P$ on one side and $Q$ on the other.
We will also need an ordered version of this terminology: we say that $u$ separates $(P,Q)$ if there exists a $(u,s)$-minuct $U$ such that $P \subseteq U$ and $Q \subseteq \overline{U}$.
Thus, $u$ separates $P,Q$ iff $u$ separates $(P,Q)$ or $u$ separates $(Q,P)$.

\paragraph{Representative Nodes.}
In order to prove~\Cref{thm:blackbox}, we first need \emph{to define} $\RepInf(P,Q)$.
In fact, each ordered pair $(P,Q)$ or $(Q,P)$ has its own representative information, and $\RepInf(P,Q)$ consists of the information for both;
we arbitrarily focus on $(P,Q)$.
The first ingredient is the \emph{representative nodes}, defined in the following lemma.

\begin{lemma}\label{lem:rep-nodes}
    We can find at most two \emph{representative nodes} for $(P,Q)$ in $\cD$, such that:
    \begin{itemize}
        \item Each representative node contains some vertex separating $(P,Q)$.
        
        \item
        Every node that contains a vertex separating $(P,Q)$ is a descendant of
        some representative node.
    \end{itemize}
\end{lemma}
\begin{proof}
    It suffices to show that among any three nodes that contain some vertex separating $(P,Q)$, at least two must be related by ancestry:
    indeed, this immediately means that among them we can find at most two ``highest ones'' which are ancestors of all the others.
    By~\Cref{lem:at-most-two-incomparable-ancestors}, in order to prove this, it is enough if we can find a vertex $p$ such that $p \in \cDdown(\nu)$ for every $\nu$ that contains some vertex $u$ separating $(P,Q)$.
    To this end, we can take any arbitrary $p \in P$:
    indeed, as $u$ separates $(P,Q)$, in particular there is a $(u,s)$-mincut $U$ such that $P \subseteq U$, and hence $p \in U \subseteq \Far(u,s) = \cDdown(\nu)$.
\end{proof}
\begin{figure}
  \begin{center}
    \includegraphics[width=\textwidth]{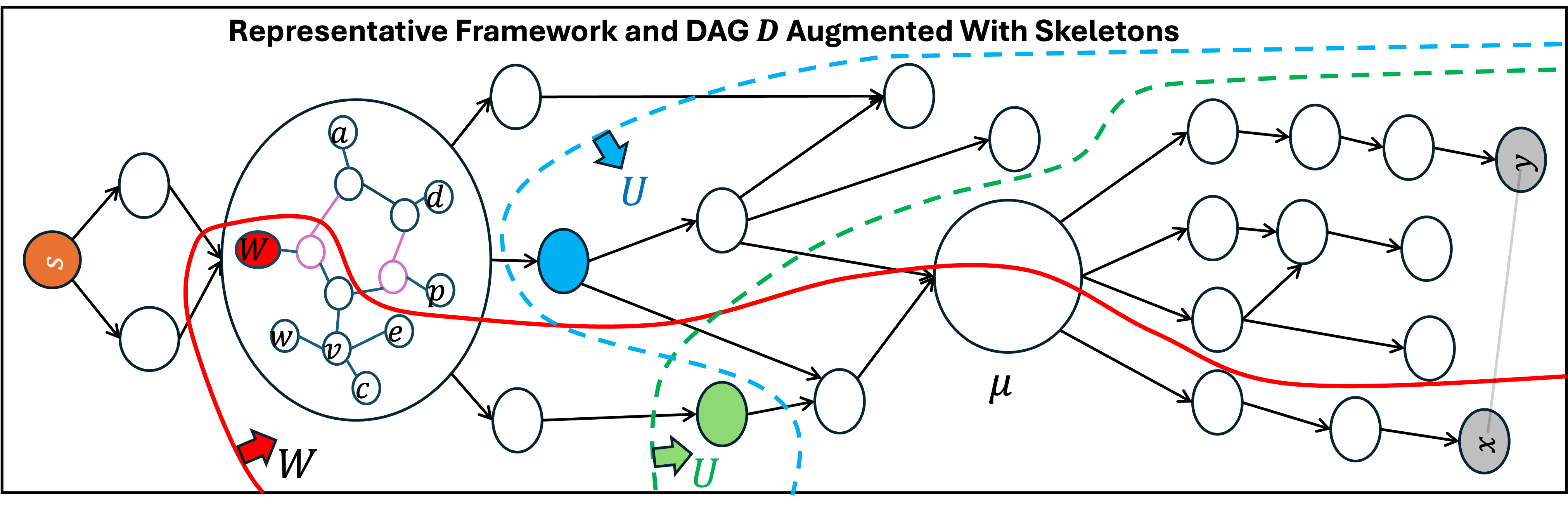}
  \end{center}
  \caption{Let $W$ be a $w,s$-mincut, $U$ with color green, and $U$ with color blue be the farthest mincuts for vertices from nodes green and blue, respectively. Edge $\{x,y\}$ belongs to the graph for which $\mu\in \LCA(x,y)$ in $\cD$. Each node of the DAG contains a Skeleton.} %
  \label{fig : DAG Skeleton RepInf}
\end{figure}
One may hope for the following ``converse to~\Cref{lem:rep-nodes}'': if $\cD(u)$ has some representative node for $(P,Q)$ as an ancestor, then $u$ should separate $(P,Q)$.
While this fails, we next prove an ``almost-converse''; the caveats are in conditions (ii) and (iii), and in the conclusion of separating the unordered $P,Q$.
\begin{lemma}\label{lem:rep-nodes-converse}
    Let $u \in V \setminus \{s\}$ such that the following hold:
    \begin{itemize}
        \item[(i)] Some representative node $\nu$ for $(P,Q)$ is an ancestor of $\cD(u)$.
        \item[(ii)] $\nu \neq \cD(u)$. (So by (i), $\nu$ is a \emph{strict} ancestor of $\cD(u)$.)
        \item[(iii)] $P \cup Q \subseteq \cDdown(u)$.
    \end{itemize}
    Then $u$ separates the unordered pair $P,Q$.
\end{lemma}
\begin{proof} We refer to \Cref{fig : DAG Skeleton RepInf} for an illustration of the proof.
    By~\Cref{lem:rep-nodes}, the representative node $\nu$ contains some vertex $w$ that separates $(P,Q)$.
    Let $W$ be a $(w,s)$-mincut with $P \subseteq W$, $Q \subseteq \overline{W}$.
    Denote $U := \cDdown(u) = \Far(u,s)$. 
    As $\nu$ is a strict ancestor of $\cD(u)$ by (i) and (ii), we get that $w \notin U$.
    Also, $P, Q \subseteq U$ by (iii).
    We split into cases:

    \begin{itemize}
        
        \item
        If $u \in W$, then $U \cap W$ is a $(u,s)$-cut and $U \cup W$ is a $(w,s)$-cut. So by submodularity (\Cref{lem:sub-posi-general}), $U \cap W$ is a $(u,s)$-mincut, where $P \subseteq U \cap W$ and $Q \subseteq \overline{U \cap W}$.

        \item If $u \notin W$, then $U \setminus W$ is a $(u,s)$-cut and $W \setminus U$ is a $(w,s)$-cut.
        So by posimodularity (\Cref{lem:sub-posi-general}), $U \setminus W$ is a $(u,s)$-mincut, where $Q \subseteq U \setminus W$ and $P \subseteq \overline{U \setminus W}$.

    \end{itemize}
    So in the first case $u$ separates $(P,Q)$, and in the second it separates $(Q,P)$.
\end{proof}

\paragraph{Representative Cuts.}
Let us now focus on characterizing which vertices \emph{inside} the representative nodes separate $(P,Q)$; these are not treated by~\Cref{lem:rep-nodes-converse} because of condition (ii).
While a representative node $\nu$ must contain \emph{some} vertex separating $(P,Q)$, it may also contain vertices that do not separate $(P,Q)$, or even that do not separate $P,Q$.
The second ingredient of the representative information consists of canonical cuts of the skeleton $\cH_\nu$, that precisely characterize which vertices inside $\nu$ separate $(P,Q)$.
These are called the \emph{representative} cuts, defined in the following lemma:

\begin{lemma}\label{lem:rep-cuts}
    Let $\nu$ be a representative node for $(P,Q)$.
    There exists at most two \emph{representative cuts} of the skeleton $\cH_\nu$ with the following properties:
    \begin{itemize}
        \item Each representative cut is either a canonical cut in the skeleton $\cH_\nu$, or a degenerate cut with one side containing all of $\cH_\nu$.
        Further, the representative cut has a \emph{marked side}.

        \item For every $u \in \nu$, the following holds: $u$ separates $(P,Q)$ if and only if $u$ is mapped to the marked side of some representative cut in $\cH_\nu$.
    \end{itemize}    
\end{lemma}
\begin{proof}
    For a vertex $u$ separating $(P,Q)$, let $\Far(u,s \mid P,Q)$ denote the unique $(u,s)$-minuct which keeps $P$ inside (on the side of $u$), keeps $Q$ outside (on the side of $s$), and contains every other $u,s$-mincut with these properties.
    $\Far(u,s \mid P,Q)$ is well defined due to submodularity (it is the union of all $(u,s)$-mincuts keeping $P$ inside and $Q$ outside).
    The crux of the proof is showing there exists at most two vertices $u_1, u_2 \in \nu$ such that, for any $u \in \nu$, it holds that $u$ separates $(P,Q)$ if and only if $u \in \Far(u_i, s \mid P,Q)$ for some $i \in \{1,2\}$.
    
    Before we show this, let us explain how this yields the required representative cuts by using~\Cref{lem:mu-mincuts}.
    Denote $U_i := \Far(u_i, s \mid P,Q)$.
    If $\nu$ is not contained in $U_i$, then~\Cref{lem:mu-mincuts} ensures that $\nu \cap U_i, \nu \cap \overline{U_i}$ is a valid partition of $\nu$.
    Hence, by~\Cref{thm:skeleton} we can take the representative cut as the canonical cut $C_i$ of the skeleton $\cH_\nu$ corresponding to this valid partition, with its marked side being the side of $\nu \cap U_i$.
    In case $\nu$ is contained in $U_i$, we can just take $C_i$ as the degenerate cut whose marked side is all of $\nu$.
    Thus, for any $u \in \nu$, it holds that $u \in U_i$ iff $u$ is mapped to the marked side of $C_i$.

    We now show the existence of $u_1,u_2$ with the required properties.
    We choose them by the following procedure.
    First, we just take $u_1$ as some arbitrary vertex from $\nu$ which separates $(P,Q)$.
    Next, we choose any $u_2 \in \nu \setminus \Far(u_1, s \mid P,Q)$ which separates $(P,Q)$ (and if there is no such $u_2$, we take just $u_1$ and $u_2$ is undefined).
    Denote $U_i = \Far(u_i, s \mid P,Q)$ for $i \in \{1,2\}$.
    So, to conclude the proof, we need to show that for any $u_3 \in \nu$, it holds that $u_3$ separates $(P,Q)$ iff $u_3 \in U_1 \cup U_2$.

    ($\Leftarrow$)
    Suppose $u_3 \in U_i$ for some $i \in \{1,2\}$.
    Note that $\Far(u_3, s) = \cDdown(\nu) = \Far(u_i, s)$, so the mincut values of $u_3$ and $u_i$ to $s$ are equal.
    Thus, $U_i$ itself is a $(u_3,s)$-mincut which keeps $P$ inside and $Q$ outside.

    ($\Rightarrow$)
    Seeking contradiction, say $u_3$ separates $(P,Q)$ but $u_3 \notin U_1 \cup U_2$.
    Denote $U_3 := \Far(u_3, s \mid P,Q)$.
    By our process for choosing $u_1$ and $u_2$, we get that for any two distinct indices $i,j \in \{1,2,3\}$, either $u_i \notin U_j$ or $u_j \notin U_i$.
    We next derive the contradiction by showing there must be two different indices $i,j$ such that $u_i \in U_j$ and $u_j \in U_i$. 
    Since $\emptyset \neq P \subseteq U_1 \cap U_2 \cap U_3$, by~\Cref{lem:three-mincuts-to-source} there must be two different indices $i,j \in \{1,2,3\}$ such that $u_i \in U_j$, so it remains to show that $u_j \in U_i$.
    Observe that $U_i \cap U_j$ is a $(u_i,s)$-cut, and $U_i \cup U_j$ is a $(u_j,s)$-cut.
    So, by submodularity (\Cref{lem:sub-posi-general}), $U_i \cup U_j$ is a $(u_j,s)$-cut.
    But, as $\Far(u_i,s) = \cDdown(\nu) = \Far(u_j,s)$, $u_i$ and $u_j$ have the same mincut value to $s$.
    Thus, $U_i \cup U_j$ must also be a $(u_i,s)$-mincut.
    Note that $P \subseteq U_i \cup U_j$ and $Q \subseteq \overline{U_i \cup U_j}$, so we get $U_i \cup U_j \subseteq \Far(u_i, s \mid P,Q) = U_i$,
    meaning that $U_j \subseteq U_i$, hence $u_j \in U_i$, and so we have reached the desired contradiction.
\end{proof}

Crucially, each (non-degenerate) representative cut is a canonical cut in its skeleton, which is represented just by its one or two skeleton edges with markings on their endpoints from the marked side (for example, the cut defined by two purple cycle-edges in the Skeleton in \Cref{fig : DAG Skeleton RepInf}).

\paragraph{Putting It All Together.}
We can now finally define $\RepInf(P,Q)$ and prove~\Cref{thm:blackbox}.

\begin{definition}
    For two non-empty sets $P,Q \subseteq V$, the representative information $\RepInf(P,Q)$ consists of the representative nodes
    and cuts for $(P,Q)$ and for $(Q,P)$, given by~\Cref{lem:rep-nodes} and~\Cref{lem:rep-cuts}.
    If there are no vertices that separate $P,Q$, then we just define $\RepInf(P,Q)$ to be a null value indicating this situation.
\end{definition}

\begin{proof}[Proof (of~\Cref{thm:blackbox})]
    Let $\nu_1, \dots, \nu_r$ be all the representative nodes that are stored in some $\RepInf(P_i, Q_i)$.
    By~\Cref{lem:rep-nodes}, every vertex that separates some pair $P_i, Q_i$ must lie in a descendant of some $\nu_j$.
    To identify these, we execute multi-source BFS (breadth first search) from $\nu_1, \dots, \nu_r$ in $\cD$, which takes $O(n)$ time.
    During the BFS, we also determine which nodes are \emph{strict} descendants of at least one $\nu_j$:
    when we scan an edge $(\omega, \nu)$, we mark that $\nu$ is such a strict descendant.
    This is because the BFS has already found some path from some $\nu_j$ to $\omega$, so $\nu_j$ is an ancestor of $\nu$, and $\nu_j \neq \nu$ because $\cD$ has no cycles.
    (It could be that $\nu$ is itself one of the sources, and we still mark it.)
    Next, we run a reverse-BFS from $\mu$ to find all its ancestors, and determine which of them are also strict descendants of some $\nu_j$.
    By~\Cref{lem:rep-nodes-converse}, every vertex in one of these nodes separates some $P_i, Q_i$.
    
    The remaining vertices to consider are those that lie \emph{inside} the representative nodes $\nu_1, \dots, \nu_r$, or more precisely, inside nodes $\nu_j$ such that $\nu_j$ is an ancestor of $\mu$ but not a strict descendant of any other $\nu_{k}$.
    Suppose all of $\nu_1, \dots, \nu_r$ fulfill these conditions (otherwise, we remove those that do not).
    Let us first focus on a specific $\nu_j$.
    Note that $\nu_j$ could be a representative node for many different pairs $P_i, Q_i$, hence there could be many representative cuts corresponding to $\nu_j$ originating from different $\RepInf(P_i,Q_i)$; denote these cuts by $C_1, \dots, C_{\ell_j}$.
    By~\Cref{lem:rep-cuts}, the vertices in $\nu_j$ which separate \emph{some pair} $P_i, Q_i$ are exactly those vertices mapped to a marked side of at least one of $C_1, \dots, C_{\ell_j}$ in the skeleton $\cH_{\nu_j}$.
    Thus, to treat $\nu_j$, we just need to compute the union of the marked sides of $C_1, \dots, C_{\ell_j}$, which takes $O(|\nu_j| + \ell_j)$ time by~\Cref{lem:union-of-cuts-in-skeleton}.
    Thus, the overall time spent to treat vertices inside 
    representative nodes is
    $\sum_{i=j}^r O(|\nu_j| + \ell_j)$.
    As the nodes of $\cD$ partition $V$, we have $\sum_j |\nu_j| = O(n)$.
    Also, each of the $\ell$ pairs $P_i, Q_i$ gives rise to only $O(1)$ representative cuts over all the representative nodes, we have that $\sum_j \ell_j = O(\ell)$.
    Thus, the total running time is $O(n + \ell)$.
\end{proof}

\section{Anchors and Query Translation}\label{sec:translation}

In this section, we formally prove the useful properties of the anchors used for the query translation process, overviewed in~\Cref{sec:translation-overview}.
We first restate the definition of pivots and anchors:

\anchors*

The following lemma formalizes the essential equivalence between the endpoints $x,y$ of the failing edge, and their anchors $z_x, z_y$ with respect to an LCA $\mu$ of $\cD(x),\cD(y)$.

\begin{figure}
  \begin{center}
    \includegraphics[width=0.6\textwidth]{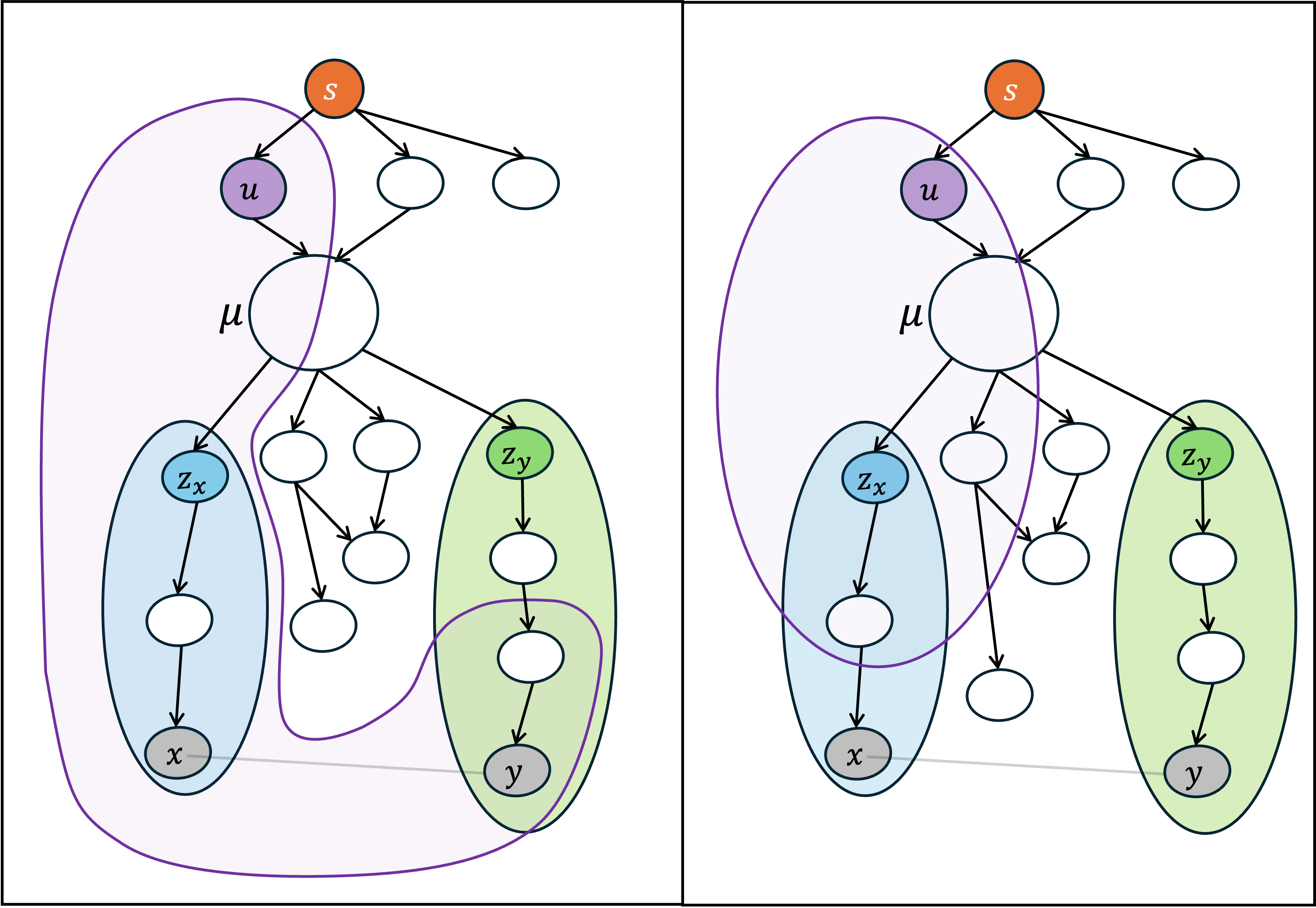}
  \end{center}
  \caption{The figures shows two possible cases -- a $u,s$-mincut separating $z_x,z_y$ that contains both $x,y$ or contains none of $x$ and $y$.
  The proof of Item 2 in~\Cref{lem:translation-DAG} shows that using sub/posimodularity with the farthest $z_x,s$-mincut and $z_y,s$-mincut, one can find another $u,s$-mincut separating $x,y$ having the same intersection with $\mu$.} 
  \label{fig : DAG translation}
\end{figure}

\begin{lemma}\label{lem:translation-DAG}
    Let $e = \{x,y\} \in E$, and let node $\mu$ be an LCA of $\cD(x),\cD(y)$.
    Let $z_x$ and $z_y$ be the anchors of $x$ and of $y$ w.r.t.\ $\mu$.
    Then for every $u \in \cDup(\mu)$, the following hold:
    \begin{enumerate}
        \item If $U$ is a $u,s$-mincut with $x \in U$ and $y \notin U$, then $z_x \in U$ and $z_y \notin U$.
        \item If $U$ is a $u,s$-mincut with $z_x \in U$ and $z_y \notin U$, then there exists another $u,s$-mincut $U'$, such that the following hold:
        (i) $x \in U'$ and $y \notin U'$,
        (ii) $U' \cap \mu = U \cap \mu$, and
        (iii) $u \in U'$ iff $u \in U$.
    \end{enumerate}
\end{lemma}
\begin{proof}
The crux is exploring the relationship of the given $U$ with $Z_x := \cDdown(z_x)$ and $Z_y := \cDdown(z_y)$, which are (the farthest) $(z_x,s)$-mincut and $(z_y,s)$-mincut, 
by using sub/posimodularity.
We first highlight some key properties of $Z_x$ (and symmetrically, of $Z_y$).
Suppose $z_x \neq x$.
Then by~\Cref{def: anchors}, $\cD(z_x)$ is a child of $\mu$ and an ancestor of $\cD(x)$.
Because $\mu$ is an LCA, this means that $\cD(z_x)$ is not an ancestor of $\cD(y)$, and hence it also cannot be an ancestor of $\cD(z_y)$ (as $\cD(z_y)$ is an ancestor of $\cD(y)$ by~\Cref{def: anchors}).
Also, because it is assumed that $\cD(u)$ is an ancestor of $\mu$, $\cD(u)$ cannot be a descendant of $\cD(z_x)$. 
These conclusions can be summarized as follows:
\begin{align}
    z_x \neq x & \quad \implies \quad x \in Z_x \quad \text{and} \quad u,y,z_y \notin Z_x \quad \text{and} \quad \mu \cap Z_x = \emptyset, \label{eq:z_x} \\ 
    z_y \neq y & \quad \implies \quad y \in Z_y \quad \text{and} \quad u,x,z_x \notin Z_y \quad \text{and} \quad \mu \cap Z_y = \emptyset. \label{eq:z_y}
\end{align}
We are now ready to prove the two items of the lemma.

    \emph{Item 1.}
    Let $U$ be a $u,s$-mincut with $x \in U$, $y \notin U$.
    We show that $z_x \in U$ and that $z_y \notin U$ by similar arguments, using posimodularity and submodularity, respectively.

    \begin{itemize}
        \item
        Seeking contradiction, suppose $z_x \notin U$. 
        Then in particular $z_x \neq x$, so~\Cref{eq:z_x} applies.
        Note that $U \setminus Z_x$ is a $u,s$-cut because $|U \cap \{u,s\}| = 1$ and $u,s \notin Z_x$.
        Also, $Z_x \setminus U$ is a $z_x,s$-cut.
        Hence, by posimodularity (\Cref{lem:sub-posi-general}), there cannot be an edge between $U \cap Z_x$ and $\overline{U \cup Z_x}$.
        But $x \in U \cap Z_x$ and $y \in \overline{U \cup Z_x}$, so $e$ is such an edge --- contradiction.

        \item
        Seeking contradiction, suppose $z_y \in U$.
        Then in particular $z_y \neq y$, so~\Cref{eq:z_y} applies.
        Note that $U \cup Z_y$ is a $u,s$-cut because $|U \cap \{u,s\}| = 1$ and $u,s \notin Z_y$.
        Also, $U \cap Z_y$ is a $z_y,s$-cut.
        Hence, by submodularity (\Cref{lem:sub-posi-general}), there cannot be an edge between $U \setminus Z_y$ and $Z_y \setminus U$.
        But $x \in U \setminus Z_y$ and $y \in Z_y \setminus U$, so $e$ is such an edge --- contradiction.
    \end{itemize}

    \emph{Item 2.}
    Let $U$ be a $u,s$-minuct with $z_x \in U$, $z_y \notin \tilde{U}$.
    We construct the desired $U'$ in two steps,
    starting with $U_0 := U$.
    In the first step, we find an intermediate cut $U_1$, ensuring that $x \in U_1$.
    In the second, we find the final cut $U_2 = U'$, which additionally ensures that $y \notin U_2$.
    See~\Cref{fig : DAG translation}.

    \begin{itemize}
        \item
        First Step:
        If $z_x = x$ take $U_1 = U_0$,
        otherwise $U_1 = U_0 \cup Z_x$.
        We claim that $U_1$ is a $u,s$-mincut:
        \begin{itemize}
            \item 
            If $z_x = x$, this is trivial.
            \item 
            If $z_x \neq x$ then~\Cref{eq:z_x} applies.
            Note that $U_0 \cup Z_x$ is a $u,s$-cut because $|U_0 \cap \{u,s\} | = 1$ and $u,s \notin Z_x$.
            Also, $U_0 \cap Z_x$ is a $z_x,s$-cut.
            So the claim follows by submodularity (\Cref{lem:sub-posi-general}).
        \end{itemize}
        Observe that
        (a) $x \in U_1$, $z_y \notin U_1$,
        (b) $\mu \cap U_1 = \mu \cap U_0$, and
        (c) $u \in U_1$ iff $u \in U_0$.
        Indeed, this is trivial if $z_x = x$, and otherwise this is by~\Cref{eq:z_x}.
    
        \item
        Second Step:
        If $z_x = x$ take $U_2 = U_1$, otherwise $U_2 = U_1 \setminus Z_y$.
        We claim that $U_2$ is a $u,s$-mincut:
        \begin{itemize}
            \item If $z_y = y$, this is clear by the previous step.
            \item If $z_y \neq y$, then~\Cref{eq:z_y} applies.
            Note that $U_1 \setminus Z_y$ is a $u,s$-cut because $|U_1 \cap \{u,s\}|=1$ and $u,s \notin Z_y$.
            Also, $Z_y \setminus U_1$ is a $z_y,s$-cut (as $z_y \notin U_1$ by the first step).
            So the claim follows by posimodularity (\Cref{lem:sub-posi-general}).
        \end{itemize}
        Finally, we observe that (i), (ii) and (iii) for $U_2$ hold, because their counterpart properties (a), (b) and (c) hold for $U_1$.
        Indeed, this is trivial if $z_y = y$, and otherwise this follow immediately by~\Cref{eq:z_y}.
    \end{itemize}
\end{proof}
\begin{figure}
  \begin{center}
    \includegraphics[width=\textwidth]{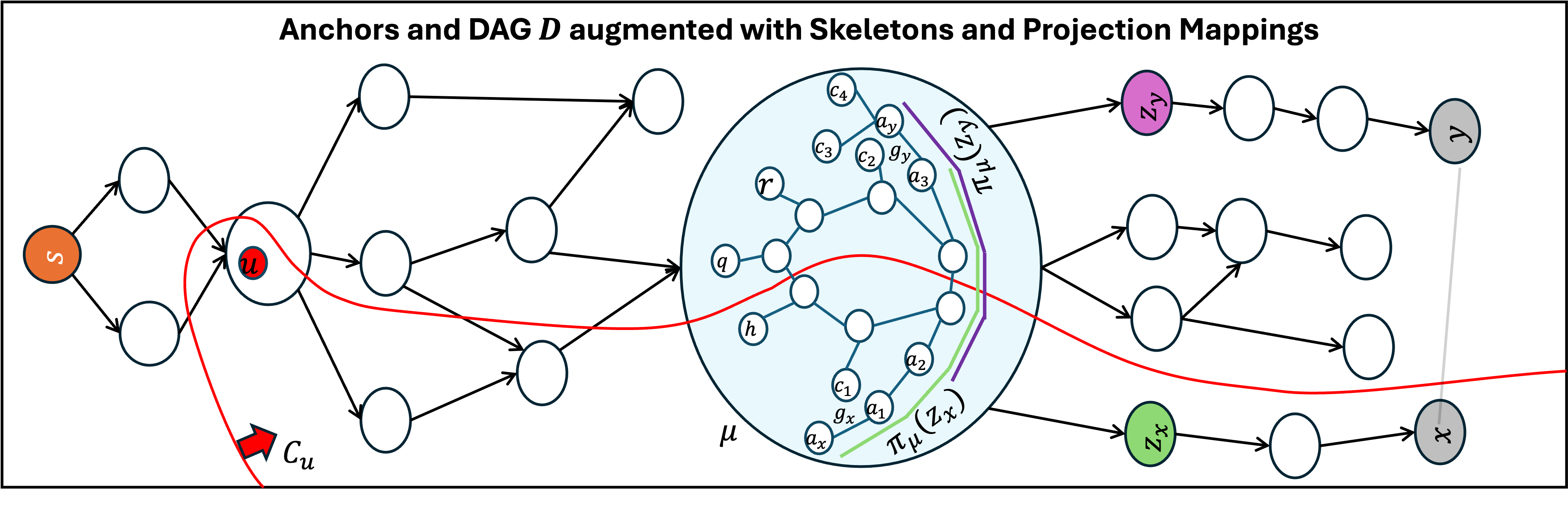}
  \end{center}
  \caption{Vertices $z_x$ and $z_y$ are the anchors of $x$ and $y$ respectively for $\mu\in \LCA(x,y)$. The projection of $\pi_\mu(z_x)$ (resp., $\pi_\mu(z_y)$) is represented by the green (resp., purple) polyline on the skeleton inside node $\mu$.
  \Cref{cor:projection-of-edge-skeleton-translation} shows that they form a prefix and suffix of the edge projection $\pi_\mu (e)$, where $e = \{x,y\}$.} 
  \label{fig : DAG anchors Projection mapping}
\end{figure}

Finally, as explained in~\Cref{sec:translation-overview}, the projections of the anchors in the skeleton of the LCA $\mu$ play a crucial role, as they enable us to find the projection of the edge $e$.
This is obtained by~\Cref{cor:projection-of-edge-skeleton-translation}, which is derived from the equivalence of anchors in~\Cref{lem:translation-DAG} and our carcass translation lemma (\Cref{lem:skeleton-translation}); see~\Cref{fig : DAG anchors Projection mapping}.

\begin{corollary} \label{cor:projection-of-edge-skeleton-translation}
    Let $e = \{x,y\} \in E$, and let node $\mu$ be an LCA of $\cD(x),\cD(y)$.
    Let $z_x$ and $z_y$ be the anchors of $x$ and of $y$ w.r.t.\ $\mu$.
    Suppose that the following conditions hold:
    \begin{itemize}
        \item[(i)] $c(\cDdown(\mu)) = \lambda_\mu$. 
        \item[(ii)] $\phi_\mu (z_x) \neq \phi_\mu (z_y)$. That is, $z_x,z_y$ are not $\mu$-equivalent. %
    \end{itemize}
    Then the projection path $\pi_\mu(e)$ is defined, and it has $\pi_\mu (z_x)$ as a prefix and $\pi_\mu (z_y)$ as a suffix, where $\pi_\mu (e)$ is considered in the direction from its $x$-endpoint to its $y$-endpoint.
\end{corollary}
\begin{proof}
    This follows by~\Cref{lem:skeleton-translation} with $S = \mu$, $x' = z_x$, $y' = z_y$,
    once we show that~\ref{prop:P1} and~\ref{prop:P2} hold.
    \ref{prop:P2} is just (ii).
    For \ref{prop:P1}, let $S_1,S_2$ be a valid partition of $\mu$.
    We should show:
    there is an $(S_1,S_2)$-mincut $U$ with $z_x \in U$, $z_y \notin U$ iff there is an $(S_1,S_2)$-mincut $U'$ with $x \in U'$, $y \notin U'$.

    First, suppose $U'$ exists.
    Take any $u \in \mu$ lying in the opposite side to $s$ of the cut $U'$.
    Then $U'$ is a $u,s$-cut, and because $\cDdown(\mu) = \Far(u,s)$, by (i) it is in fact a $u,s$-mincut.
    Thus, by~\Cref{lem:translation-DAG}(1), $z_x \in U'$ and $z_y \notin U'$, so we just take $U = U'$.
    Now, suppose $U$ exists.
    Then $U$ must be a $u,s$-mincut for some $u \in \mu$ by a similar argument as in previous direction.
    Thus, by~\Cref{lem:translation-DAG}(2), there is a $u,s$-mincut $U'$ with $x \in U'$, $y \notin U'$, and $U' \cap \mu = U \cap \mu$.
    The last condition means that the $u,s$-mincut $U'$ is also $(S_1, S_2)$-cut, but by (i) this means $U'$ is $(S_1,S_2)$-mincut, as needed.
\end{proof}

\section{Mincuts Splitting the LCA}\label{sec:splitting-the-LCA}

This section provides the data structure for Query~\ref{query-Q1}, restated below; See the overview in~\Cref{sec:general-case-overview}.

\begin{enumerate}[label=\textbf{(Q1)}]
    \item Given $e = \{x,y\} \in E$ and an LCA $\mu$ of $\cD(x),\cD(y)$,
    report all $u \in \cDup(\mu)$ such that some $u,s$-mincut separates $x,y$ and splits $\mu$.
\end{enumerate}
Note that by the equivalence of  $x,y$ to their anchors $z_x,z_y$ w.r.t.\ $\mu$ from~\Cref{lem:translation-DAG}, the condition in Query~\ref{query-Q1} is equivalent to ``some $u,s$-mincut separates $z_x,z_y$ and splits $\mu$''; this new condition would be the one that we will use.
Our goal in this section is to show:
\begin{lemma}\label{lem:q1-data-structure}
    There is a data structure for Query~\ref{query-Q1} with $O(n)$ space and $O(n)$ query time.
\end{lemma}

Our first step is to formally define the $A_\mu (\cdot)$ sets, which are crucial to our query algorithm, as mentioned in the overview (\Cref{sec:general-case-overview}).
\begin{definition}\label{def:A-sets}
    Let $\mu$ be a node in $\cD$, and let $z \in V$ be a stretched vertex w.r.t.\ $\mu$.
    Let $a_0 (z)$ and $a_1 (z)$ be the endpoints of the projection path $\pi_\mu (z)$ in $\cH_\mu$, where $a_0$ has smaller identifier than $a_1$.
    Let $g_0 (z), g_1 (z)$ be the edges of $\pi_\mu (z)$ incident to $a_0, a_1$ respectively.
    Let $\mu_0 (z) := \mu(a_0 (z), g_0 (z))$ and $\mu_1 (z) := \mu(a_1 (z), g_1 (z))$ as defined in~\Cref{def:associated-terminals}.
    We then define:
    \begin{align*}
        A_\mu^{(0)} (z) &:= \{ u \in \cDup(\mu) \mid \text{$u$ separates $\mu_0 (z) \cup \{z\}, \mu_1 (z)$} \}, \\
        A_\mu^{(1)} (z) &:= \{ u \in \cDup(\mu) \mid \text{$u$ separates $\mu_0 (z), \mu_1 (z) \cup \{z\}$} \}.
    \end{align*}
\end{definition}

\subsection{Construction}

We now formally describe the information stored by our data structure for Query~\ref{query-Q1}.

First, we store the $O(n)$ global data structure for the representative framework of~\Cref{thm:blackbox}, which consists of the DAG $\cD$ and the skeletons $\cH_\mu$ for each node $\mu$ of $\cD$.
We augment each skeleton edge $g = \{a,b\}$ in $\cH_\mu$ with its $O(1)$ representative information $\RepInf(g) := \RepInf(\mu(a,g),\mu(b,g))$ (see~\Cref{def:associated-terminals} for the definition of $\mu(a,g)$, and $\mu(b,g)$).

Next, for each node $\mu$ in $\cD(\mu)$, we store the values $\lambda_\mu$ and $c(\cDdown(\mu))$, and the following $O(1)$ information for every $z \in \{z_\nu \mid \text{$\nu$ child of $\mu$}\} \cup \mu$, consumes a total of $O(|\mu| + \operatorname{outdeg}_\cD (\mu))$ space for $\mu$:
\begin{itemize}
    \item The $\mu$-equivalence class number $\phi_\mu (z)$, %
    and the $\sigma_\mu (z)$ number of~\Cref{lem:sigma-numbers}.
    \item The endpoints $a_0 (z)$ and $a_1 (z)$ of the projection $\pi_\mu (z)$ (see~\Cref{def:A-sets}).
    Note that these endpoints are distinct if $z$ is stretched, and equal otherwise.
    \item 
    The representative information for $A_\mu^{(0)}(z)$ and to $A_\mu^{(1)}(z)$, namely $\RepInf(\mu_0(z) \cup \{z\}, \mu_1 (z))$ and $\RepInf(\mu_0(z), \mu_1(z) \cup \{z\})$.
    These are defined only if $z$ is stretched (see~\Cref{def:A-sets}).
\end{itemize}

As the nodes of $\cD$ partition the vertices in $V$, and because $\sum_\mu \operatorname{outdeg}_\cD (\mu)$ counts each of the $O(n)$ edges of $\cD$ exactly once, the total space occupied by the data structure is $O(n)$.

\begin{figure}
  \begin{center}
    \includegraphics[width=0.8\textwidth]{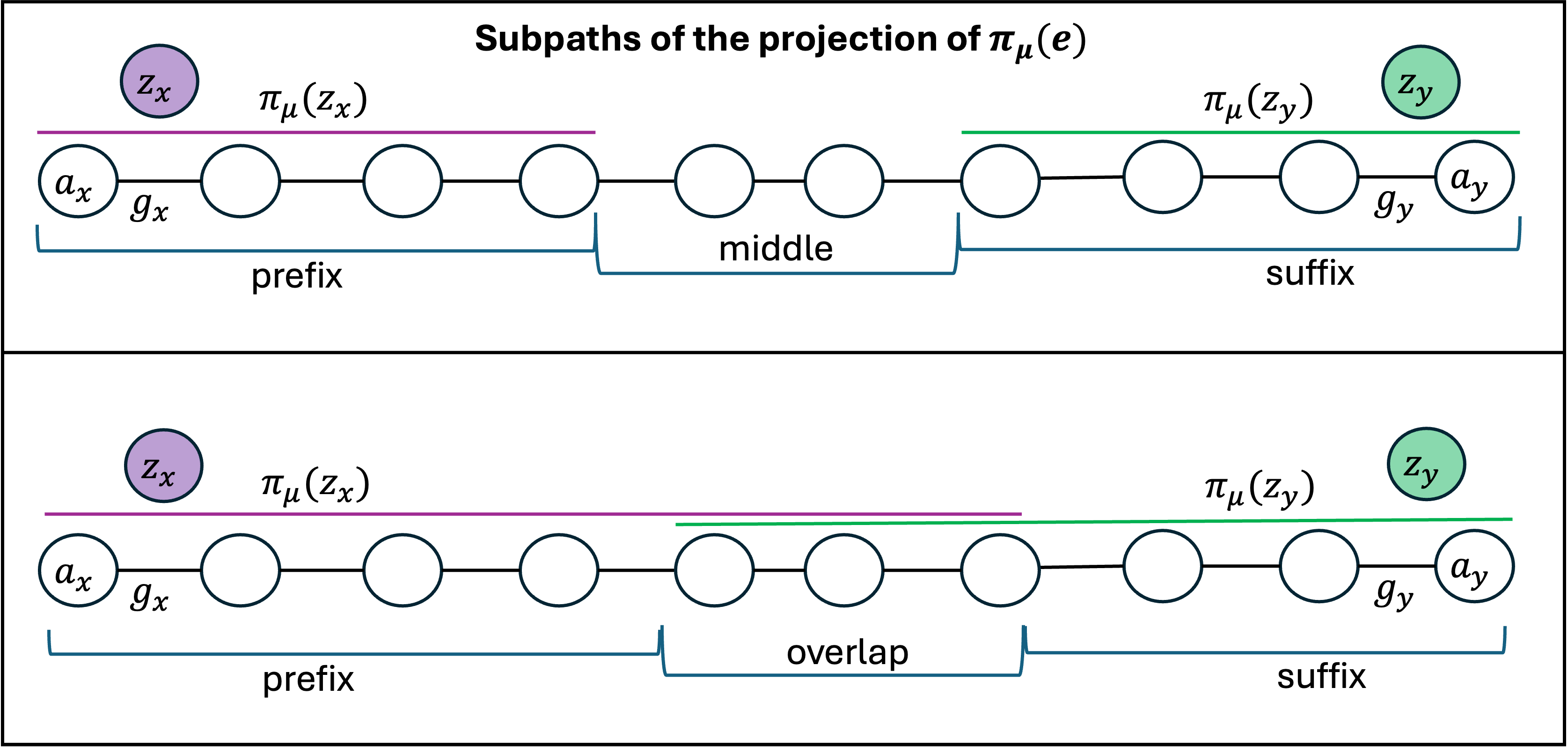}
  \end{center}
  \caption{Partitioning the edges of $\pi_\mu(e)$ into four subpaths for answering queries.} 
  \label{fig : subpaths of projection of an edge}
\end{figure}

\subsection{Query Algorithm}\label{sec:query-alg-Q1}\label{sec:splitting-mu}
We now give the algorithm to answer Query~\ref{query-Q1}, given edge $e = \{x,y\} \in E$ and LCA $\mu$ of $\cD(x),\cD(y)$.

\begin{description}
    \item[Step 0:]
    Find the anchors $z_x,z_y$ of $x,y$ w.r.t.\ $\mu$, which takes $O(n)$ time by searching $\cD$.
    Note that by their definition (\Cref{def: anchors}), $z_x,z_y \in \{z_\nu \mid \text{$\nu$ child of $\mu$}\} \cup \mu$ .

    \item[Step 1:]
    If $\lambda_\mu \neq c(\cDdown(\mu))$ or $\phi_\mu (z_x) = \phi_\mu (z_y)$:
    Halt without reporting any vertex.

    \item[Step 2:]
    Use the stored endpoints of $\pi_\mu (z_x)$ and $\pi_\mu (z_y)$ as the unique proper paths between their endpoints (or, if $z_x$ and/or $z_y$ are non-stretched, the projections are nodes) in $\cH_\mu$.
    Next, compute $\pi_\mu (e)$ as the unique proper path that starts with $\pi_\mu (z_x)$ and ends with $\pi_\mu (z_y)$; this path indeed exists and is equal to $\pi_\mu (e)$ by~\Cref{cor:projection-of-edge-skeleton-translation}, since we did not halt in Step 1.
    These computations on $\cH_\mu$ take $O(|\mu|) \leq O(n)$ time, and let us to determine the edges in each of the following subpaths of $\pi_\mu (e)$ (see~\Cref{fig : subpaths of projection of an edge}):
    \begin{enumerate}[label=(\roman*)]
        \item 
        middle part $\pi_\mu (e) \setminus (\pi_\mu (z_x) \cup \pi_\mu (z_y))$,
        \item 
        overlap $\pi_\mu (z_x) \cap \pi_\mu (z_y)$,
        \item 
        prefix $\pi_\mu (z_x) \setminus \pi_\mu (z_y)$,
        \item 
        suffix $\pi_\mu (z_y) \setminus \pi_\mu (z_x)$.
    \end{enumerate}
    
    For each $P$ among (i)-(iv), find the set $\{ u \in \cDup(\mu) \mid \text{$u$ separates some edge $g$ of $P$}\}$ by applying the representative framework from~\Cref{thm:blackbox} with $\{\RepInf(g)\}_{g \in P}$.
    As $P$ has at most $O(|\mu|) \leq O(n)$ edges, this takes  $O(n)$ time by~\Cref{thm:blackbox}.
    
    \item[Step 3:]
    Determine the $x$-endpoint $a_x$ and the $y$-endpoint $a_y$ of $\pi_\mu (e)$, and hence also the edges $g_x$ and $g_y$ of $\pi_\mu(e)$ adjacent to $a_x$ and to $a_y$ respectively (where possibly $g_x = g_y$), as follows.

    \begin{itemize}
        \item If $\pi_\mu (z_x) \neq \pi_\mu (y)$, then either the $x$-endpoint lies uniquely on the prefix (iii), or the $y$-endpoint lies uniquely on the suffix (iv), so we are done.

        \item If $\pi_\mu (z_x) = \pi_\mu (y)$, then compare $\sigma_\mu (z_x)$ and $\sigma_\mu (z_y)$.
        If $\sigma_\mu (z_x) < \sigma_\mu (z_y)$, then the $x$-endpoint of $\pi_\mu (e)$ is the one with smaller identifier; otherwise it is the one with larger identifier.
        The correctness of this step is proved in~\Cref{lem:identifying-x-endpoint}.
    \end{itemize}

    If $z_x$ is stretched, let $i \in \{0,1\}$ be such that $a_x = a_i (z_x)$, and use $\RepInf(\mu_i (z_x) \cup \{x\}, \mu_{1-i} (z_x))$ with the representative framework of~\Cref{thm:blackbox} to compute $A_\mu (z_x) := A_\mu^{(i)} (z_x)$ within $O(n)$ time.
    If $z_x$ is non-stretched, let $A_\mu (z_x) := \emptyset$.
    Similarly, if $z_y$ is stretched, let $j \in \{0,1\}$ be such that $a_y = a_j (z_y)$, and compute $A_\mu (z_y) := A_\mu^{(j)}(z_y)$ is a similar fashion, and if $z_y$ is non-stretched then $A_\mu (z_y) := \emptyset$.
    This takes $O(n)$ time by~\Cref{thm:blackbox}.
        
    \item[Step 4:]
    Find all vertices $u \in \cDup(\mu)$ within $O(n)$ time by searching $\cD$.
    Then, for each $u \in \cDup(\mu)$, decide within $O(1)$ time whether to report $u$ as part of the output for Query~\ref{query-Q1} as follows:
    \begin{enumerate}[label=(\alph*)]
        \item If $u$ separates an edge of the middle part (i): Report $u$.

        \item Else, if $u$ separates an edge of the overlap (ii): Report $u$ iff $u \in A_\mu (z_x) \cap A_\mu (z_y)$.

        \item Else, if $u$ separates an edge of the prefix (iii) or the suffix (iv): Report $u$ iff $u \in A_\mu (z_x) \cup A_\mu (z_y)$.

        \item Else: Do not report $u$.
    \end{enumerate}   
\end{description}
Thus, our data structure answers Query~\ref{query-Q1} within $O(n)$ time.

\subsection{Correctness}
We first need to prove that Step 1 is correct:
\begin{lemma}
    If $\lambda_\mu \neq c(\cDdown(\mu))$ or $\phi_\mu (z_x) = \phi_\mu (z_y)$, then for every vertex $u$, there is no $u,s$-mincut that splits $\mu$ and separates $z_x,z_y$.
\end{lemma}
\begin{proof}
    Suppose there is such a $u,s$-mincut $U$. Without loss of generality, say $u \in U$.
    Then, by~\Cref{lem:mu-mincuts}, $W := U \cap \cDdown(\mu)$ is a $\mu$-mincut, and we have that $\lambda_\mu = c(\cDdown(\mu))$.
    We also see that $\phi_\mu (z_x) \neq \phi_\mu (z_y)$, i.e. that $z_x, z_y$ are separated by some $\mu$-mincut:
    Indeed, as $z_x,z_y$ both belong to $\cDdown(\mu)$ and are separated by $U$, they are also separated by $W$.
\end{proof}
So, for the rest of the correctness proof (i.e., in all of the following lemmas), we will implicitly assume that the query algorithm did not stop in Step 1, i.e., that $\lambda_\mu = c(\cDdown(\mu))$ and $\phi_\mu (z_x) \neq \phi_\mu (z_y)$.
Thus, \Cref{cor:projection-of-edge-skeleton-translation} is assumed to apply: $\pi_\mu(e)$ is defined, and has $\pi_\mu (z_x)$ as a prefix and $\pi_\mu (z_y)$ as a suffix in the direction from the $x$-endpoint $a_x$ to its $y$-endpoint $a_y$.
Next, we need to prove that Step 3 correctly identifies the $x$-endpoint and the $y$-endpoint.
This is clear in case $\pi_\mu (z_x) \neq \pi_\mu (z_y)$.
For the other case, this is due to the properties of the $\sigma_\mu (\cdot)$ numbers from~\Cref{lem:sigma-numbers}:

\begin{lemma}\label{lem:identifying-x-endpoint}
    Suppose $\pi_\mu (z_x) = \pi_\mu (z_y)$.
    Then, the $x$-endpoint of $\pi_\mu (e)$ is the one with smaller identifier if and only if $\sigma_\mu (z_x) < \sigma_\mu (z_y)$.
\end{lemma}
\begin{proof}
    ($\Leftarrow$) By~\Cref{lem:sigma-numbers}, there is a $\mu$-mincut $W$ with $z_x \in W$, $z_y \notin \overline{W}$ and a canonical cut $C$ of $\cH_\mu$ which keeps the endpoint of smaller identifier with $\mu \cap W$ and the one with larger identifier with $\mu \cap \overline{W}$.
    Choose any vertex $w \in \mu$ from the side without $s$ of the cut $W$.
    Then $c(W) = \lambda_\mu = c(\cDdown(\mu)) = c(\Far(w,s))$, so $W$ is a $w,s$-mincut.
    By~\Cref{lem:translation-DAG}, there exists some $w,s$-mincut $W'$ such that $x \in W'$, $y \in \overline{W'}$ and $W' \cap \mu = W \cap \mu$.
    Thus, $W'$ is a $\mu$-mincut inducing the same valid partition as $W$, and $e$ contributes to $W'$ with $x \in W'$ and $y \in \overline{W'}$.
    By the unidirectionality of the edge projection $\pi_\mu (e)$ (\Cref{thm:proj-edge}), this means that the smaller-identifier endpoint, which $C$ keeps with $W \cap \mu = W' \cap \mu$, must be the $x$-endpoint.

    ($\Rightarrow$) Suppose that $\sigma_\mu (z_x) \geq \sigma_\mu (y)$.
    Note that equality is impossible by~\Cref{lem:sigma-numbers}.
    So, by the exact same proof as the previous direction with $x,y$ swapped, the $y$-endpoint must have the smaller identifier.
\end{proof}

We now arrive at the main part of the correctness analysis: proving that Step 4 is correct.
From now on, we fix $u \in \cDup(\mu)$, and our goal is to show that our query algorithm determines whether to report $u$, namely, $u$ is reported precisely when there exists a $u,s$-mincut which separates $z_x,z_y$ and splits $\mu$.

Before treating each substep separately, we show a ``helper lemma'', useful for several of them.

\begin{lemma}\label{lem:finiding-the-separated-edge}
    Suppose there exists a $u,s$-mincut which splits $\mu$ and separates $z_x,z_y$.
    Then, there exists a $u,s$-mincut $U$ and an edge $g = \{a,b\}$ of $\pi_\mu (e)$ such that the following hold:
    \begin{itemize}
        \item Node $a$ (resp., $b$) is the endpoint of $g$ closer to the $x$-endpoint $a_x$ (resp., the $y$-endpoint $a_y$) of $\pi_\mu (e)$.

        \item The cut $U$ keeps $\{z_x\} \cup \mu(a,g)$ on one side and $\{z_y\} \cup \mu(b,g)$ on the other.
    \end{itemize}
\end{lemma}
\begin{proof}
    By~\Cref{lem:translation-DAG}, there must also be a $(u,s)$-mincut $U$ which splits $\mu$ and keeps $\{x,z_x\}$ and $\{y,z_y\}$ on different sides.
    By~\Cref{lem:mu-mincuts}, $W := U \cap \cDdown(\mu)$ is a $\mu$-mincut, and it also separates $\{x,z_x\}$ from $\{y,z_y\}$ as $x,z_x,y,z_y \in \cDdown(\mu)$.
    Let $C$ be the canonical cut of $\cH_\mu$ corresponding to the valid partition $\mu \cap W, \mu \cap \overline{W}$.
    Since the latter distinguishes $e = \{x,y\}$, $C$ must have an $g$ from $\pi_\mu (e)$ by~\Cref{thm:proj-edge}.
    Let $a$ and $b$ be the endpoint of $g$ closer to $a_x$ or to $a_y$ on $\pi_\mu (e)$, respectively.
    By the unidirectionality property of $\pi_\mu (e)$ in~\Cref{thm:proj-edge}, $W$ must keep $\mu(a,g)$ with $x$ (and $z_x$), and $\mu(b,g)$ with $y$ (and $z_y$).
    So, $W$ separates $\{z_x\} \cup \mu(a,g)$ from $\{z_y\} \cup \mu(a,g)$.
    But these are both subsets of $\cDdown(\mu)$, hence this is also true for $U$.
\end{proof}

Our proofs also repeatedly use the following observations, which follow immediately from~\Cref{lem:containment-of-associated-terminals} and the definitions of $\mu_i (z_x), \mu_{1-i}(z_x), \mu_j (z_y), \mu_{1-j}(z_y)$.
\begin{observation} \label{obs:terminal-partition}
The following hold (see~\Cref{fig : terminals containments}):
\begin{itemize}
    \item If $\mu_0 (z_x), \mu_1 (z_x)$ are defined (i.e., $z_x$ is stretched), then $\mu_i (z_x) = \mu(a_x,g_x)$ and $\mu_{1-i} (z_x) \supseteq \mu(a_y, g_y)$.

    \item If $\mu_0 (z_y), \mu_1 (z_y)$ are defined (i.e., $z_y$ is stretched), then $\mu_j (z_y) = \mu(a_y, g_y)$ and $\mu_{1-j} (z_y) \supseteq \mu(a_x, g_x)$.

    \item Let $g = \{a,b\}$ be an edge in $\pi_\mu (e)$, where $a$ is closer to the $x$-endpoint $a_x$ and $b$ is closer to the $y$-endpoint $a_y$ on $\pi_\mu (e)$.
    Then $\mu(a,g) \supseteq \mu(a_x, g_x) = \mu_i (z_x)$ and $\mu(b,g) \supseteq \mu(a_y, g_y) = \mu_j (z_y)$ (where the last equalities hold whenever $\mu_i (z_x)$ and $\mu_j (z_y)$ are defined).
    Furthermore:
    \begin{itemize}
        \item If $g$ is in $\pi_\mu (z_x)$, then $\mu(g,a) \supseteq \mu_i (z_x)$ and $\mu(g,b) \supseteq \mu_{1-i}(z_x)$.

        \item If $g$ is in $\pi_\mu (z_y)$, then $\mu(g,b) \supseteq \mu_j (z_y)$ and $\mu(a,g) \supseteq \mu_{j-1}(z_y)$.
    \end{itemize}
\end{itemize}
\end{observation}

\begin{figure}
  \begin{center}
    \includegraphics[width=\textwidth]{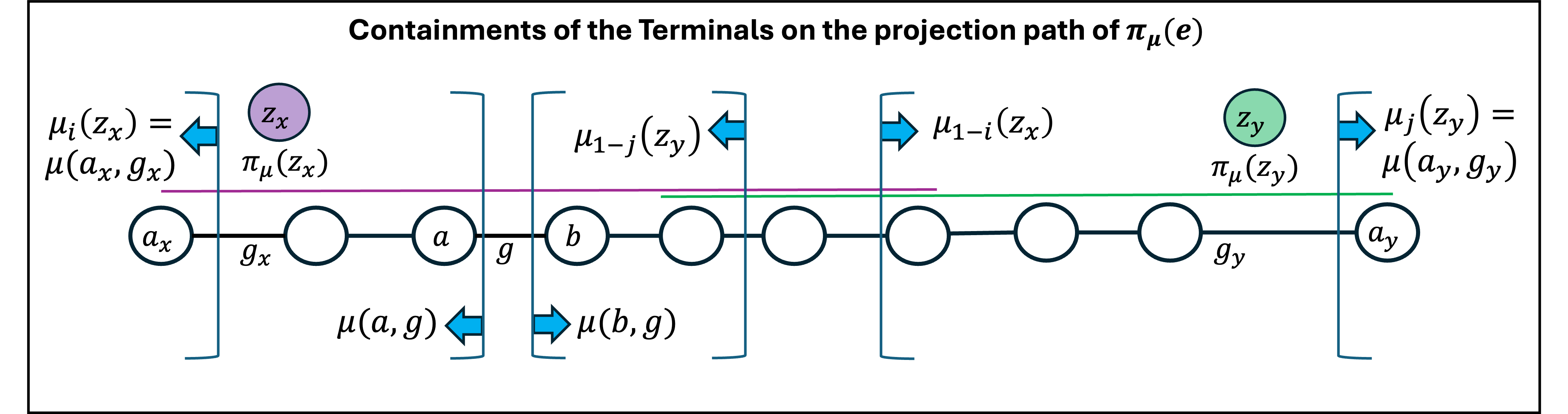}
  \end{center}
  \caption{Depicting Observation~\ref{obs:terminal-partition}.} 
  \label{fig : terminals containments}
\end{figure}

Finally, the following four lemmas correspond to the correctness of each of the four substeps 4a-4b.

\begin{lemma}[4a]
    Suppose $u$ separates an edge of the middle part $\pi_\mu (e) \setminus (\pi_\mu (z_x) \cup \pi_\mu (z_y))$.
    Then, there exists a $u,s$-mincut which splits $\mu$ and separates $z_x,z_y$.
\end{lemma}
\begin{proof}
    By the assumption, there is an edge $g = \{a,b\}$ in the middle part $\pi_\mu (e) \setminus (\pi_\mu (z_x) \cup \pi_\mu (z_y))$ and a $(u,s)$-mincut $U$ keeping $\mu(a,g)$ on one side and $\mu(b,g)$ on the other.
    So $U$ splits $\mu$, and we next prove it also separates $z_x,z_y$.
    By~\Cref{lem:mu-mincuts}, $W := U \cap \cDdown(\mu)$ is a $\mu$-mincut, so $\mu \cap W, \mu \cap \overline{W}$ is a valid partition keeping $\mu(a,g)$ and $\mu(b,g)$ on different sides.
    Thus, by~\Cref{lem:associated-terminals}, the canonical cut $C$ corresponding to this valid partition contains the edge $g$, so $C$ keeps all $\pi_\mu (z_x)$ on one side and all of $\pi_\mu (z_y)$ on the other side.
    Hence, by the properties of vertex projections in \Cref{thm:proj-vertex}, every $\mu \cap W, \mu \cap \overline{W}$-mincut must separate $z_x,z_y$.
    In particular, $W$ itself must separate $z_x, z_y$, and because $z_x, z_y \in \cDdown(\mu)$, so does $U$.
\end{proof}

\begin{lemma}[4b]\label{lem:4b}
    Suppose $u$ separates an edge of the overlap $\pi_\mu (z_x) \cap \pi_\mu (z_y)$.
    Then, it holds that $u \in A_\mu (z_x) \cap A_\mu (z_y)$ if and only if there exists a $u,s$-mincut which splits $\mu$ and separates $z_x,z_y$.
\end{lemma}
\begin{proof}
    We will show the existence of $W$ which is both a $\mu$-mincut and a $w,s$-mincut for some $w \in \mu$, with
    \[
    \{z_x\} \cup \mu_i (z_x) \cup \mu_{1-j}(z_y) \subseteq W \quad \text{and} \quad \{z_y\} \cup \mu_{1-i}(z_x) \cup \mu_j (z_y) \subseteq \overline{W}.
    \]
    We defer this argument to the end of the proof, and start by showing how $W$ is used to prove both directions.
    See~\Cref{fig : forward direction} for an illustration of the proof.
\begin{figure}
  \begin{center}
    \includegraphics[width=.9\textwidth]{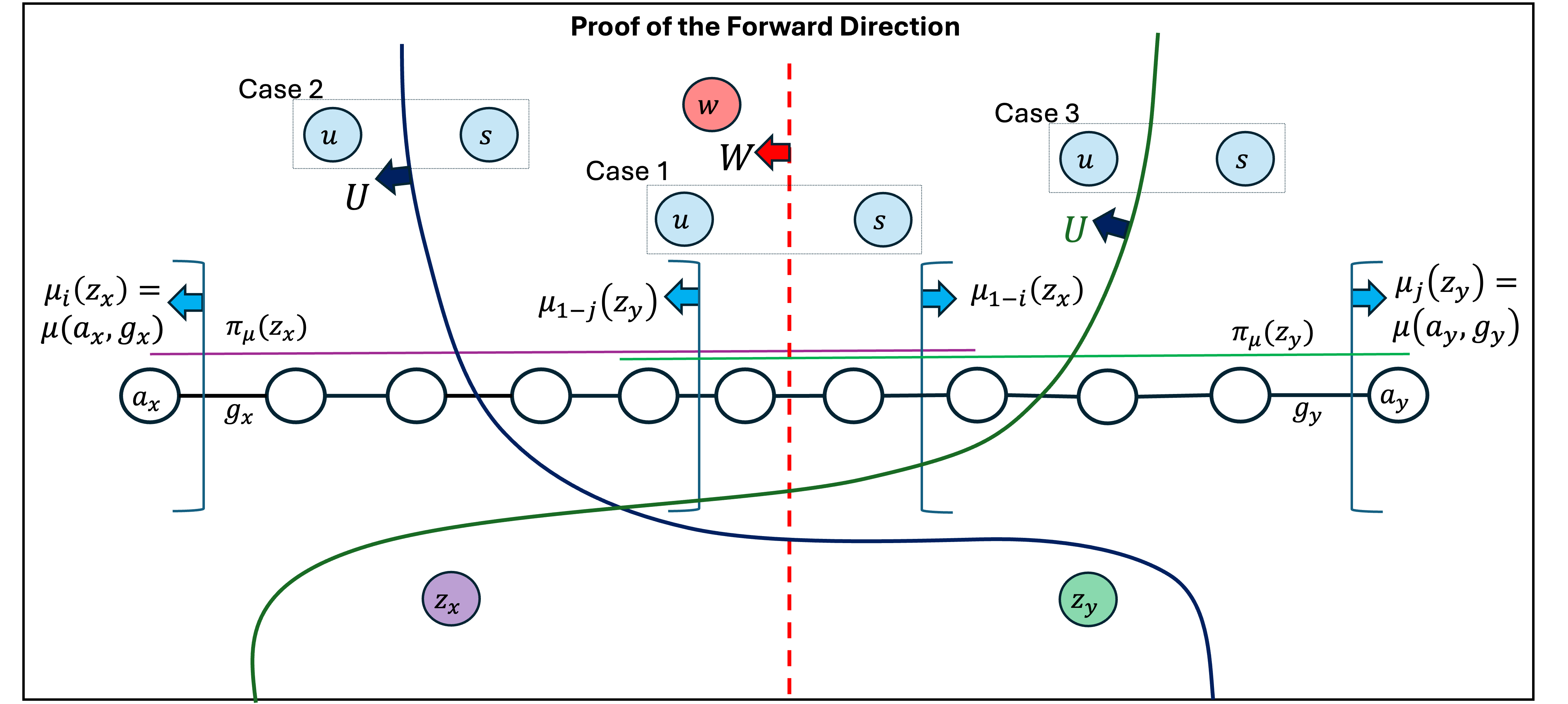}
    \includegraphics[width=.9\textwidth]{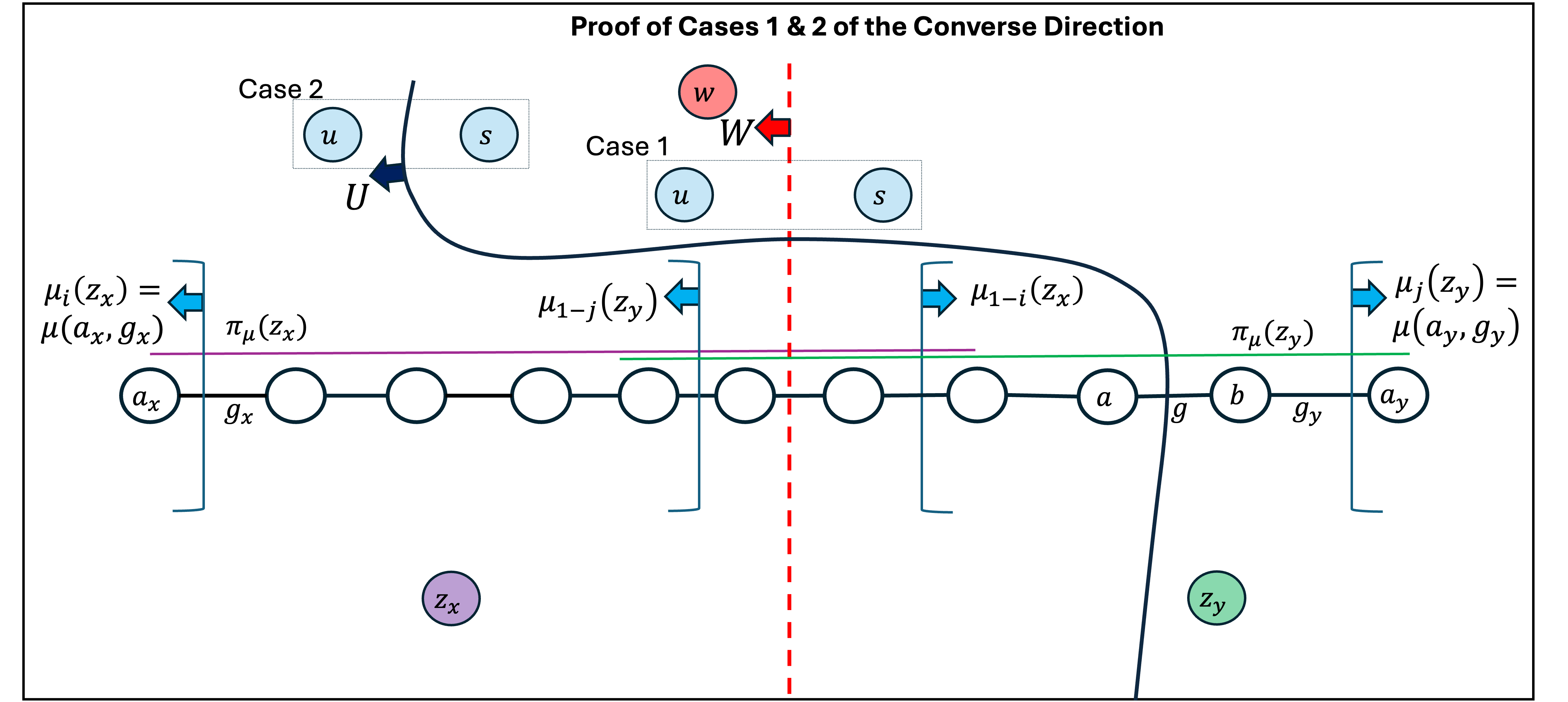}

    \includegraphics[width=.9\textwidth]{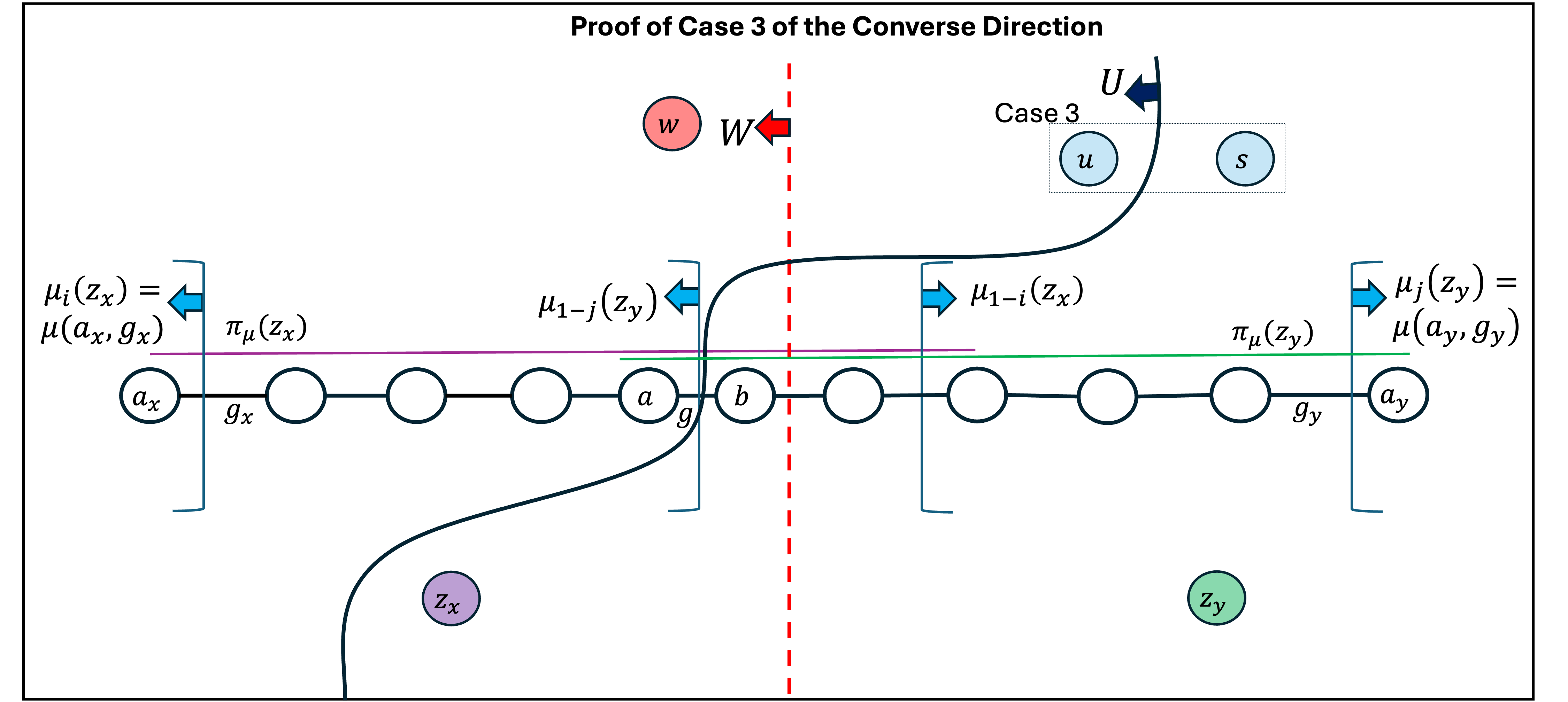}

  \end{center}
  \caption{
  Illustrating the proof of~\Cref{lem:4b}.
  } 
  \label{fig : forward direction}
\end{figure}

    ($\Rightarrow$) 
    We split into cases:
    \begin{itemize}
        \item  
        Case 1: $W$ separates $u,s$.
        Then $W$ is a $u,s$-cut, so $c(W) \geq c(\Far(u,s))$.
        But, because $u \in \cDup(\mu)$, we have that $w \in \mu\subseteq \cDdown(u) = \Far(u,s)$, meaning that $\Far(u,s)$ is a $w,s$-cut, and hence $c(W) \leq c(\Far(u,s))$.
        We conclude that $c(W) = c(\Far(u,s))$, so $W$ itself is a $u,s$-mincut, and it separates $z_x,z_y$ and splits $\mu$ as required.
         
        \item 
        Case 2: $u,s \in W$.
        Since $u \in A_\mu (z_x)$, there exists a $u,s$-mincut $U$ with $\{z_x\} \cup \mu_i (z_x) \subseteq U$ and $\mu_{1-i}(z_x) \subseteq \overline{U}$.
        Note that $U \cap W$ is a $u,s$-cut, since $|U \cap \{u,s\}| = 1$ and $u,s \in W$.
        Also $U \cup W$ is a $\mu$-cut, as it separates $\mu_i (z_x)$ from $\mu_{1-i}(z_x)$.
        Thus, by submodularity (\Cref{lem:sub-posi-general}), $U \cap W$ is a $u,s$-mincut, and we have $\{z_x\} \cup \mu_i (z_x) \subseteq U \cap W$ and $\{z_y\} \cup \mu_{1-i}(z_x) \subseteq \overline{U \cap W}$,
        namely, it separates $z_x,z_y$ and splits $\mu$.

        \item 

        Case 3: $u,s \in \overline{W}$.
        Symmetric to Case 2, using $u \in A_\mu (z_y)$ and considering $\overline{W}$ instead of $W$.
    \end{itemize}

    ($\Leftarrow$) 
    We show that $u \in A_\mu (z_x)$, and the proof for $u \in A_\mu (z_y)$ is symmetric.
    We split to cases:

    \begin{itemize}
        \item Case 1: $W$ separates $u,s$.
        Then as in Case 1 of the previous direction, $W$ is a $u,s$-mincut, and it separates $\{z_x\} \cup \mu_i (z_x)$ from $\mu_{1-i} (z_x)$, hence $u \in A_\mu (z_x)$.

        \item Case 2: $u,s \in W$.
        In this case, consider the $u,s$-mincut $U$ and the edge $g = \{a,b\}$ guaranteed by~\Cref{lem:finiding-the-separated-edge}, so in particular $\{z_x\} \cup \mu_i (z_x) \subseteq \{z_x\} \cup  \mu(a,g) \subseteq U$.
        We have that $\{z_x\} \cup \mu_i (z_x) \subseteq U \cap W$ and $\mu_{1-i} (z_x) \subseteq \overline{W \cap U}$. 
        Thus, in order to prove that $u \in A_\mu (z_x)$,
        it suffices to show that $U \cap W$ is a $u,s$-mincut.
        Note that it is a $u,s$-cut because $|U \cap \{u,s\}| = 1$ and $u,s \in W$.
        Also, note that $U \cup W$ is a $\mu$-cut since $\mu(a_x, g_x) = \mu_i (z_x) \subseteq U \cup W$ and $\mu(a_y,g_y) \subseteq \mu_{1-i}(z_i) \cap \mu(b,g) \subseteq \overline{U \cup W}$.
        Hence, by submodularity (\Cref{lem:sub-posi-general}), we get that $U \cap W$ is indeed a $u,s$-mincut.

        \item Case 3: $u,s \notin W$.
        In this case, let $g = \{a,b\}$ be any edge in the overlap that $u$ separates, which exists by the assumption of the lemma, such that $a$ is closer to $a_x$ and $b$ is closer to $a_y$ on $\pi_\mu (e)$.
        Then there exist a $u,s$-mincut $U$ with $\mu_i (z_x) \subseteq \mu(a,g) \subseteq U$ and $\mu_{1-i} (z_x) \subseteq \mu(b,g) \subseteq \overline{U}$.
        We have that $\{z_x\} \cup \mu_i (z_x) \subseteq U \cup W$ and $\mu_{1-i} (z_x) \subseteq \overline{U \cup W}$.
        Thus, to prove that $u \in A_\mu (z_x)$,
        it suffices to show that $U \cup W$ is a $u,s$-mincut.
        Note that it is a $u,s$-cut because $|U \cap \{u,s\}| = 1$ and $u,s \notin W$.
        Also, note that $U \cap W$ is a $\mu$-cut since $\mu_i (z_x) \subseteq U \cap W$ and $\mu_{1-i} (z_x) \subseteq \overline{U \cap W}$.
        Hence, by submodularity (\Cref{lem:sub-posi-general}), we get that $U \cup W$ is indeed a $u,s$-mincut.
    \end{itemize}
    (Existence of $W$)
    Take any edge $g = \{a,b\}$ in the overlap, where $a$ is closer to the $x$-endpoint $a_x$ and $b$ is closer to the $y$-endpint $a_y$ on $\pi_\mu (e)$, so $\mu_i (z_x) \cup \mu_{1-j} (z_y) \subseteq \mu(a,g)$ and $\mu_{1-i}(z_x) \cup \mu_j (z_y) \subseteq \mu(b,g)$.
    Take any canonical cut $C$ of $\cH_\mu$ that contains $g$, and let $\mu_1,\mu_2$ be its corresponding valid partition, where $\mu_1$ are the terminals mapped to the side of $a$ (and the $x$-endpoint $a_x$), and $\mu_2$ are those mapped to the side of $b$ (and the $y$-endpoint $a_y$), so $\mu(a,g) \subseteq \mu_1$ and $\mu(b,g) \subseteq \mu_2$.
    Because $g$ is an edge of $\pi_\mu (e)$,~\Cref{thm:proj-edge} ensures that $\mu_1,\mu_2$ distinguishes $e$.
    Therefore, there must exists a $\mu$-mincut $W$ with $\mu_1 \subseteq W$ and $\mu_2 \subseteq W$ where $e$ is contributing;
    Furthermore, by the unidirectionality property in~\Cref{thm:proj-edge}, it must be that $x \in W$ and $y \in \overline{W}$.
    Now, take any $w \in \mu$ which is from the side without $s$ of $W$ (that is, if $s \in \overline{W}$ take $w \in \mu_1$, and if $s \in W$ take $w \in \mu_2$).
    Then, as $c(W) = \lambda_\mu = c(\cDdown(\mu))$ and $\cDdown(\mu) = \Far(w,s)$, we have that $W$ is a $w,s$-mincut.
    Finally, \Cref{lem:translation-DAG} now yields that $z_x \in W$ and $z_y \in \overline{W}$.
    This concludes the proof.
\end{proof}

\begin{lemma}[4c]\label{lem:4c}
    Suppose $u$ does not separate any edge of the middle/overlap, but $u$ separates the prefix $\pi_\mu (z_x) \setminus \pi_\mu (z_y)$ or the suffix $\pi_\mu (z_y) \setminus \pi_\mu (z_x)$.
    Then $u \in A_\mu (z_x) \cup A_\mu (z_y)$ if and only if there exists a $u,s$-mincut which splits $\mu$ and separates $z_x,z_y$
\end{lemma}
\begin{proof}
    ($\Rightarrow$) 
    Say $u \in A_\mu (z_x)$, as $u \in A_\mu (z_y)$ is symmetric.
    Then, there exists a $(u,s)$-mincut $U$ keeping $\mu_i (z_x) \cup \{z_x\}$ and $\mu_{1-i} (z_x)$ on different sides.
    By~\Cref{lem:mu-mincuts}, we have that $W := U \cap \cDdown(\mu)$ is a $\mu$-mincut, 
    which also separates $\mu_i (z_x) \cup \{z_x\}$ from $\mu_{i-1} (z_x)$.
    Let $C$ be the canonical cut of $\cH_\mu$ which corresponds to the valid partition $\mu \cap W, \mu \cap \overline{W}$.
    Recall that $\mu_i (z_x)$ and $\mu_{i-1} (z_x)$ are the terminal sets associated with the proper path $\pi_\mu (e)$, so by~\Cref{lem:associated-terminals}, $C$ must have an edge $g = \{a,b\}$ from $\pi_\mu (z_x)$, where $a$ is closer to the endpoint $a_x = a_i (z_x)$ of $\pi_\mu (z_x)$, and $b$ is closer its other endpoint $a_{1-i}(z_x)$.
    Thus, $W$ keeps $\mu(a,g)$ on the side of $\{z_x\} \cup \mu_i (z_x)$, and $\mu(b,g)$ on the side of $\mu_{i-1} (z_x)$.
    As these are subsets of $\cDdown(\mu)$, this is also true for $U$, so $u$ separates $g$.
    Now, since $u$ does not separate edges of the overlap, $g$ must be in the prefix $\pi_\mu (z_x) \setminus \pi_\mu (z_y)$.
    Thus, $C$ keeps all of $\pi_\mu (z_y)$ on the side of $\mu (b,g)$.
    This means that $W$ must keep $z_y$ with $\mu (b,g)$ by~\Cref{thm:proj-vertex}.
    Hence, $W$ separates $\{z_x\} \cup \mu(a,g)$ from $\{z_y\} \cup \mu(b,g)$.
    Because these are subsets of $\cDdown(\mu)$, this is also true for $U$, so $U$ is a $u,s$-mincut which splits $\mu$ and separates $x,y$.
        
    ($\Leftarrow$) 
    Let $U$ and $g = \{a,b\}$ be $u,s$-mincut and the of $\pi_\mu (e)$ guaranteed by~\Cref{lem:finiding-the-separated-edge}.
    As $u$ does not separate any edge in the middle part, $g$ must be either in $\pi_\mu (z_x)$ or in $\pi_\mu (y)$.
    If $g$ is in $\pi_\mu (x)$, then we have $\mu_i (z_x) \subseteq \mu(a,g)$ and $\mu_{1-i} (z_x) \subseteq \mu(b,g)$, hence $U$ separates $\{z_x\} \cup \mu_i (z_x)$ from $\mu_{1-i}(z_x)$, and thus $u \in A_\mu (z_x)$.
    The case where $g$ is in $\pi_\mu (y)$ is symmetric, leading to the conclusion that $u \in A_\mu (z_y)$.
\end{proof}

\begin{lemma}[4d]
    Suppose $u$ does not separate $\pi_\mu (e)$.
    Then no $u,s$-mincut splits $\mu$ and separates $z_x,z_y$.
\end{lemma}
\begin{proof}
    Immediate by~\Cref{lem:finiding-the-separated-edge}.
\end{proof}

This concludes the proof of~\Cref{lem:q1-data-structure}.

\section{Mincuts Not Splitting the LCA}\label{sec:not-splitting-mu}

This section provides the data structure for Query~\ref{query-Q2}, restated below.

\begin{enumerate}[label=\textbf{(Q2)}]
    \item Given $e = \{x,y\} \in E$ and an LCA $\mu$ of $\cD(x),\cD(y)$,
    report all $u \in \cDup(\mu)$ such that some $u,s$-mincut separates $x,y$ and keeps $\mu$ on one side.
\end{enumerate}
Again, as was with Query~\Cref{query-Q1}, 
by~\Cref{lem:translation-DAG} the condition in Query~\ref{query-Q2} is equivalent to ``some $u,s$-mincut separates $z_x,z_y$ and keeps $\mu$ on one side'', where $z_x,z_y$ are the anchors of $x,y$ w.r.t.\ $\mu$. 
So, our goal in this section is to show:

\begin{lemma}\label{lem:q2-data-structure}
    There is a data structure for Query~\ref{query-Q2} with $O(n)$ space and $O(n)$ query time.
\end{lemma}

We will need the following tool, which is similar to~\Cref{lem:sigma-numbers}.

\begin{lemma}\label{lem:tau-numbers}
    Let $\mu$ be a node in $\cD$.
    Then we can assign each vertex $z \in V$ a number $\tau_\mu (z) \in [n]$, such that the following hold for every $z_1, z_2 \in V$:
    \begin{itemize}
        \item If $\tau_\mu (z_1) = \tau_\mu (z_2)$, then there is no $\mu,s$-mincut separating $z_1,z_2$.
        
        \item If $\tau_\mu (z_1) < \tau_\mu (z_2)$, then there exists a $\mu,s$-mincut $W$ with $\{z_1\} \cup \mu \in W$, $z_2 \in \overline{W}$.
    \end{itemize}
\end{lemma}
\begin{proof}
    Define a directed graph $H$ on the vertices $V$ where edge $(z,z')$ is \textbf{not} in $E(H)$ iff there is a $\mu,s$-mincut $W$ with $\mu \cup \{z'\} \subseteq W$, $z \in \overline{W}$.
    Contract each SCC (strongly connected component) of $H$, and find a topological ordering for the resulting DAG:
    $\tau_\mu (z)$ is the place of the SCC containing $z$ in this ordering.

    The second item is now easy to prove: if $\tau_\mu (z_1) < \tau_\mu (z_2)$, then the edge $(z_2,z_1)$ cannot exist in $H$ as it would violate the topological ordering.

    For the first item, say $\tau_\mu (z_1) = \tau_\mu (z_2)$.
    Seeking contradiction, suppose there is some $\mu,s$-mincut $W$ separating $z_1,z_2$.
    Without loss of generality, say $\mu \cup \{z_2\} \in W$ and $z_1 \in \overline{W}$ (the case $z_1 \in W$ and $z_2 \in \overline{W}$ is symmetric).
    Take some path $(z_1 = v_0, v_1, \dots, v_k = z_2)$ from $z_1$ to $z_2$ in $H$, which exists as $z_1,z_2$ belong to the same SCC.
    Consider the largest $0 \leq r \leq k-1$ with $v_r \in \overline{W}$ (which is well-defined since $v_0 = z_1 \in \overline{W}$ and $v_k = z_2 \in W$).
    Then $v_{r+1} \in W$, but this contradicts the existence of edge $(v_r, v_{r+1})$ in $H$.
\end{proof}

\subsection{Construction}
We now formally describe the information stored by our data structure for Query~\ref{query-Q2}.
First, we store the $O(n)$ space data structure of the representative framework from~\Cref{thm:blackbox}.
Next, for each node $\mu$ of $\cD$, and each $z \in \{z_\nu \mid \text{$\nu$ child of $\mu$}\} \cup \mu$, we store $\RepInf(\mu, z)$, and the number $\tau_\mu (z)$ from~\Cref{lem:tau-numbers}, which consumes a total of $O(|\mu| + \operatorname{outdeg}_\cD(\mu))$ space for $\mu$.
Because the nodes of $\cD$ partition the vertices $V$, and because $\sum_\mu \operatorname{outdeg}_\cD (\mu)$ counts each of the $O(n)$ edges of $\cD$ exactly once, the total space is $O(n)$.

\subsection{Query Algorithm}
We now give the algorithm to answer Query~\ref{query-Q2}, given $e = \{x,y\} \in E$ and LCA $\mu$ of $\cD(x),\cD(y)$.
\begin{description}
    \item[Step 0:] Find the anchors $z_x,z_y$ of $x,y$ w.r.t.\ $\mu$, which takes $O(n)$ time by searching $\cD$.
    Note thta $z_x,z_y \in \{z_\nu \mid \text{$\nu$ child of $\mu$}\}$ by~\Cref{def: anchors}.
    
    \item[Step 1:] 
    Use $\RepInf(\mu, z_x)$ in~\Cref{thm:blackbox} to determine the set $B(z_x) = \{u \in \cDup(\mu) \mid \text{$u$ separates $\mu,z_x$}\}$.
    Similarly, use $\RepInf(\mu,z_y)$ to compute $B(z_y) = \{u \in \cDup(\mu) \mid \text{$u$ separates $\mu,z_y$}\}$.
    This takes $O(n)$ time by~\Cref{thm:blackbox}.

    \item[Step 2:]
    Answer the query as follows.
    \begin{enumerate}[label=(\alph*)]
        \item If $\tau_\mu (z_x) = \tau_\mu (z_y)$: Return $\emptyset$.
        \item If $\tau_\mu (z_x) < \tau_\mu (z_y)$: Return $B(z_y)$.
        \item If $\tau_\mu (z_x) > \tau_\mu (z_y)$: Return $B(z_x)$.
    \end{enumerate}
\end{description}
Thus, our data structure answers Query~\ref{query-Q2} within $O(n)$ time.

\subsection{Correctness}

First, we prove the following ``helper'' lemma:
\begin{lemma}\label{lem:mu-not-split-completeness}
    Suppose that for some vertex $u \in \cDup(\mu)$, there is a $u,s$-mincut $U$ with $\{z_x\} \cup \mu \subseteq U$ and $z_y \in \overline{U}$.
    Then $\tau_\mu (z_x) < \tau_\mu (z_y)$.
    (Symmetrically, if for some vertex $u \in \cDup(\mu)$ there is a $u,s$-mincut $U$ with $\{z_y\} \cup \mu \subseteq U$ and $z_y \in \overline{U}$, then $\tau_\mu (z_x) > \tau_\mu (z_y)$.)
\end{lemma}
\begin{proof}
    By~\Cref{lem:translation-DAG}, we may also assume that $x \in U$ and $y \in \overline{U}$.
    We observe that $U \cup \cDdown(\mu)$ is a $u,s$-cut:
    \begin{itemize}
        \item If $u \in \mu$, then as $\mu \subseteq U$ we have $u \in U$, and $s \notin U$.
        Thus, $u \in U \cup \cDdown(\mu)$ and $s \in \overline{U \cup \cDdown(\mu)}$.
        \item If $u \notin \mu$, then as $u \in \cDup(\mu)$, it must be that $u \notin \cDdown(\mu)$.
        Thus, $|U \cap \{u,s\}| = 1$ and $u,s \notin \cDdown(\mu)$.
    \end{itemize}
    Also, $U \cap \cDdown(\mu)$ is a $\mu,s$-cut as $\mu \subseteq U$ and $s \notin \cDdown(\mu)$.
    Thus, by submodularity (\Cref{lem:sub-posi-general}), we obtain that $W := U \cap \cDdown(\mu)$ is a $\mu,s$-mincut.
    Note that $\mu \cup \{x, z_x\} \subseteq W$ and $\{y, z_y\} \subseteq \overline{W}$.
    So, by the first item in~\Cref{lem:tau-numbers}, we have $\tau_\mu (z_1) \neq \tau_\mu (z_2)$.
    Now, assume towards contradiction that $\tau_\mu (z_x) > \tau_\mu (z_y)$.
    Then by~\Cref{lem:tau-numbers}, there is a $\mu,s$-mincut $W'$ with $\{z_y\} \cup \mu \subseteq W'$ and $z_x \in \overline{W'}$.
    By~\Cref{lem:translation-DAG}, we may also assume that $x \in W'$ and $y \in \overline{W'}$.
    Note that both $W \cap W'$ and $W \cup W'$ are $\mu,s$-cuts, so by submodularity (\Cref{lem:sub-posi-general}) there cannot be any edge between $W \setminus W'$ and $W' \setminus W$.
    But $x \in W \setminus W'$ and $y \in W' \setminus W$, so $e$ is such an edge --- contradiction.
    We therefore conclude that $\tau_\mu (z_x) < \tau_\mu (z_y)$.
\end{proof}

We now prove the correctness of the substeps of Step 2.
As 2b and 2c are completely symmetric, we only consider 2b here.
For this, we fix any $u \in \cDup(\mu)$, and show that $u$ is reported in the output of our algorithm precisely when there is some $u,s$-mincut which separates $z_x,z_y$ and keeps $\mu$ on one side.

\begin{lemma}[2a]
    If $\tau_\mu (z_x) = \tau_\mu (z_y)$, there is no $u,s$-mincut which separates $z_x,z_y$ and does not split $\mu$.
\end{lemma}
\begin{proof}
    Immediate by~\Cref{lem:mu-not-split-completeness}, as such a $u,s$-mincut implies either that $\tau_\mu (z_x) < \tau_\mu (z_y)$ or that $\tau_\mu (z_x) > \tau_\mu (z_y)$, depending on whether it keeps $\mu$ with $z_x$ or with $z_y$.
\end{proof}

\begin{lemma}[2b]\label{lem:2b}
    Suppose $\tau_\mu (z_x) < \tau_\mu (z_y)$.
    Then it holds that $u \in B(z_y)$ if and only if there exists a $u,s$-mincut which separates $z_x,z_y$ and does not split $\mu$.
\end{lemma}
\begin{proof}
    ($\Rightarrow$) 
    Suppose $u \in B(z_y)$, so by definition of $B(z_y)$ there is a $u,s$-mincut with $\mu \subseteq U$ and $z_y \in \overline{U}$.
    As $\tau_\mu (z_x) < \tau_\mu (z_y)$, by~\Cref{lem:tau-numbers} there is a $\mu,s$-mincut $W$ with $\{z_x\} \cup \mu \in W$, $z_y \in \overline{W}$.
    Note that $U \cap W$ is a $\mu,s$-cut.
    Also, we claim that $U \cup W$ is a $u,s$-cut as follows:
    \begin{itemize}
        \item If $u \in U$ and $s \in \overline{U}$, then as we also have $s \in \overline{W}$ the claim holds.

        \item Otherwise, $s \in U$ and $u \in \overline{U}$, and we should show that $u \in \overline{W}$.
        Seeking contradiction, suppose $u \in W$.
        As $W \subseteq \Far(\mu,s) = \cDdown(\mu)$, this means that $\cD(u)$ is a descendant of $\mu$ in $\cD$. But $\cD(u)$ is also an ancestor of $\mu$, so this means that $\cD(u)=\mu$, i.e., $u \in \mu \subseteq U$ --- contradiction as $u \in \overline{U}$.
    \end{itemize}
    Thus, by submodularity (\Cref{lem:sub-posi-general}), $U \cup W$ is a $u,s$-mincut, with $\{z_x\} \cup \mu \subseteq U \cup W$ and $z_y \in \overline{U \cup W}$.
    
    ($\Leftarrow$) 
    Let $U$ be a $u,s$-mincut which separates $z_x,z_y$ and does not split $\mu$;
    then $U$ must keep $\mu$ on the opposite side from $z_y$, as otherwise we would get $\tau_\mu (z_x) > \tau_\mu (z_y)$ by~\Cref{lem:mu-not-split-completeness}, in contradiction to the assumption of the lemma.
    Thus, the $u,s$-mincut $U$ shows that $u \in B(z_y)$.
\end{proof}

This concludes the proof of~\Cref{lem:q2-data-structure}.

\section{Data Structures for Fast Queries: Output-Sensitive and One-Destination}\label{sec:faster-queries}

In this section, we show how, by allowing space of $O(n^{1.5})$, we can support faster edge-failure queries for single-source mincuts as stated in~\Cref{thm:faster-queries}: an output-sensitive query algorithm that reports all vertices whose mincut to $s$ in time which is proportional to their number up to an $\log^3 n$ factor, and a one-destination query algorithm that, given a single destination vertex of interest as part of the query, determines whether its mincut to $s$ changes upon the failure in $O(\log n)$ time.

\subsection{Ancestry and LCA in Farthest Mincut DAG}\label{sec:DAG-data-structure}

The most important component of~\Cref{thm:faster-queries} is a data structure for the farthest mincut $\cD$ that occupies $O(n^{1.5})$ and supports fast ancestry and LCA queries, which is given in~\Cref{lem:ancestry-lca-data-structure} below.
In fact, this is the \emph{only} component of our data structure with $O(n^{1.5})$ space, where all other components only need $O(n \log n)$ space.
Improving the space of this ancestry/LCA structure beyond $O(n^{1.5})$ while maintaining polylogarithmic query time would directly yield an improved space bound for the entire data structure, and doing so appears to be a challenging and interesting problem.

We mention that the $O(n^{1.5})$ space bound is by itself non-trivial, and might seem surprising in light of the \emph{Reachability Conjecture}~\cite{DBLP:journals/siamcomp/Patrascu11, DBLP:conf/wads/GoldsteinKLP17}.
This conjecture states that any data structure supporting reachability (i.e., ancestry) queries for DAGs with $n$ vertices and $m$ edges must either spend $\tilde{\Omega}(m)$ query time or consume $\tilde{O}(n^2)$ space \emph{in the worst case}.
However, the farthest mincut DAG $\cD$ is not worst-case; it is very well-structured due to~\Cref{lem:at-most-two-incomparable-ancestors}, which is the key to obtaining the data structure.

\begin{lemma}\label{lem:ancestry-lca-data-structure}
    There is an $O(n^{1.5})$ space data structure for the farthest mincut DAG $\cD$ that supports the following queries:
    \begin{itemize}
        \item \textbf{(Ancestry)} Given two nodes $\alpha$ and $\beta$, determine if $\alpha$ is ancestor of $\beta$.
        The query takes $O(1)$ time.

        \item
        \textbf{(LCA)}
        Given nodes $\alpha, \beta$, return the (at most two) LCAs of $\alpha,\beta$.
        Furthermore, for each such LCA $\mu$, if $\alpha \neq \mu$ (resp., if $\beta \neq \mu$), also return a child $\nu$ of $\mu$ such that $\nu$ is an ancestor of $\alpha$ (resp., of $\beta$).
        The query takes $O(\log n)$ time.
    \end{itemize}
\end{lemma}

\begin{proof}
    The key property of $\cD$ that yields this data structure is~\Cref{lem:at-most-two-incomparable-ancestors}, saying that among any three ancestors of a given node in $\cD$, at least two must be related by ancestry.
    To utilize it, we will exploit \emph{chain-antichain decompositions}, and \emph{Dilworth's Theorem}:
    \begin{itemize}
        \item A chain (resp., antichain) is a set of nodes where every pair is related (resp., unrelated) by ancestry.
        I.e., chains and antichains are respectively paths and independent sets in the transitive closure $TC(\cD)$.

        \item A \emph{chain-antichain} decomposition is a partition of the nodes in $\cD$ into a collection of sets, where each set is either a chain or an antichain.
        For every $1 \leq \ell \leq n$, there always exists such a decomposition with $O(\ell)$ chains and $O(n/\ell)$ antichains~\cite{GrandoniW20J}.

        \item The \emph{width} of a set of nodes is the maximal size of an antichain within it.
        \emph{Dilworth's Theorem}~\cite{Dilworth50} states that if a set has width $k$ then it has a \emph{$k$-chain cover}, that is, it can be partitioned into $k$ chains.
    \end{itemize}

    We start with the simple intuition underlying our strategy for ancestry queries.
    Fix a node $\nu$ of $\cD$, and let us consider how its ancestors behave w.r.t.\ chains and anti-chains.
    Obviously, the ancestors of $\nu$ within any arbitrary chain $C$ form a prefix of $C$, which makes chains rather easy to deal with.
    Our key property (\Cref{lem:at-most-two-incomparable-ancestors}) renders antichains easy as well: an arbitrary antichain $A$ can only contain two ancestors of $\nu$, since $\nu$ cannot have three ancestors which are pairwise-unrelated.
    
    In light of this, to support ancestry queries, we decompose $\cD$ into $O(\sqrt{n})$ chains $C_1, C_2, \dots$ and $O(\sqrt{n})$ antichains $A_1, A_2, \dots$.
    For each node $\nu$, we store its last ancestor from each chain $C_i$, and its at most two ancestors from each antichain $A_i$; denote this set of $O(\sqrt{n})$ stored ancestors by $S(\nu)$.
    Now, to answer a query asking whether $\alpha$ is an ancestor $\beta$, we consider two cases. If $\alpha$ is from an antichain $A_i$, we just need to check whether $\beta$ explicitly stored it in $S(\beta)$ as one of its ancestors from $A_i$.
    If $\alpha$ is from a chain $C_i$, we just need to check if it appears before the last ancestor of $\beta$ on $C_i$, which was explicitly stored by $\beta$.
    Thus, we can answer ancestry queries within $O(1)$ time.

    Let us now discuss LCA queries.
    Our first observation is that the LCAs of $\alpha,\beta$ must lie in $S(\alpha) \cup S(\beta)$.
    Indeed, if an LCA $\mu$ is from an antichain $A_i$, this is immediate.
    If $\mu$ is from a chain $C_i$, then it must be the last ancestor of either $\alpha$ or $\beta$ on $C_i$: otherwise, the node $\mu'$ following $\mu$ on $C_i$ would be a common ancestor of $\alpha,\beta$ which is a strict descendant of $\mu$, a contradiction.
    
    The next crucial observation is that, due to the key property of $\cD$ (\Cref{lem:at-most-two-incomparable-ancestors}), each of $S(\alpha)$ and $S(\beta)$ has width at most $2$ (as each of them consists of only ancestors of a given node, and thus cannot contain three pairwise-unrelated nodes).
    Therefore, we can apply Dilworth's Theorem (during preprocessing) and decompose 
    each of them into two chains.
    To answer the LCA query between $\alpha$ and $\beta$, we consider each of these four chains, and find the last common ancestor of $\alpha,\beta$ on each by an $O(\log n)$ time binary search.
    Note that the common ancestors form a prefix in each chain, and we can check if a node is a common ancestor by two ancestry queries of $O(1)$ time, so binary search is indeed applicable.
    The LCAs of $\alpha,\beta$  (which are at most two by~\Cref{lem:at-most-two-incomparable-ancestors}) must be among these four candidate nodes; we detect them by $O(1)$ ancestry queries between the candidates, to filter out those that are strict ancestors of others.

    Finally, let us explain how we find children of the LCAs that are ancestors of $\alpha$ and of $\beta$.
    To this end, we augment our data structure as follows: Whenever node $\nu$ stores some ancestor $\mu$ in $S(\nu)$, it also stores a child of $\mu$ which is an ancestor of $\nu$.
    Additionally, whenever some node $\mu$ from a chain $C_i$ is stored (in some $S(\nu$)), we also store a child of $\mu$ which is an ancestor of the node $\mu'$ following $\mu$ on $C_i$.
    Now, consider an LCA $\mu$ of $\alpha,\beta$, and say $\alpha \neq \mu$ and we want to find a child of $\mu$ that is an ancestor of $\alpha$.
    If $\mu \in S(\alpha)$, then $\alpha$ has explicitly stored such a child.
    Otherwise, $\mu$ must appear before the last ancestor of $\alpha$ on some chain $C_i$; in this case, the child we have stored for $\mu$ w.r.t.\ $C_i$, which is an ancestor of the node $\mu'$ following $\mu$ on $C_i$, must also be an ancestor of $\alpha$.
\end{proof}

\begin{remark}
    As the main cause for the $O(n^{1.5})$ space bound is the decomposition into $O(\sqrt{n})$ chains and antichains, which works for every DAG, one might hope for a better decomposition utilizing the special properties of $\cD$, namely~\Cref{lem:at-most-two-incomparable-ancestors}.
    However, this is not possible.
    Indeed, take $\cD$ as a tree of $\sqrt{n}$ paths connected only at the root $s$, each of length $\sqrt{n}$.
    Note that \Cref{lem:at-most-two-incomparable-ancestors} holds for $\cD$: That is, the set of ancestors of any given node has width at most $2$ (in fact $1$).
    As the longest chain and the largest antichain in $\cD$ are both of size $\sqrt{n}$, it is impossible to partition its $n$ nodes into $o(\sqrt{n})$ chains and antichains.
\end{remark}

\subsection{Extensions of The Representatives Framework}

We now present two extensions of the representative framework using the ancestry data structure for $\cD$ of~\Cref{lem:ancestry-lca-data-structure}.
The first extension is useful for the One-Destinations queries:

\begin{lemma}\label{lem:rep-framework-one-vertex}
    There is an $O(n^{1.5})$ space data structure such that, given $\RepInf(P,Q)$ for a pair of non-empty $P,Q \subseteq V$, and given $u$ such that $P \cup Q \subseteq \cDdown(u)$, reports if $u$ separates $P,Q$ in $O(1)$ time.
\end{lemma}
\begin{proof}
    The data structure stores the DAG $\cD$ and its ancestry/LCA data structure from~\Cref{lem:ancestry-lca-data-structure}, which takes $O(n^{1.5})$ space.
    Additionally, for every node $\nu$ in $\cD$, it stores the skeleton $\cH_\nu$ and its data structure of~\Cref{lem:skeleton-data-structure}, which consumes $O(|\nu|)$ space for each $\nu$, and hence only $O(n)$ space overall.
    
    Let $\nu_1, \dots, \nu_r$ be the $O(1)$ representative nodes for $P,Q$ found in $\RepInf(P,Q)$.
    For each $\nu_i$, we check if it is an ancestor of $\cD(u)$, which takes $O(1)$ time using the ancestry data structure of~\Cref{lem:ancestry-lca-data-structure}.
    If no $\nu_i$ is an ancestor, then $u$ does not separate $P,Q$ by~\Cref{lem:rep-nodes}.
    If some $\nu_i$ is a \emph{strict} ancestor, then $u$ separate $P,Q$ by~\Cref{lem:rep-nodes-converse}.
    It remains to consider the case where $u \in \nu_i$ (and $\nu_i$ has no strict ancestors among $\nu_1, \dots, \nu_r$).
    In this case, let $C_1, \dots, C_\ell$ be the $O(1)$ representative cuts for $P,Q$ of the node $\nu_i$ (with specified marked sides).
    Then $u$ separates $P,Q$ iff $\pi_\nu (u)$ is on the marked side of some $C_i$, which we can check within $O(1)$ time using the data structure for $\cH_{\nu_i}$ of~\Cref{lem:skeleton-data-structure}.
\end{proof}

The second extension is useful for the output-sensitive query algorithm, and it a bit more involved.
Informally, the setting is the following.
We have some node of interest $\nu$ in $\cD$, and we are given $\RepInf(P_i,Q_i)$ for $k$ pairs $P_i,Q_i \subseteq V$.
For each pair, we also have a \emph{demand}: ``separate'' or ``do not separate''.
We wish to find the set $Y$ consisting of every vertex $v$ inside $\nu$ that fulfills the demands: i.e., $v$ separates every pair $P_i,Q_i$ whose demand is ``separate'', and does not separate any pair $P_i,Q_i$ whose demand is ``do not separate''.

\begin{lemma}\label{lem:blackbox-reporting-all}
The $O(n^{1.5})$ space data structure of~\Cref{lem:rep-framework-one-vertex} can also support the following queries:
    Given
    \begin{itemize}
        \item $\RepInf(P_i, Q_i)$ for non-empty $P_i, Q_i \subseteq V$, $i = 1, \dots, k$,
        \item Boolean formula $\varphi(x_1,\dots,x_k)$ where each variable $x_i$ occurs exactly once, and
        \item node $\nu$ of $\cD$ s.t.\ $P_i \cup Q_i \subseteq \cDdown(\nu)$, for every $1 \leq i \leq k$,
    \end{itemize}
    return the set
    \[
    Y  = \Big\{v \in \nu \mid \varphi\big(\mathbf{1}_{[\text{$v$ separates $P_1,Q_1$}]}, \dots, \mathbf{1}_{[\text{$v$ separates $P_k,Q_k$}]} \big) = 1 \Big\}.
    \]
    The time to answer the query is $O(|Y| + k^2)$. 
\end{lemma}
\begin{proof}
See the start of the proof of~\Cref{lem:rep-framework-one-vertex} for the definition of the data structure.

We now explain how to answer the query.
First, for every $i =1, \dots, k$, we check if some representative node of $P_i,Q_i$ is an ancestor of $\nu$ in $\cD$ using the data structure of~\Cref{lem:ancestry-lca-data-structure}.
This takes $O(k)$ time, as there are $O(1)$ representative nodes for each $P_i,Q_i$, and each ancestry query takes $O(1)$ time.
\begin{itemize}
    \item If some representative node of $P_i,Q_i$ is a \emph{strict} ancestor 
    of $\nu$, then every vertex in $\nu$ separates $P_i,Q_i$ by~\Cref{lem:rep-nodes-converse}.
    So in this case, we can set the variable $x_i$ to $1$.
    \item  If no representative node of $P_i,Q_i$ is an ancestor of $\nu$, then no vertex in $\nu$ separates $P_i,Q_i$ by~\Cref{lem:rep-nodes}.
    So, in this case, we can set the variable $x_i$ to $0$.
\end{itemize}
So, by eliminating all the above $P_i,Q_i$ pairs and restricting the formula $\varphi$ accordingly, we may assume that for every $1 \leq i \leq k$, the vertices separating $P_i,Q_i$ in $\nu$ are those mapped to the marked side of some representative cut of $P_i,Q_i$ in the skeleton $\cH_\nu$.
Let $C_{i,1}, \dots ,C_{i,r_i}$ be these $O(1)$ representative cuts in $\nu$, with marked sides $A_{i,1}, \dots, A_{i,r_i}$.
Thus, $v \in \nu$ separates $P_i,Q_i$ iff $v \in A_{i,1} \cup \cdots \cup A_{i,r_i}$.

We can now translate the query to a ``cut-sides formula'' query in in $\cH_\nu$ and use its data structure from~\Cref{lem:skeleton-data-structure}.
In more detail, we create a new Boolean formula $\psi$ from $\varphi$ as follows: For each $1 \leq i \leq k$, create new variables $y_{i,1}, \dots, y_{i,r_i}$ representing the marked sides $A_{i,1}, \dots, A_{i,r_i}$.
We replace each occurrence of $x_i$ by an OR of these news variables $y_{i,1} \vee \cdots \vee y_{i,r_i}$.
The new formula $\psi$ has $O(k)$ variables, and $Y$ is precisely the set of terminals in $\nu$ obtained by evaluating $\psi$ as a set expression on the $A_{i,j}$'s.
Thus, the query time to find $Y$ is $O(|Y| + k^2)$ as needed.
\end{proof}

\subsection{Output-Sensitive and One-Destination Query Algorithms: Proof of~\Cref{thm:faster-queries}}

We now provide the proof for the failure variant of~\Cref{thm:faster-queries}; the proof for the simpler insertion variant appears in~\Cref{sec:edge-insertion}.
We first focus on \Cref{item:faster-queries-reporting-all} of~\Cref{thm:faster-queries}.
That is, we show a data structure of $O(n^{1.5})$ space that, given a query of failed edge $e = \{x,y\} \in E$, reports $X = \{u \in V \setminus \{s\} \mid \text{$u$ separates $x,y$}\}$ within $O(|X| \log^3 n)$ time.

The data structure construction augments the previous $O(n)$ data structures for Query~\ref{query-Q1} (\Cref{lem:q1-data-structure}) and Query~\ref{query-Q2} (\Cref{lem:q2-data-structure}), by adding also the $O(n^{1.5})$ space structures for ancestry/LCA queries in $\cD$ from~\Cref{lem:ancestry-lca-data-structure} and for the extensions of the representative framework from~\Cref{lem:rep-framework-one-vertex,lem:blackbox-reporting-all}.
We will also use additional $O(|\mu| \log |\mu|)$ space for each node $\mu$ of $\cD$, and hence only $O(n \log n)$ additional space in total, as will be explained shortly (in the discussion after~\Cref{lem:separating-subpaths}).

\paragraph{Treating Exclusive Ancestors.}
We start by executing a multi-source BFS from $\cD(x),\cD(y)$ \emph{in the reverse of $\cD$}, that keeps exploring the edges of a reached node only if it is not a common ancestor of $\cD(x),\cD(y)$, which is checked within $O(1)$ using~\Cref{lem:ancestry-lca-data-structure}.
The ``halting nodes'' where this condition doesn't hold are precisely the LCAs of $\cD(x),\cD(y)$, which are at most two by~\Cref{lem:at-most-two-incomparable-ancestors}.
Except for the LCAs, the BFS visits precisely those ``exclusive ancestors'' from~\Cref{sec:DAG-overview}, i.e., ancestors of exactly one of $\cD(x),\cD(y)$).
As explained there, all vertices in these nodes belong to $X$.
Because each node in $\cD$ has at most two parents, the number of edges that this BFS-in-reverse touches is at most twice the number of nodes it reaches.
Therefore, this procedure reports all those vertices of $X$ found in exclusive ancestors within $O(X)$ time.
All other vertices in $X$ are found in ancestors of the LCAs, as explained in~\Cref{sec:DAG-overview}, and finding these becomes our main goal.

So, from now on we focus on an LCA $\mu$ of $\cD(x),\cD(y)$.
We can also ensure that the BFS procedure above finds two children of $\mu$ which are ancestors of $\cD(x)$ and of $\cD(y)$, and hence allows us to determine the anchors $z_x,z_y$ of $x,y$ w.r.t.\ $\mu$ by~\Cref{def: anchors}.
Given the equivalence of the anchors to $x,y$ from~\Cref{lem:translation-DAG}, our problem now reduces to the following task:
Given $e = \{x,y\} \in E$, an LCA $\mu$ of $\cD(x),\cD(y)$, and the anchors $z_x,z_y$ of $x,y$ w.r.t.\ $\mu$, we wish to report the set $X^\mu : = \{u \in \cDup(\mu) \mid \text{$u$ separates $z_x,z_y$}\}$ in $\tilde{O}(|X^{\mu}|)$ time.
Since this is our only focus, and in order to avoid clutter, from now on we will slightly abuse notation and write $X$ instead of $X^\mu$.

\paragraph{$\RepInf(\cdot)$ for Proper Paths.}

We will need some additional notions regarding proper paths in the skeleton $\cH_\mu$.
Let $P$ be such a proper path with endpoint nodes $a$ and $b$.
Vertex $u$ is said to \emph{separate $P$} if $u$ separates the associated terminal sets $\mu(a,P),\mu(b,P)$ of $P$ (which have been defined in~\Cref{def:associated-terminals}).
That is, there exists a $u,s$-mincut which keeps $\mu(a,P)$ on one side, and $\mu(b,P)$ on the other side.
(Note that a single skeleton edge is trivially also a proper path.)
We define \emph{the representative information of $P$} as a shorthand to the representative information of its associated terminal sets, and write this as $\RepInf(P) := \RepInf(\mu(a,P), \mu(b,P))$.
We have the following easy consequence of~\Cref{lem:associated-terminals} and~\Cref{lem:mu-mincuts}:
\begin{lemma}\label{lem:separating-subpaths}
    Let $P$ be a proper path in $\cH_\mu$.
    Then, for each vertex $u \in V \setminus \{s\}$, it holds that 
    $u$ separates $P$ 
    iff 
    $u$ separates some edge $g$ of $P$.
\end{lemma}
\begin{proof}
    First, consider any $\mu$-mincut $W$.
    By~\Cref{lem:associated-terminals}, $W$ separates the associated terminal sets of $P$ 
    iff the valid partition $\mu \cap W, \mu \cap \overline{W}$ corresponds to a canonical cut of $\cH_S$ that has an edge $g$ from $P$.
    But by~\Cref{lem:associated-terminals}, this happens iff $W$ separates the associated terminal sets of some edge $g$ of $P$.
    Now, let $U$ be any $(u,s)$-mincut.
    By~\Cref{lem:mu-mincuts}, $U$ separates the associated terminal sets of $P$ iff $U \cap \cDdown(\mu)$ is a $\mu$-mincut that separates these sets.
    By the previous discussion, this happens iff $U \cap \cDdown(\mu)$ is a $\mu$-mincut that separates the associated terminal sets of some edge $g$ of $P$, which, by~\Cref{lem:mu-mincuts}, happens iff $U$ separates them.
\end{proof}

We store the $O(|\mu| \log |\mu|)$ data structure of~\Cref{lem:proper-path-information-data-structure} for the skeleton $\cH_\mu$, where the information associated with a proper path is $\RepInf(P)$.
So, when queries with the two endpoints of a proper path, this data structure returns a partition of the path into $O(\log |\mu|)$ subpaths (specified by their endpoints), along with $\RepInf(\cdot)$ for each of these subpaths.
Additionally, we will store the $O(|\mu|)$ data structure of~\Cref{lem:skeleton-data-structure}, which lets us determine the order of nodes on a proper path.
So, over all nodes, this only adds $O(n \log n)$ additional storage, which is negligible as we already use $O(n^{1.5})$ storage.

\paragraph{Restriction to Specific Node.}
Let us focus on a specific node $\nu$ which is an ancestor of $\mu$, and suppose we just wish to find $X_\nu := X \cap \nu$, i.e., all those vertices in $\nu$ that separate $z_x,z_y$, in $\tilde{O}(X_\nu)$ time.
Thus, $\nu$ consists of all those vertices $u \in \nu$ that would have been reported by one the query algorithms for Query~\ref{query-Q1} from~\Cref{sec:splitting-mu}, or for Query~\ref{query-Q2} from~\Cref{sec:not-splitting-mu}.
Denote by $X_\nu^{(1)}$ and $X_\nu^{(2)}$ the respective sets of vertices in $X$, so $X_\nu = X_\nu^{(1)} \cup X_\nu^{(2)}$.
Thus, it suffices to report each $X_\nu^{(i)}$ in time $\tilde{O}(|X_\nu^{(i)}|)$.
From now on, we focus on the more complicated Query~\ref{query-Q1}, as the procedure for~\ref{query-Q2} is very similar (and much simpler).
Let us consider the steps of the query algorithm, and the way these are modified to obtain our goal.

First, Step 0 has already been executed, as we have found $z_x,z_y$ before.
Next, if Step 1 fails, i.e. $\lambda_\mu \neq c(R(\mu))$ or $\phi_\mu (z_x) = \phi_\mu (z_y)$, we know $X_\nu^{(1)} = \emptyset$, and we are done.

We ``simulate'' Step 2 by the following procedure.
We have stored the endpoints of the projection paths $\pi_\mu (z_x)$ and $\pi_\mu (z_y)$.
By a query to the data structure of~\Cref{lem:skeleton-data-structure} for $\cH_\mu$, we get their order on $\pi_\mu (e)$.
This allows us to determine the endpoints of the middle/overlap/prefix/suffix subpaths of $\pi_\mu (e)$.
For each such subpath $P$, we use the data structure from~\Cref{lem:proper-path-information-data-structure} which we have stored for $\cH_\mu$ (with $I(\cdot) = \RepInf(\cdot)$) to get a partition of $P$ to $O(\log n)$ paths $P_1,P_2,\dots$, along with $\RepInf(P_i)$ for each $i$.
We call $P_1,P_2, \dots$ the \emph{parts} of $P$.
We do not yet use their $\RepInf(\cdot)$-s, but ``save'' them for later usage.
So, this modified Step 2 takes us only $O(\log n)$ time.

Step 3 can be executed in exactly the same manner as before, except we do not compute the sets $A_\mu (z_x)$ and $A_\mu (z_y)$, but rather just ``save'' their corresponding representative information, i.e., $\RepInf(A_\mu (z_x)) := \RepInf(\mu_i (z_x) \cup \{x\}, \mu_{1-i} (z_x))$ for $A_\mu (z_x)$ and $\RepInf(A_\mu (z_y)) := \RepInf(\mu_j (z_x) \cup \{x\}, \mu_{1-j} (z_x))$ for the next step, which is the most crucial.

We now arrive at the key Step 4.
To demonstrate the essence, suppose that we just want to report all those vertices from $X^{(1)}_\nu$ which would have been reported in Substep 4c; similar arguments apply to each of the other substeps.
Thus, our goal is to report all those vertices $u \in \nu$ such that:
\begin{enumerate}
    \item $u$ does not separate any edge in the middle/overlap, and

    \item $u$ separates some edge in the prefix or in the suffix, and

    \item $u$ belongs to either $A_\mu (z_x)$ or $A_\mu (z_y)$.
\end{enumerate}

Condition 1 is equivalent (by \Cref{lem:separating-subpaths}) to $u$ not separating any of the $O(\log n)$ subpaths to which we have partitioned the middle/overlap is partitioned. 
So this condition translates to a negation of an AND on separating $O(\log n)$ pairs-of-sets, whose $\RepInf(\cdot)$'s are associated with these subpaths.
In a similar fashion, Condition 2 translates into an OR on separating $O(\log n)$ known $\RepInf(\cdot)$'s.
Condition 3 in just an OR on two $\RepInf(\cdot)$'s.
This means that we can report the set of vertices $u \in \nu$ who satisfy Conditions 1-3 within $O(\text{\# satisfying $u$'s} + \log^2 n)$ time.

As previously mentioned, a similar analysis applies to every other substep of Step 4 for Query~\ref{query-Q1}, and a simpler one applies to Query~\ref{query-Q2}.
We thus get the following:

\begin{corollary}\label{cor:report-separating-inside-a-node}
    There is an $O(n^{1.5})$ space data structure for the following queries:
    Given $e = \{x,y\} \in E$, an LCA $\mu$ of $\cD(x),\cD(y)$, the anchors $z_x,z_y$ of $x,y$ w.r.t.\ $\mu$,
    and given a node $\nu$ which is an ancestor of $\mu$, report the set $X_\nu = \{u \in \nu \mid \text{$u$ separates $z_x,z_y$}\}$.
    The query takes $O(|X_\nu| + \log^2 n)$ time.
\end{corollary}

\paragraph{Finding The Relevant Nodes.}
We are now ready to explain how we find the set $X$ consisting of every $u \in \cDup(\mu)$ which separates $z_x,z_y$ in time $\tilde{O}(X)$.

Call a node $\nu$ which is an ancestor of $\mu$ \emph{relevant} if it contains a vertex from $X$, i.e., if $X_\nu = X \cap \nu \neq \emptyset$.
If we somehow knew all the relevant nodes, we could have just applied the query of~\Cref{cor:report-separating-inside-a-node} on each of them, and return the union of found vertices, which is precisely $X$, and this would take $\tilde{O}(X)$ time.
So, our goal is to find these relevant nodes without spending much time looking into irrelevant ones.

To this end, we execute the following BFS-like procedure, which works by \emph{visiting} nodes $\nu$ of $\cD$, where in such a visit reports the set $X_\nu$ and marks $\nu$ as visited.
We initialize a queue $\mathcal{Q}$ to contain $\mu$, and the following iterations while $\mathcal{Q} \neq \emptyset$:
\begin{itemize}
    \item Dequeue node $\nu$ from the front of $\mathcal{Q}$.
    \item If $\nu$ is not marked as visit it, visit $\nu$ as follows:
    \begin{itemize}
        \item Find and report $U_\nu$ using~\Cref{cor:report-separating-inside-a-node}.
        \item Mark $\nu$ as visited.
        \item If $\nu$ is relevant (i.e., $U_\nu \neq \emptyset$), enqueue its (at most two) parents.
    \end{itemize}
\end{itemize}

To prove correctness, we should show that this procedure indeed visits all the relevant nodes.
Let $\nu$ be a relevant node, and let $(\nu = \nu_0, \nu_1, \dots, \nu_k = \mu)$ be a path from $\nu$ to $\mu$ in $\cD$.
Using the insights behind the representative framework from~\Cref{sec:rep-framework}, we can show that every $\nu_i$ on this path must be relevant.
Indeed, because $\nu = \nu_0$ is useful, i.e., contains a vertex which separates $z_x,z_y$, by~\Cref{lem:rep-nodes} it must be a descendant of some representative node $\eta$ for $(z_x,z_y)$ or for $(z_y,z_x)$.
Thus, every other $\nu_i$ is a \emph{strict} descendant of $\eta$, and as $\mu$ is a descendant of $\nu_i$, we have have $z_x,z_y \in \cDdown(\mu) \subseteq \cDdown(\nu_i)$.
Therefore, by~\Cref{lem:rep-nodes}, every vertex in $\nu_i$ must separate $z_x,z_y$, and in particular $\nu_i$ is relevant.
Now, as the procedure starts by visiting $\mu = \nu_k$, it must find that it is relevant and enqueue $\nu_{k-1}$, and so it must also visit $\nu_{k-1}$, find that it is relevant and enqueue $\nu_{k-2}$, and so on until $\nu_0 = \nu$ is visited. %

Finally, as each visit to a relevant node causes at most two enqueue operations, and visits to irrelevant nodes do not cause any enqueues, the number of visited node is bounded by at most twice the number of relevant nodes, which is at most $X$.
A visit to node $\nu$ takes $O(|X_\nu| + \log^2 n)$ time.
Hence, the total running time is $O(\sum_\nu |X_\nu| + |X| \log^2 n) = O(|X| + |X| \log^2 n) = O(|X| \log^2 n)$.

This concludes the proof of (the failure variant) of~\Cref{thm:faster-queries}(1).

\paragraph{One Destination Queries.}
We now discuss \Cref{item:faster-queries-one-vertex} of~\Cref{thm:faster-queries}, where a query consists of failed edge $e = \{x,y\} \in E$ and destination $u \in V \setminus \{s\}$, and we should determine whether the $u,s$-mincut value decreases after the failure of $e$, or equivalently, whether $u$ separates $x,y$.

The data structure is the same as for \Cref{item:faster-queries-reporting-all}.

To answer the query, we first determine the ancestry relations between $\cD(u)$ and $\cD(x),\cD(y)$, which takes $O(1)$ time using~\Cref{lem:ancestry-lca-data-structure}.
As discussed in~\Cref{sec:DAG-overview}, the only case where this does not suffice is when $\cD(u)$ is an ancestor of both.
Next, using~\Cref{lem:ancestry-lca-data-structure}, we find the (at most two) LCAs of $\cD(x),\cD(y)$ within $O(\log n)$ time.
For each LCA $\mu$, we also get two children of it which are ancestors of $\cD(x)$ and of $\cD(y)$, and hence we know the anchors of $x,y$ w.r.t.\ $\mu$ (by~\Cref{def: anchors}).
By two more ancestry queries, we find one LCA $\mu$ s.t.\ $u \in \cDup(\mu)$, and let $z_x,z_y$ be the anchors of $x,y$ w.r.t.\ $\mu$.
At this point, our goal becomes to determine whether $u$ would have been reported by the query algorithms of Query~\ref{query-Q1} or~\ref{query-Q2}; we focus on the first one, as the latter is similar but much simpler.
We execute the modified steps, from Step 0 up to Step 3, exactly as we did for~\Cref{item:faster-queries-reporting-all}.
So, when we arrive at Step 4, we have the following:
\begin{itemize}
    \item For each subpath $P$ of $\pi_\mu (e)$ among the middle/overlap/prefix/suffix, we have a partition of $P$ into $O(\log n)$ parts $P_1,P_2, \cdots$, along with $\RepInf(P_i)$ for each part.
    \item $\RepInf(A_\mu (z_x)) := \RepInf(\mu_i (z_x) \cup \{z_x\}, \mu_{1-i} (z_x))$ and $\RepInf(A_\mu (z_y)) := \RepInf(\mu_j (z_y) \cup \{z_y\}, \mu_{1-j} (z_y))$.
\end{itemize}
Thus, by $O(\log n)$ calls to the representatives framework structure of~\Cref{lem:rep-framework-one-vertex} (each taking $O(1)$ time), we determine for each $P \in \{\text{middle, overlap, prefix, suffix}\}$ whether $u$ separates $P$ (or equivalently, by~\Cref{lem:separating-subpaths}, separates an edge of $P$), and also determine whether $u \in A_\mu (z_x)$ and whether $u \in A_\mu (z_y)$.
So we can now decide in $O(1)$ time whether to report $u$, by exactly the same logic of Step 4 from~\Cref{sec:splitting-mu}.

This concludes the proof of~\Cref{item:faster-queries-one-vertex} of~\Cref{thm:faster-queries}.

\subsection{Output-Sensitive Queries for All-Pairs Mincuts: Proof of~\Cref{thm:all-pairs-reporting-all}}

We now prove~\Cref{thm:all-pairs-reporting-all}, and show a data structure with $O(n^{2.5})$ space that, given a query of failed edge $e = \{x,y\}$, reports the set $X$ of all vertex pairs $u,w$ such that the $u,w$-mincut value changes after the failure of $e$, within $O(|X| \log^2 n)$ time.
We store, for every $s \in V$ the single-source data structures of~\Cref{thm:faster-queries} 
corresponding to this source $s$, so the space is clearly $O(n^{2.5})$.
The key insight is:

\begin{lemma}\label{lem:reporting-all-pairs}
    Let $e = \{x,y\} \in E$, and let $u,w \in V$.
    If the failure of $e$ causes the $u,w$-mincut value to decrease, then it also causes either the $u,x$-mincut value or the $u,y$-mincut value to decrease.
\end{lemma}
\begin{proof}
    Let $A$ be a $(u,w)$-mincut with $e$ contributing, which exists by the assumption of the lemma.
    Without loss of generality, say $x \in A$ and $y \in \overline{A}$ (otherwise, swap $x$ and $y$ in the following argument).
    Let $B$ be some $(u,y)$-mincut ($u \neq y$ since $u \in A$ and $y \notin A$).
    Observe that $A \cap B$ is a $(u,w)$-cut, and $A \cup B$ is a $(u,y)$-cut.
    Hence, by submodularity (\Cref{lem:sub-posi-general}), $A \cup B$ is a $(u,y)$-mincut.
    Also, $x \in A \cup B$ and $y \in \overline{A \cup B}$, so $e$ contributes to $A \cup B$.
    Therefore, the failure of $e$ causes the $u,y$-mincut value to decrease.
\end{proof}

This leads to the following query algorithm:
    First, using the data structures for $x$ and for $y$ as the sources, find all vertices $u$ such that either the $u,x$-mincut value or the $u,y$-mincut value decreases after the failure of $e$.
    Let $U$ be the set of all such vertices $u$ that were found.
    Then, for every $u \in U$, use the data structure for $u$ as the source to find all vertices $w$ such that the $u,w$-mincut value decreases after the failure of $e$.
    Let $W_u$ be the set of all such vertices $w$ that were found.
    Report $X$ as the set of pairs $u,w$ such that $w \in W_u$.

The correctness follows immediately from~\Cref{lem:reporting-all-pairs}.
As for the running time:
The first step takes us $O(|U| \log^2 n)$ time, and the second $O(\sum_{u \in U} |W_u|\log^2 n)$ time.
Note that, by definition of $U$, for each $u \in U$, $W_u$ contains at least one of $x$ or $y$, hence $|U| \leq \sum_{u \in U} |W_u|$.
Also, this last sum counts every pair from $X$ at most twice.
Thus, the running time of the query algorithm is $O(|X| \log^2 n)$.

\phantomsection
\addcontentsline{toc}{section}{References}
{\small
\bibliographystyle{alphaurl}
\bibliography{references}

\newcommand{\etalchar}[1]{$^{#1}$}
\begin{thebibliography}{GHN{\etalchar{+}}23}

\bibitem[ACP{\etalchar{+}}26]{DBLP:conf/innovations/AhiCPPS26}
Mridul Ahi, Keerti Choudhary, Shlok Pande, Pushpraj, and Lakshay Saggi.
\newblock Maximum-flow and minimum-cut sensitivity oracles for directed graphs.
\newblock In Shubhangi Saraf, editor, {\em 17th Innovations in Theoretical Computer Science Conference, {ITCS} 2026, Bocconi University, Milan, Italy, January 27-30, 2026}, volume 362 of {\em LIPIcs}, pages 5:1--5:24. Schloss Dagstuhl - Leibniz-Zentrum f{\"{u}}r Informatik, 2026.
\newblock URL: \url{https://doi.org/10.4230/LIPIcs.ITCS.2026.5}, \href {https://doi.org/10.4230/LIPICS.ITCS.2026.5} {\path{doi:10.4230/LIPICS.ITCS.2026.5}}.

\bibitem[AKL{\etalchar{+}}22]{DBLP:conf/focs/AbboudK0PST22}
Amir Abboud, Robert Krauthgamer, Jason Li, Debmalya Panigrahi, Thatchaphol Saranurak, and Ohad Trabelsi.
\newblock Breaking the cubic barrier for all-pairs max-flow: Gomory-hu tree in nearly quadratic time.
\newblock In {\em 63rd {IEEE} Annual Symposium on Foundations of Computer Science, {FOCS} 2022, Denver, CO, USA, October 31 - November 3, 2022}, pages 884--895. {IEEE}, 2022.
\newblock \href {https://doi.org/10.1109/FOCS54457.2022.00088} {\path{doi:10.1109/FOCS54457.2022.00088}}.

\bibitem[AKL{\etalchar{+}}25]{DBLP:conf/focs/AbboudKLPGSYY25}
Amir Abboud, Rasmus Kyng, Jason Li, Debmalya Panigrahi, Maximilian~Probst Gutenberg, Thatchaphol Saranurak, Weixuan Yuan, and Wuwei Yuan.
\newblock Deterministic almost-linear-time gomory-hu trees.
\newblock In {\em 66th {IEEE} Annual Symposium on Foundations of Computer Science, {FOCS} 2025, Sydney, Australia, December 14-17, 2025}, pages 659--666. {IEEE}, 2025.
\newblock \href {https://doi.org/10.1109/FOCS63196.2025.00035} {\path{doi:10.1109/FOCS63196.2025.00035}}.

\bibitem[AKT22]{DBLP:conf/soda/AbboudKT22}
Amir Abboud, Robert Krauthgamer, and Ohad Trabelsi.
\newblock Friendly cut sparsifiers and faster gomory-hu trees.
\newblock In Joseph~(Seffi) Naor and Niv Buchbinder, editors, {\em Proceedings of the 2022 {ACM-SIAM} Symposium on Discrete Algorithms, {SODA} 2022, Virtual Conference / Alexandria, VA, USA, January 9 - 12, 2022}, pages 3630--3649. {SIAM}, 2022.
\newblock \href {https://doi.org/10.1137/1.9781611977073.143} {\path{doi:10.1137/1.9781611977073.143}}.

\bibitem[ALPS23]{DBLP:conf/focs/Abboud0PS23}
Amir Abboud, Jason Li, Debmalya Panigrahi, and Thatchaphol Saranurak.
\newblock All-pairs max-flow is no harder than single-pair max-flow: Gomory-hu trees in almost-linear time.
\newblock In {\em 64th {IEEE} Annual Symposium on Foundations of Computer Science, {FOCS} 2023, Santa Cruz, CA, USA, November 6-9, 2023}, pages 2204--2212. {IEEE}, 2023.
\newblock \href {https://doi.org/10.1109/FOCS57990.2023.00137} {\path{doi:10.1109/FOCS57990.2023.00137}}.

\bibitem[BB24]{DBLP:conf/icalp/BaswanaB24}
Surender Baswana and Koustav Bhanja.
\newblock Vital edges for (s, t)-mincut: Efficient algorithms, compact structures, {\&} optimal sensitivity oracles.
\newblock In Karl Bringmann, Martin Grohe, Gabriele Puppis, and Ola Svensson, editors, {\em 51st International Colloquium on Automata, Languages, and Programming, {ICALP} 2024, Tallinn, Estonia, July 8-12, 2024}, volume 297 of {\em LIPIcs}, pages 17:1--17:20. Schloss Dagstuhl - Leibniz-Zentrum f{\"{u}}r Informatik, 2024.
\newblock URL: \url{https://doi.org/10.4230/LIPIcs.ICALP.2024.17}, \href {https://doi.org/10.4230/LIPICS.ICALP.2024.17} {\path{doi:10.4230/LIPICS.ICALP.2024.17}}.

\bibitem[BBP23]{DBLP:journals/talg/BaswanaBP23}
Surender Baswana, Koustav Bhanja, and Abhyuday Pandey.
\newblock Minimum+1 (\emph{s, t})-cuts and dual-edge sensitivity oracle.
\newblock {\em {ACM} Trans. Algorithms}, 19(4):38:1--38:41, 2023.
\newblock \href {https://doi.org/10.1145/3623271} {\path{doi:10.1145/3623271}}.

\bibitem[BBR25]{DBLP:conf/esa/BaswanaBR25}
Surender Baswana, Koustav Bhanja, and Anupam Roy.
\newblock Faster algorithm for second (s, t)-mincut and breaking quadratic barrier for dual edge sensitivity for (s, t)-mincut.
\newblock In Anne Benoit, Haim Kaplan, Sebastian Wild, and Grzegorz Herman, editors, {\em 33rd Annual European Symposium on Algorithms, {ESA} 2025, September 15-17, 2025, Warsaw, Poland}, volume 351 of {\em LIPIcs}, pages 68:1--68:19. Schloss Dagstuhl - Leibniz-Zentrum f{\"{u}}r Informatik, 2025.
\newblock URL: \url{https://doi.org/10.4230/LIPIcs.ESA.2025.68}, \href {https://doi.org/10.4230/LIPICS.ESA.2025.68} {\path{doi:10.4230/LIPICS.ESA.2025.68}}.

\bibitem[BCC{\etalchar{+}}23]{bilo2023approximate}
Davide Bil{\`o}, Shiri Chechik, Keerti Choudhary, Sarel Cohen, Tobias Friedrich, Simon Krogmann, and Martin Schirneck.
\newblock Approximate distance sensitivity oracles in subquadratic space.
\newblock In {\em Proceedings of the 55th Annual ACM Symposium on Theory of Computing}, pages 1396--1409, 2023.

\bibitem[BCHR20]{DBLP:journals/talg/BaswanaCHR20}
Surender Baswana, Keerti Choudhary, Moazzam Hussain, and Liam Roditty.
\newblock Approximate single-source fault tolerant shortest path.
\newblock {\em {ACM} Trans. Algorithms}, 16(4):44:1--44:22, 2020.
\newblock \href {https://doi.org/10.1145/3397532} {\path{doi:10.1145/3397532}}.

\bibitem[Bel14]{Belous2014}
A.~Belous.
\newblock The general structure of edge-connectivity of a vertex subset in a graph: Even case.
\newblock M.sc. thesis, University of Haifa, 2014.
\newblock Supervised by Alexander Vainshtein.

\bibitem[BF00]{BenderF00}
Michael~A. Bender and Martin Farach{-}Colton.
\newblock The {LCA} problem revisited.
\newblock In {\em {LATIN} 2000: Theoretical Informatics, 4th Latin American Symposium, Punta del Este, Uruguay, April 10-14, 2000, Proceedings}, volume 1776 of {\em Lecture Notes in Computer Science}, pages 88--94. Springer, 2000.
\newblock \href {https://doi.org/10.1007/10719839\_9} {\path{doi:10.1007/10719839\_9}}.

\bibitem[BFC00]{BF-C00}
M.~A. Bender and M.~Farach-Colton.
\newblock The {LCA} problem revisited.
\newblock In {\em Proceedings 4th Latin American Symp. on Theoretical Informatics (LATIN), LNCS Vol. 1776}, pages 88--94, 2000.

\bibitem[BFC04]{BenderF04}
M.~A. Bender and M.~Farach-Colton.
\newblock The level ancestor problem simplified.
\newblock {\em Theoretical Computer Science}, 321(1):5--12, 2004.

\bibitem[BGK22]{BaswanaGK22}
Surender Baswana, Shiv~Kumar Gupta, and Till Knollmann.
\newblock Mincut sensitivity data structures for the insertion of an edge.
\newblock {\em Algorithmica}, 84(9):2702--2734, 2022.
\newblock URL: \url{https://doi.org/10.1007/s00453-022-00978-0}, \href {https://doi.org/10.1007/S00453-022-00978-0} {\path{doi:10.1007/S00453-022-00978-0}}.

\bibitem[Bha25]{DBLP:conf/icalp/Bhanja25}
Koustav Bhanja.
\newblock Minimum+1 steiner cut and dual edge sensitivity oracle: Bridging gap between global and (s, t)-cut.
\newblock In Keren Censor{-}Hillel, Fabrizio Grandoni, Jo{\"{e}}l Ouaknine, and Gabriele Puppis, editors, {\em 52nd International Colloquium on Automata, Languages, and Programming, {ICALP} 2025, July 8-11, 2025, Aarhus, Denmark}, volume 334 of {\em LIPIcs}, pages 27:1--27:20. Schloss Dagstuhl - Leibniz-Zentrum f{\"{u}}r Informatik, 2025.
\newblock URL: \url{https://doi.org/10.4230/LIPIcs.ICALP.2025.27}, \href {https://doi.org/10.4230/LIPICS.ICALP.2025.27} {\path{doi:10.4230/LIPICS.ICALP.2025.27}}.

\bibitem[BK96]{BenczurK96}
A.~A. Bencz{\'u}r and D.~R. Karger.
\newblock Approximating $s$-$t$ minimum cuts in $\tilde{O}(n^2)$ time.
\newblock In {\em Proceedings 28th ACM Symposium on Theory of Computing (STOC)}, pages 47--55, 1996.

\bibitem[BP22]{DBLP:conf/soda/BaswanaP22}
Surender Baswana and Abhyuday Pandey.
\newblock Sensitivity oracles for all-pairs mincuts.
\newblock In Joseph~(Seffi) Naor and Niv Buchbinder, editors, {\em Proceedings of the 2022 {ACM-SIAM} Symposium on Discrete Algorithms, {SODA} 2022, Virtual Conference / Alexandria, VA, USA, January 9 - 12, 2022}, pages 581--609. {SIAM}, 2022.
\newblock \href {https://doi.org/10.1137/1.9781611977073.27} {\path{doi:10.1137/1.9781611977073.27}}.

\bibitem[BP25]{BaswanaP25}
Surender Baswana and Abhyuday Pandey.
\newblock The connectivity carcass of a vertex subset in a graph: both odd and even case.
\newblock In Ioana~Oriana Bercea and Rasmus Pagh, editors, {\em 2025 Symposium on Simplicity in Algorithms, {SOSA} 2025, New Orleans, LA, USA, January 13-15, 2025}, pages 385--422. {SIAM}, 2025.
\newblock \href {https://doi.org/10.1137/1.9781611978315.30} {\path{doi:10.1137/1.9781611978315.30}}.

\bibitem[Cho16]{DBLP:conf/icalp/Choudhary16}
Keerti Choudhary.
\newblock An optimal dual fault tolerant reachability oracle.
\newblock In Ioannis Chatzigiannakis, Michael Mitzenmacher, Yuval Rabani, and Davide Sangiorgi, editors, {\em 43rd International Colloquium on Automata, Languages, and Programming, {ICALP} 2016, July 11-15, 2016, Rome, Italy}, volume~55 of {\em LIPIcs}, pages 130:1--130:13. Schloss Dagstuhl - Leibniz-Zentrum f{\"{u}}r Informatik, 2016.
\newblock \href {https://doi.org/10.4230/LIPIcs.ICALP.2016.130} {\path{doi:10.4230/LIPIcs.ICALP.2016.130}}.

\bibitem[Cho25]{ChoudharyFoGA2025}
Keerti Choudhary.
\newblock Fault-tolerant oracles for max-flow and min-cut problems.
\newblock talk at Frontiers of Graph Algorithms Workshop at IISC, India, December 8--12, 2025.
\newblock Website: \url{https://algo.csa.iisc.ac.in/graphworkshop/} video: \url{https://www.youtube.com/watch?v=OviwLyuVHBY}.

\bibitem[DG24]{DBLP:conf/stoc/Dey024}
Dipan Dey and Manoj Gupta.
\newblock Nearly optimal fault tolerant distance oracle.
\newblock In Bojan Mohar, Igor Shinkar, and Ryan O'Donnell, editors, {\em Proceedings of the 56th Annual {ACM} Symposium on Theory of Computing, {STOC} 2024, Vancouver, BC, Canada, June 24-28, 2024}, pages 944--955. {ACM}, 2024.
\newblock \href {https://doi.org/10.1145/3618260.3649697} {\path{doi:10.1145/3618260.3649697}}.

\bibitem[Dil50]{Dilworth50}
R.~P. Dilworth.
\newblock A decomposition theorem for partially ordered sets.
\newblock {\em The Annals of Mathematics}, 51(1):161--166, 1950.

\bibitem[DKL76]{DinicKL76}
E.~A. Dinic, A.~V. Karzanov, and M.~V. Lomonosov.
\newblock On the structure of the system of minimum edge cuts in a graph.
\newblock {\em Studies in Discrete Optimization}, pages 290--306, 1976.
\newblock (in Russian).

\bibitem[DN95]{DinitzN95}
Yefim Dinitz and Zeev Nutov.
\newblock A $2$-level cactus model for the system of minimum and minimum$+1$ edge-cuts in a graph and its incremental maintenance.
\newblock In {\em Proceedings 27th ACM Symposium on Theory of Computing (STOC)}, pages 509--518, 1995.

\bibitem[DV94]{DBLP:conf/stoc/DinitzV94}
Yefim Dinitz and Alek Vainshtein.
\newblock The connectivity carcass of a vertex subset in a graph and its incremental maintenance.
\newblock In Frank~Thomson Leighton and Michael~T. Goodrich, editors, {\em Proceedings of the Twenty-Sixth Annual {ACM} Symposium on Theory of Computing, 23-25 May 1994, Montr{\'{e}}al, Qu{\'{e}}bec, Canada}, pages 716--725. {ACM}, 1994.
\newblock \href {https://doi.org/10.1145/195058.195442} {\path{doi:10.1145/195058.195442}}.

\bibitem[DV95]{DBLP:conf/soda/DinitzV95}
Yefim Dinitz and Alek Vainshtein.
\newblock Locally orientable graphs, cell structures, and a new algorithm for the incremental maintenance of connectivity carcasses.
\newblock In Kenneth~L. Clarkson, editor, {\em Proceedings of the Sixth Annual {ACM-SIAM} Symposium on Discrete Algorithms, 22-24 January 1995. San Francisco, California, {USA}}, pages 302--311. {ACM/SIAM}, 1995.
\newblock URL: \url{http://dl.acm.org/citation.cfm?id=313651.313711}.

\bibitem[DV00]{DBLP:journals/siamcomp/DinitzV00}
Yefim Dinitz and Alek Vainshtein.
\newblock The general structure of edge-connectivity of a vertex subset in a graph and its incremental maintenance. odd case.
\newblock {\em {SIAM} J. Comput.}, 30(3):753--808, 2000.
\newblock \href {https://doi.org/10.1137/S0097539797330045} {\path{doi:10.1137/S0097539797330045}}.

\bibitem[FHHP11]{FungHHP11}
W.~S. Fung, R.~Hariharan, N.~J.~A. Harvey, and D.~Panigrahi.
\newblock A general framework for graph sparsification.
\newblock In {\em Proceedings 43rd ACM Symposium on Theory of Computing (STOC)}, pages 71--80, 2011.

\bibitem[GH61]{GomoryH61}
R.~E. Gomory and T.~C. Hu.
\newblock Multi-terminal network flows.
\newblock {\em Journal of the Society for Industrial and Applied Mathematics}, 9, 1961.

\bibitem[GH23]{DBLP:conf/icalp/GoranciH23}
Gramoz Goranci and Monika Henzinger.
\newblock Efficient data structures for incremental exact and approximate maximum flow.
\newblock In Kousha Etessami, Uriel Feige, and Gabriele Puppis, editors, {\em 50th International Colloquium on Automata, Languages, and Programming, {ICALP} 2023, Paderborn, Germany, July 10-14, 2023}, LIPIcs, pages 69:1--69:14. Schloss Dagstuhl - Leibniz-Zentrum f{\"{u}}r Informatik, 2023.
\newblock URL: \url{https://doi.org/10.4230/LIPIcs.ICALP.2023.69}, \href {https://doi.org/10.4230/LIPICS.ICALP.2023.69} {\path{doi:10.4230/LIPICS.ICALP.2023.69}}.

\bibitem[GHN{\etalchar{+}}23]{DBLP:conf/soda/GoranciHNSTW23}
Gramoz Goranci, Monika Henzinger, Danupon Nanongkai, Thatchaphol Saranurak, Mikkel Thorup, and Christian Wulff{-}Nilsen.
\newblock Fully dynamic exact edge connectivity in sublinear time.
\newblock In Nikhil Bansal and Viswanath Nagarajan, editors, {\em Proceedings of the 2023 {ACM-SIAM} Symposium on Discrete Algorithms, {SODA} 2023, Florence, Italy, January 22-25, 2023}, pages 70--86. {SIAM}, 2023.
\newblock URL: \url{https://doi.org/10.1137/1.9781611977554.ch3}, \href {https://doi.org/10.1137/1.9781611977554.CH3} {\path{doi:10.1137/1.9781611977554.CH3}}.

\bibitem[GHRS25]{DBLP:conf/icalp/GoranciHRS25}
Gramoz Goranci, Monika Henzinger, Harald R{\"{a}}cke, and A.~R. Sricharan.
\newblock Incremental approximate maximum flow via residual graph sparsification.
\newblock In Keren Censor{-}Hillel, Fabrizio Grandoni, Jo{\"{e}}l Ouaknine, and Gabriele Puppis, editors, {\em 52nd International Colloquium on Automata, Languages, and Programming, {ICALP} 2025, Aarhus, Denmark, July 8-11, 2025}, LIPIcs, pages 91:1--91:20. Schloss Dagstuhl - Leibniz-Zentrum f{\"{u}}r Informatik, 2025.
\newblock URL: \url{https://doi.org/10.4230/LIPIcs.ICALP.2025.91}, \href {https://doi.org/10.4230/LIPICS.ICALP.2025.91} {\path{doi:10.4230/LIPICS.ICALP.2025.91}}.

\bibitem[GK17]{DBLP:conf/icalp/GuptaK17}
Manoj Gupta and Shahbaz Khan.
\newblock Multiple source dual fault tolerant {BFS} trees.
\newblock In Ioannis Chatzigiannakis, Piotr Indyk, Fabian Kuhn, and Anca Muscholl, editors, {\em 44th International Colloquium on Automata, Languages, and Programming, {ICALP} 2017, July 10-14, 2017, Warsaw, Poland}, volume~80 of {\em LIPIcs}, pages 127:1--127:15. Schloss Dagstuhl - Leibniz-Zentrum f{\"{u}}r Informatik, 2017.
\newblock \href {https://doi.org/10.4230/LIPIcs.ICALP.2017.127} {\path{doi:10.4230/LIPIcs.ICALP.2017.127}}.

\bibitem[GKLP17]{DBLP:conf/wads/GoldsteinKLP17}
Isaac Goldstein, Tsvi Kopelowitz, Moshe Lewenstein, and Ely Porat.
\newblock Conditional lower bounds for space/time tradeoffs.
\newblock In Faith Ellen, Antonina Kolokolova, and J{\"{o}}rg{-}R{\"{u}}diger Sack, editors, {\em Algorithms and Data Structures - 15th International Symposium, {WADS} 2017, St. John's, NL, Canada, July 31 - August 2, 2017, Proceedings}, volume 10389 of {\em Lecture Notes in Computer Science}, pages 421--436. Springer, 2017.
\newblock \href {https://doi.org/10.1007/978-3-319-62127-2\_36} {\path{doi:10.1007/978-3-319-62127-2\_36}}.

\bibitem[Gup25]{DBLP:conf/focs/Gupta25}
Manoj Gupta.
\newblock Improved 2-approximate shortest paths for close vertex pairs.
\newblock In {\em 66th {IEEE} Annual Symposium on Foundations of Computer Science, {FOCS} 2025, Sydney, Australia, December 14-17, 2025}, pages 1248--1261. {IEEE}, 2025.
\newblock \href {https://doi.org/10.1109/FOCS63196.2025.00065} {\path{doi:10.1109/FOCS63196.2025.00065}}.

\bibitem[GW20]{GrandoniW20J}
Fabrizio Grandoni and Virginia~Vassilevska Williams.
\newblock Faster replacement paths and distance sensitivity oracles.
\newblock {\em {ACM} Trans. Algorithms}, 16(1):15:1--15:25, 2020.
\newblock \href {https://doi.org/10.1145/3365835} {\path{doi:10.1145/3365835}}.

\bibitem[HHS24]{DBLP:conf/soda/HeHS24}
Zhongtian He, Shang{-}En Huang, and Thatchaphol Saranurak.
\newblock Cactus representations in polylogarithmic max-flow via maximal isolating mincuts.
\newblock In David~P. Woodruff, editor, {\em Proceedings of the 2024 {ACM-SIAM} Symposium on Discrete Algorithms, {SODA} 2024, Alexandria, VA, USA, January 7-10, 2024}, pages 1465--1502. {SIAM}, 2024.
\newblock \href {https://doi.org/10.1137/1.9781611977912.60} {\path{doi:10.1137/1.9781611977912.60}}.

\bibitem[HK99]{HenzingerK99}
M.~Henzinger and V.~King.
\newblock Randomized fully dynamic graph algorithms with polylogarithmic time per operation.
\newblock {\em J.~ACM}, 46(4):502--516, 1999.

\bibitem[HT84]{HT84}
D.~Harel and R.~E. Tarjan.
\newblock Fast algorithms for finding nearest common ancestors.
\newblock {\em SIAM J.~Comput.}, 13(2):338--355, 1984.

\bibitem[IKP21]{DBLP:conf/soda/ItalianoKP21}
Giuseppe~F. Italiano, Adam Karczmarz, and Nikos Parotsidis.
\newblock Planar reachability under single vertex or edge failures.
\newblock In D{\'{a}}niel Marx, editor, {\em Proceedings of the 2021 {ACM-SIAM} Symposium on Discrete Algorithms, {SODA} 2021, Virtual Conference, January 10 - 13, 2021}, pages 2739--2758. {SIAM}, 2021.
\newblock \href {https://doi.org/10.1137/1.9781611976465.163} {\path{doi:10.1137/1.9781611976465.163}}.

\bibitem[JST24]{DBLP:conf/soda/JinST24}
Wenyu Jin, Xiaorui Sun, and Mikkel Thorup.
\newblock Fully dynamic min-cut of superconstant size in subpolynomial time.
\newblock In David~P. Woodruff, editor, {\em Proceedings of the 2024 {ACM-SIAM} Symposium on Discrete Algorithms, {SODA} 2024, Alexandria, VA, USA, January 7-10, 2024}, pages 2999--3026. {SIAM}, 2024.
\newblock \href {https://doi.org/10.1137/1.9781611977912.107} {\path{doi:10.1137/1.9781611977912.107}}.

\bibitem[Kar00]{DBLP:journals/jacm/Karger00}
David~R. Karger.
\newblock Minimum cuts in near-linear time.
\newblock {\em J. {ACM}}, 47(1):46--76, 2000.
\newblock \href {https://doi.org/10.1145/331605.331608} {\path{doi:10.1145/331605.331608}}.

\bibitem[KK26]{DBLP:conf/stoc/YKKrauthgamer26}
Yotam Kenneth{-}Mordoch and Robert Krauthgamer.
\newblock Faster all-pairs minimum cut: Bypassing exact max-flow.
\newblock {\em to appear in 58th Annual ACM Symposium on Theory of Computing (STOC)}, 2026.
\newblock URL: \url{https://doi.org/10.48550/arXiv.2511.10036}.

\bibitem[KT19]{DBLP:journals/jacm/KawarabayashiT19}
Ken{-}ichi Kawarabayashi and Mikkel Thorup.
\newblock Deterministic edge connectivity in near-linear time.
\newblock {\em J. {ACM}}, 66(1):4:1--4:50, 2019.
\newblock \href {https://doi.org/10.1145/3274663} {\path{doi:10.1145/3274663}}.

\bibitem[Lov93]{lovasz1979combinatorial}
L{\'{a}}szl{\'{o}} Lov{\'{a}}sz.
\newblock {\em Combinatorial problems and exercises {(2.} ed.)}.
\newblock North-Holland, 1993.

\bibitem[NI92]{NagamochiI92}
Hiroshi Nagamochi and Toshihide Ibaraki.
\newblock A linear-time algorithm for finding a sparse $k$-connected spanning subgraph of a $k$-connected graph.
\newblock {\em Algorithmica}, 7(5{\&}6):583--596, 1992.

\bibitem[P{\u{a}}t11]{DBLP:journals/siamcomp/Patrascu11}
Mihai P{\u{a}}tra{\c{s}}cu.
\newblock Unifying the landscape of cell-probe lower bounds.
\newblock {\em {SIAM} J. Comput.}, 40(3):827--847, 2011.
\newblock \href {https://doi.org/10.1137/09075336X} {\path{doi:10.1137/09075336X}}.

\bibitem[PQ80]{DBLP:journals/mp/PicardQ80}
Jean{-}Claude Picard and Maurice Queyranne.
\newblock On the structure of all minimum cuts in a network and applications.
\newblock {\em In Rayward-Smith V.J. (eds) Combinatorial Optimization II. Mathematical Programming Studies}, 13(1):8--16, 1980.
\newblock \href {https://doi.org/10.1007/BFb0120902} {\path{doi:10.1007/BFb0120902}}.

\bibitem[Tho07]{Thorup07}
M.~Thorup.
\newblock Fully-dynamic min-cut.
\newblock {\em Combinatorica}, 27(1):91--127, 2007.

\end{thebibliography}
}

\appendix

\section{Examples for Key Structures}

\subsection{Connectivity Carcass}\label{sec:carcass-example}

\begin{figure}[h]
  \begin{center}
    \includegraphics[width=\textwidth]{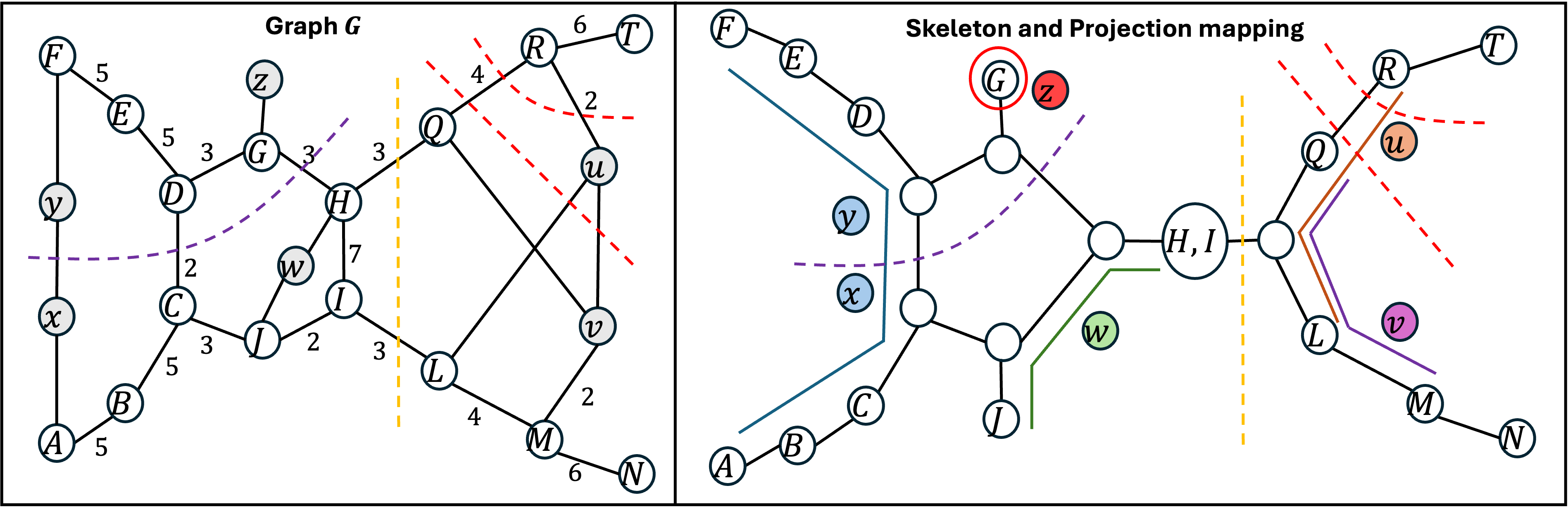}
  \end{center}
  \caption{The Connectivity Carcass}
  \label{fig : skeleton and projection mapping}
\end{figure}

\Cref{fig : skeleton and projection mapping} gives an example of the connectivity carcass, formally introduced in~\Cref{sec:carcass-prelim}.

\begin{itemize}
\item
Numbers next to edges in $G$ represent their multiplicities (no number means multiplicity $1$).
The terminals in $S$ and non-terminals in $V \setminus S$ are denoted by capital and lowercase letters, respectively.
The capacity of $S$-mincuts is $\lambda_S = 6$.

\item
Each terminal vertex is mapped to a node in $\cH_S$, but the mapping is neither injective (two terminals can go in the same node, as terminals $H,I$ in this example) nor surjective (there are empty nodes).

\item
The dashed colored lines show several canonical cuts of $\cH_S$.
Canonical cuts bijectively correspond to valid partitions of $S$, but such a valid partition can be induced by many different $S$-mincuts.
The two red $S$-mincut in $G$ both give rise to the same valid partition $\{R,T\}, S \setminus \{R,T\}$, and hence to the same canonical cut in $\cH_S$ (the tree-edge between the nodes with $Q$ and with $R$).
These two cuts keep the non-terminal $u$ on different sides, so this valid partition (and its canonical cut) \emph{distinguishes} $u$.

\item Projections of several non-terminal vertices are depicted on the right.
Vertex $z$ is non-stretched (i.e., not distinguished by any canonical cut), so $\pi_S (z)$ is a tree node with $G$ of $\cH_S$.
Vertices $u,v,w,x,y$ are stretched, so their projections are proper paths.
The canonical cuts that distinguish a stretched vertex are precisely those that keep the endpoints of the projection path on different sides.

\item The projection of an edge (that contributes to some $S$-mincut) is the unique proper path whose prefix and suffix are the projections of its endpoints, so we easily see $\pi_S (\{u,v\})$ and $\pi_S (\{x,y\})$.

    \begin{itemize}
    \item
    For $\{u,v\}$, the prefix and suffix $\pi_S (u), \pi_S (v)$ are not the same, so we can easily determine the $u$-endpoint and the $v$-endpoint of $\pi_S (\{u,v\})$ as the nodes containing $R$ and $M$ respectively.

    \item
    For $\{x,y\}$, the $x$-endpoint and $y$-endpoints are the node with $A$ and $F$, respectively.
    This can be verified through the unidirectionality property: indeed, $\{x,y\}$ contributes to the purple $S$-mincut, and this cut keeps $x$ on the same side of $A$, and $y$ on the same side as $F$.
    \end{itemize}
Finally, consider again the valid partition $\{R,T\}, S\setminus \{R,T\}$.
Since its corresponding canonical cut is an edge from $\pi_S (\{u,v\})$, it \emph{distinguishes $e$}, which means that \emph{some} $\{R,T\}, S\setminus \{R,T\}$-mincut has $\{u,v\}$ as a contributing edge.
However, there could be $\{R,T\}, S\setminus \{R,T\}$-mincuts where $\{u,v\}$ is not contributing. This is demonstrated by the two red cuts: $\{u,v\}$ contributes to one but not the other.
\end{itemize}

\subsection{Nearest Mincut Tree and Farthest Mincut DAG}\label{sec:tree-DAG-example}

\begin{figure}[h]
  \begin{center}
    \includegraphics[width=\textwidth]{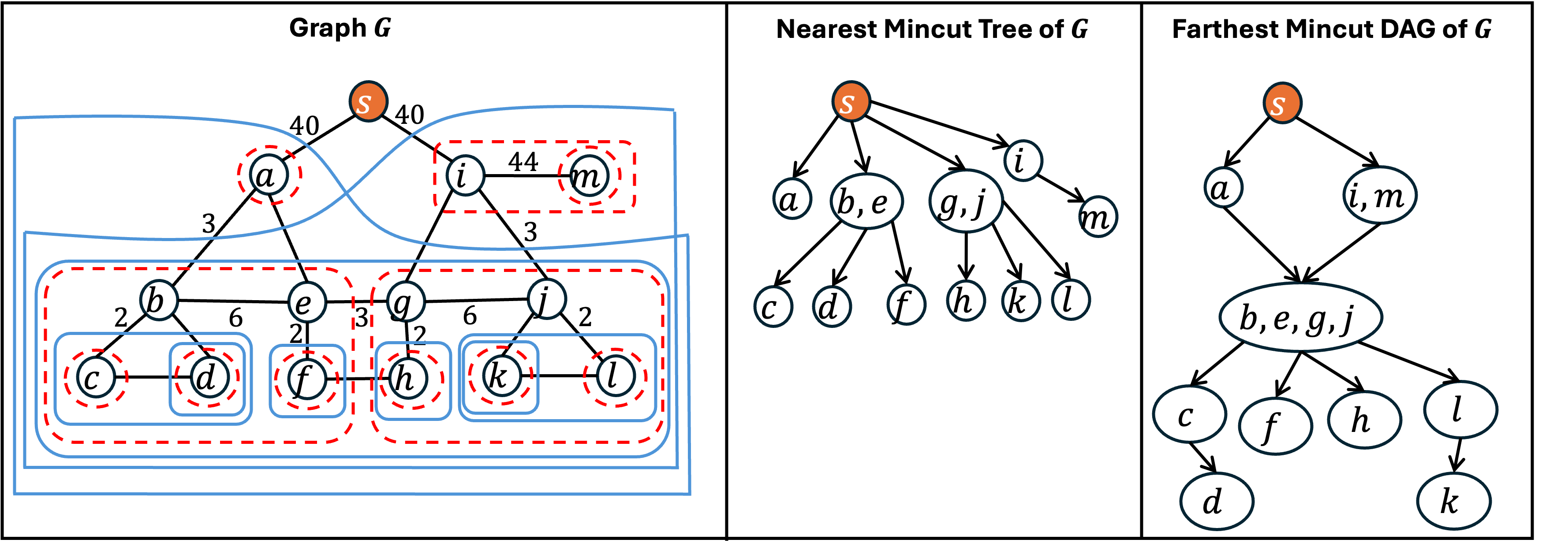}
  \end{center}
  \caption{The Nearest Mincut Tree and Farthest Mincut DAG.}
  \label{fig : tree and dag}
\end{figure}

\Cref{fig : tree and dag} gives an example of nearest mincut tree $\cT$ and the farthest mincut DAG $\cD$, which are presented in~\Cref{sec:DAG-overview}.

\begin{itemize}
\item
Numbers next to edges in $G$ represent their multiplicities (no number means multiplicity $1$).
The source $s$ is marked in orange.

\item
The minimal region defined by a dashed-red curve (resp., solid blue curve) which contains vertex $u$ is $\Near(u,s)$ (resp., $\Far(u,s)$).
The red-bounded regions corresponding to nearest mincuts are laminar: every two are either disjoint, or one contains the other.
The blue-bounded regions corresponding to farthest mincuts are not laminar.

\item The nodes of the nearest mincut tree $\cT$ (resp., farthest mincut DAG $\cD$) partition $V$, where the root node contains only $s$.
For every vertex $u$, its nearest mincut $\Near(u,s)$ (resp., farthest mincut $\Far(u,s)$) consists of all the vertices found in descendants of the node $\cT(u)$ (resp., $\cD(u)$) that contains $u$.

\end{itemize}

\section{Proofs for Preliminaries}\label{sec:missing-proofs}
\SUBANDPOSI
\THREEMINCUTSTOSOURCE

\section{Proofs for Carcass Tools}\label{sec:missing-proofs-carcass}
\CONTAINMENTPROPERPATH
\ASSOCIATEDTERMINALS
\SKELETONTRANSLATION
\UNIONOFCUTSINSKELETON
\SKELETONDATASTRUCTURE
\PROPERPATHINF
\SIGMANUMBERS

\section{Lower Bounds} 
\label{sec:lowerbound}
In this section, we present two information-theoretic lower bounds on space for our data structure in \Cref{thm:main}. Our first lower bound is in undirected graphs for fixed $s,t$-mincut, and the second lower bound is in directed graphs for single-source mincuts.
\paragraph{Undirected graphs with fixed source-sink mincut}
\begin{theorem}
    Let $D$ be any data structure for undirected, unweighted, and simple graphs with a designated source and sink vertices $s,t$, respectively. Given a query with an edge $e$, $D$ can determine whether the value of $s,t$-mincut is reduced after the failure of edge $e$. Then, $D$ must occupy $\Omega(n)$ bits of space in the worst-case, irrespective of the query time. 
\end{theorem}
\begin{proof}
    Let $B\in \{0,1\}^n$. Consider the following undirected, unweighted, and simple graph $H=(V_H,E_H)$ on $|V_H|=2n+2$ vertices. The vertices in $V_H$ are $x_1,x_2,\cdots,x_n, y_1,y_2,\cdots,y_n$, a designated source vertex $s$, and a designated sink vertex $t$. The edge set $E_H=E_M \cup \{\{s,x_i\},\{y_i,t\}~|~i\in [n]\}$, where $E_M=\{\{x_i,y_i\}~|~B[i]=1\}$. It is easy to observe that $|E_M|$ is the maximum number of edge-disjoint paths between $s$ and $t$. Each such path is of the form $\langle s, x_i, y_i, t \rangle$ if $\{x_i,y_i\}\in E_m$. So, the value of $s,t$-mincut is $|E_M|$. Therefore, in any maximum $s,t$-flow, every edge in $E_M$ is fully saturated. By strong duality of maximum flow and mincut, every edge in $E_M$ contributes to the $s,t$-mincut $C=\{s\}\cup \{x_1,x_2,\cdots, x_n\}$. Furthermore, observe that each $x_i$ has degree $2$ with exactly one edge contributing to $C$ iff $\{x_i,y_i\}\in E_M$. If we take $x_i$ out of $C$, then $C\setminus \{x_i\}$ is of the same capacity as the $s,t$-mincut iff $\{x_i,y_i\}\in E_M$. 
    So, edge $\{x_i,y_i\}$ exists iff edge $\{s,x_i\}$ contributes to an $s,t$-mincut. Therefore, the following fact is immediate from the construction of $H$.
    \vspace{-\topsep}
        \begin{center}
            \textit{For any $i\in[n]$, $B[i]=1$ if and only if $s,t$-mincut value reduces after the failure of edge $\{s,x_i\}$.}
        \end{center}
    \vspace{-\topsep}
    It follows from the above fact that any data structure $D$ that can determine whether the value of $s,t$-mincut reduces after failure of any edge given as a query, we can determine using $D$ the value of $B[i]$ for every $i\in [n]$. Therefore, data structure $D$ must occupy at least $\Omega(n)$ bits of space in the worst case.
\end{proof}
\paragraph{Directed graphs with single-source mincuts}
\begin{theorem} \label{thm:directed-graph-lower-bound}
    Let $D$ be any data structure for directed, unweighted, and simple graphs with a designated source vertex $s$ that, given any edge $e$ and a vertex $u$, can determine whether the capacity of $(s,u)$-mincut is reduced after failure of edge $e$. Then, $D$ must occupy $\Omega(n^2)$ bits of space in the worst-case, irrespective of the query time. 
\end{theorem}
\begin{proof}
    Let $B$ be an $n\times n$ binary matrix. Consider the following directed, unweighted, and simple graph $H=(V_H,E_H)$ on $|V_H|=2n+1$ vertices. The vertices in $V_H$ are $x_1,x_2,\cdots,x_n, y_1,y_2,\cdots,y_n$, and a designated source vertex $s$. The edge set $E_H=E_M \cup \{(s,x_i)~|~i\in [n]\}$, where $E_M=\{(x_i,y_j)~|~B[i,j]=1,~ i,j\in [n]\}$. Observe that for each $j\in [n]$, the number of edge-disjoint paths from $s$ to $y_j$ is the same as the indegree of $y_j$. Each such path is of the form $\langle s, x_i, y_j \rangle$, for some $i\in [n]$, if $(x_i,y_j)\in E_M$. So, the $(s,y_j)$-mincut value is the indegree of $y_j$. It follows that, in any maximum $(s,y_j)$-flow, every edge incident on $y_j$ is fully saturated. By the strong duality of maximum flow and mincut, $V_H\setminus \{y_j\}$ is a $(s,y_j)$-mincut, where each incoming edge of $y_j$ contributes. Therefore, the following fact is immediate from the construction of $H$.
    \vspace{-\topsep}
        \begin{center}
            \textit{For any $i,j\in[n]$, $B[i,j]=1$ if and only if $(s,y_j)$-mincut value reduces after the failure of edge $\{s,x_i\}$.}
        \end{center}
    \vspace{-\topsep}
    Let $D$ be any data structure that, given any vertex $u$ and an edge $e$ as a query, can determine whether the value of $(s,u)$-mincut is reduced after the failure of edge $e$. By using the above equivalence, we can use $D$ for graph $H$ to determine the value of $B[i,j]$ for every $i,j\in [n]$. Therefore, data structure $D$ must occupy $\Omega(n^2)$ bits of space in the worst case.
\end{proof}
    The following is an immediate corollary of \Cref{thm:directed-graph-lower-bound}.
    \begin{corollary}
        Let $D$ be any data structure for directed, unweighted, and simple graphs with a designated source vertex $s$ that, given any edge $e$, can report each vertex $u$ such that the  $(s,u)$-mincut value is reduced after failure of edge $e$. Then, $D$ must occupy $\Omega(n^2)$ bits of space in the worst-case, irrespective of the query time. 
    \end{corollary}

\section{Output-Sensitive and One-Destination Queries for Edge Insertion}\label{sec:edge-insertion}

In this section, we prove the insertion variant of~\Cref{thm:faster-queries}.
The data structure (for both items of~\Cref{thm:faster-queries}) stores the nearest mincut tree $\cT$ and the farthest mincut DAG $\cD$ along with data structures that support constant-time ancestry queries for both. 
For $\cT$, we can use any of-the-shelf $O(n)$ space ancestry data structure (e.g.~\cite{BF-C00}). 
For $\cD$, we use the $O(n^{1.5})$ space data structure from~\Cref{lem:ancestry-lca-data-structure}.
Thus, the total space is $O(n^{1.5})$.

We start with the simple query algorithm for~\Cref{item:faster-queries-one-vertex}:
given an inserted edge $e = \{x,y\}$ and vertex $u \in V \setminus \{s\}$, we wish to determine whether the $u,s$-mincut value increases after $e$ is inserted, within $O(1)$ time.
As discussed in~\Cref{sec:DAG-overview}, this happens iff $u \in \cTup(x) \setminus \cDup(y) \cup \cTup(y) \setminus \cDup(x)$.
To detect this situation, we just need to determine the ancestry relations between $\cT(u)$ and $\cT(x),\cT(y)$, and between $\cD(u)$ and $\cD(x),\cD(y)$, which we can do within $O(1)$ time using the ancestry data structures.

We now focus on the query algorithm for~\Cref{item:faster-queries-reporting-all}:
given an inserted edge $e = \{x,y\}$, we wish to report the set $X$ of all vertices $u \in V \setminus \{s\}$ such that the value of $u,s$-mincut increases after $e$ is inserted, within $O(|X|)$ time.
For this, we use ideas introduced by~\cite{BaswanaGK22}.
We first observe the following:

\begin{observation}\label{obs:T-vs-D}
1. The partition of $V$ by the nodes of $\cT$ refines its partition by the nodes of $\cD$.

2. If $\nu,\mu$ are two nodes of $\cT$ such that $\nu$ is an ancestor of $\mu$ in $\cT$, then the node containing $\nu$ in $\cD$ is an ancestor of the node containing $\mu$ in $\cD$.
\end{observation}
\begin{proof}
    1. Suppose $u,w$ are two vertices from the same node of $\cT$.
    This means that $u \in \Near(w,s) \subseteq \Far(w,s)$ and $w \in \Near(u,s) \subseteq \Far(u,s)$, and hence $u \in \cDup(w) \cap \cDdown(w) = \cD(w)$.

    2. Take any $w \in \mu$ and $u \in \nu$, then $w \in \cTdown(u) = \Near(u,s) \subseteq \Far(u,s) = \cDdown(u)$, and thus $\cD(w)$ is a descendant of $\cD(u)$.
\end{proof}

Recall that, by the previous discussion, $X = \cTup(x) \setminus \cDup(y) \cup \cTup(y) \setminus \cDup(x)$.
So, by Observation~\ref{obs:T-vs-D}, we conclude that for any node $\nu$ of $\cT$, either $\nu \cap X$ is all of $\nu$, or it is $\emptyset$; we call $\nu$ relevant if the former holds.
Also, Observation~\ref{obs:T-vs-D} implies that $X_1 = \cTup(x) \setminus \cDup(y)$ and $X_2 = \cTup(y) \setminus \cDup(x)$ are disjoint:
indeed, if $v \in X_1$ then $\cT(v)$ is an ancestor of $\cT(x)$, so $\cD(v)$ is an ancestor of $\cD(x)$, and thus $v \notin X_2$.
We now arrive at the key insight of~\cite{BaswanaGK22}:

\begin{lemma}[\cite{BaswanaGK22}]
    The relevant nodes of $\cT$ form at most two contiguous upwards paths in $\cT$, and each such path starts either at $\cT(x)$ or at $\cT(y)$.
\end{lemma}
\begin{proof}
    Fix a relevant node $\nu$, and say without loss of generality that $\nu \subseteq X_1$.
    Then $\nu$ is an ancestor of $\cT(x)$, and we should show that every node $\mu$ on the path between them is also contained in $X_1$.
    Because $\mu$ is an ancestor of $\cT(x)$, we just need to show that no $w \in \mu$ belongs to $\cDup(y)$.
    Seeking contradiction, suppose there is such $w$.
    Take any $u \in \nu$,
    then because $\nu = \cT(u)$ is an ancestor of $\mu = \cT(w)$, Observation~\ref{obs:T-vs-D} yields that 
    $\cD(u)$ is an ancestor of $\cD(w)$, and hence also of $\cD(y)$, but then $u \in \nu \cap \cDup(y) \subseteq X_1 \cap \cDup(y) = \emptyset$, a contradiction.
\end{proof}

We thus get the following query algorithm: 
Find all relevant nodes in $\cT$ by walking upwards, once from $\cT(x)$ and once from $\cT(y)$;
whenever a node is encountered in the walk, check if it is relevant within $O(1)$ time by using $O(1)$ ancestry queries in $\cT$ and $\cD$, and halt if an irrelevant node is reached.
This procedure visits at most two irrelevant nodes and all the useful nodes, which are at most $O(|X|)$; then, we report the set $X$ of all vertices inside the relevant nodes, which again takes $O(|X|)$ time.

\end{document}